\documentclass[11pt]{amsart}
\usepackage{fullpage}
\usepackage{parskip}
\usepackage{array}
\usepackage[table]{xcolor}
\usepackage{enumitem}
\usepackage{dcolumn}
\usepackage{mdframed}
\usepackage{subcaption}
\usepackage{amsmath, amssymb, amsthm}
\usepackage[foot]{amsaddr}
\usepackage{float}
\newtheorem{thm}{Theorem}[section]
\newtheorem{lem}[thm]{Lemma}
\newtheorem{cor}[thm]{Corollary}
\newtheorem{prop}[thm]{Proposition}

\newtheorem{rem}[thm]{Remark}
\newtheorem{definition}[thm]{Definition}
\numberwithin{thm}{section}
\numberwithin{equation}{section}
\numberwithin{figure}{section}

\makeatletter
\newcommand{\claimlabel}[1]{\def\@currentlabel{#1}\label{claim:#1}}
\makeatother
\newcommand{\claimref}[1]{\hyperref[claim:#1]{Claim~(#1)}}

\DeclareMathOperator{\sech}{sech}
\newcommand{\eps}{\varepsilon}
\newcommand{\gs}{\gamma_\star}
\definecolor{darkgreen}{rgb}{0,0.45,0.08}
\definecolor{rwgray}{RGB}{75,105,140}    
\definecolor{srbrown}{RGB}{160,75,80}    
\definecolor{trblue}{RGB}{160,75,80}     
\definecolor{wellcolor}{RGB}{140,140,140}
\definecolor{panelbgRW}{RGB}{234,234,234}
\definecolor{panelbgNew}{RGB}{220,234,255}
\definecolor{figgreen}{RGB}{50,205,50}
\definecolor{zeroblue}{RGB}{0,65,200}
\definecolor{refgray}{RGB}{130,130,130}
\definecolor{stablebg}{RGB}{232,245,233}
\definecolor{unstablebg}{RGB}{252,235,233}
\definecolor{headerbg}{RGB}{242,242,242}
\newcommand{\R}{\mathbb{R}}
\newcommand{\C}{\mathbb{C}}
\newcommand{\la}{\lambda}
\newcommand{\Scal}{\mathcal{S}}
\newcommand{\Lop}{L_+^\eps}

\newcommand{\win}{w_{\mathrm{in}}}
\newcommand{\wout}{w_{\mathrm{out}}}
\newcommand{\woutL}{w_{\mathrm{out}}^{L}}
\newcommand{\woutR}{w_{\mathrm{out}}^{R}}
\newcommand{\Ub}{U_\star(b)}

\newcommand{\Hd}{H^1_{0}(\mathbb{R}_+)}
\newcommand{\Hf}{H^1(\mathbb{R})}
\newcommand{\sgn}{\operatorname{sgn}}
\newcommand{\supp}{\operatorname{supp}}
\newcommand{\spn}{\operatorname{span}}

\newcommand{\bifwell}[4]{%
  \draw[wellcolor, line width=1pt, shift={(#1,#2)}]
    ({-#3},0)--({-#3},{-#4})--(#3,{-#4})--(#3,0);}
\newsavebox{\warmupZeroXbox}

\newcommand{\Nslope}{\nu_\star}
\newcommand{\NslopeT}{\nu_\star^{\rm TR}}
\newcommand{\NslopeS}{\nu_\star^{\rm SR}}

\newcommand{\Sreg}{\mathsf S}
\newcommand{\Ureg}{\mathsf U}

\definecolor{markergray}{rgb}{0.35,0.35,0.35}
\definecolor{markerred}{rgb}{0.45,0.05,0.05}
\definecolor{markerteal}{rgb}{0.00,0.28,0.38}

\definecolor{plotred}{rgb}{0.85,0.10,0.10}
\definecolor{plotorange}{rgb}{0.95,0.55,0.10}

\DeclareRobustCommand{\specpole}{\tikz[baseline=-0.55ex]{%
  \filldraw[figgreen, draw=black!75, line width=0.5pt] (0,0) circle (2.2pt);}}
\DeclareRobustCommand{\speczero}{\tikz[baseline=-0.55ex]{%
  \draw[zeroblue, line width=1.4pt, line cap=round]
    (-0.65ex,-0.65ex)--(0.65ex,0.65ex) (-0.65ex,0.65ex)--(0.65ex,-0.65ex);}}

\usepackage[colorlinks=true, linkcolor=red, citecolor=blue, urlcolor=cyan]{hyperref}

\usepackage{multirow, multicol, tcolorbox}
\usepackage[mathic]{mathtools}
\tcbuselibrary{skins, theorems, breakable}
\usepackage{graphicx}
\usepackage{tikz}
\usepackage{pgfplots}
\usetikzlibrary{arrows.meta, patterns}
\usepgfplotslibrary{groupplots}
\pgfplotsset{compat=newest}
\pgfmathdeclarefunction{bifwarmSech}{1}{%
  \pgfmathparse{2/(exp(max(min(#1,8),-8))+exp(-max(min(#1,8),-8)))}%
}
\pgfmathdeclarefunction{bifwarmTanh}{1}{%
  \pgfmathparse{%
    (exp(max(min(#1,8),-8))-exp(-max(min(#1,8),-8)))/%
    (exp(max(min(#1,8),-8))+exp(-max(min(#1,8),-8)))}%
}
\AtBeginDocument{%
  \savebox{\warmupZeroXbox}{\begin{tikzpicture}[baseline=-0.55ex]
    \draw[zeroblue, line width=1.5pt, line cap=round]
      (-3.2pt,-3.2pt)--(3.2pt,3.2pt) (-3.2pt,3.2pt)--(3.2pt,-3.2pt);
  \end{tikzpicture}}%
}
\usepackage[font=footnotesize, labelfont=bf, skip=4pt]{caption}
\usepackage[section]{placeins}  

\usepackage[backend=biber, date=year, giveninits=true, sorting=nyt, natbib=true, maxcitenames=2, maxbibnames=10, url=false, doi=true, backref=false]{biblatex}
\renewbibmacro{in:}{\ifentrytype{article}{}{\printtext{\bibstring{in}\intitlepunct}}}

\title{Stability and instability analysis of resonance-induced nonlinear bound states}

\author{Jackson C.\ Turner\textsuperscript{1} \and Michael I.\ Weinstein\textsuperscript{1,2}}

\thanks{\textsuperscript{1}Department of Applied Physics and Applied Mathematics, Columbia University, New York, NY}
\thanks{\textsuperscript{2}Department of Mathematics, Columbia University, New York, NY}

\email{jt3287@columbia.edu}
\email{miw2013@columbia.edu}

\date{\today}

\begin{document}

\begin{abstract}
We study  the focusing one-dimensional cubic nonlinear Schr\"odinger / Gross--Pitaevskii equation (NLS-GP), where the potential of the underlying linear Schr\"odinger operator, $H_V=-\partial_x^2+V(x)$, is compactly supported.  In \cite{turner2026resonance}  the authors proved that purely imaginary zeros (transmission resonances) or purely imaginary poles in the lower half plane  (scattering resonances) of a reflection coefficient of $H_V$ seed branches of {\it resonance-induced nonlinear bound states}, and that
these states bifurcate at a strictly positive $L^2(\mathbb R)$ excitation threshold. This is in contrast to nonlinear bound states which bifurcate at zero $L^2$-norm from point spectra of $H_V$ (corresponding to poles in the upper half plane). In this paper we establish precise criteria for the  nonlinear orbital stability and instability of the resonance-induced states near the bifurcation point. A corollary is that resonance-induced nonlinear states are orbitally stable  if i) the underlying linear (scattering or transmission) resonance frequency is sufficiently small,  and ii) the corresponding linear resonance eigenmode  is strictly positive.
Numerical simulations are presented to explore regimes not accessible to our theory, and to explore the large time dynamics for initial conditions near stable and unstable resonance-induced nonlinear bound states.
\end{abstract}

\maketitle

\section{Introduction }\label{sec:intro}

We study the focusing 1D cubic nonlinear Schr\"odinger / Gross--Pitaevskii equation (NLS/GP):
\begin{equation}\label{eq:nlse}
  i\partial_t\Psi = -\partial_x^2\Psi + V(x)\Psi - |\Psi|^2\Psi\ ,
\end{equation}
where the linear potential $V\in L^1$ is  real-valued and compactly supported. For initial data $\Psi_0 \in H^1(\R)$, \eqref{eq:nlse} has a unique global solution in $C(\R_t; H^1(\R))$; see \cite{kato1987, cazenave2003semilinear,fibich2015nonlinear}.

A nonlinear bound state of NLS/GP is a time-harmonic solution
$\Psi(x,t) = e^{-iEt}\psi(x)$, where the spatial profile $\psi(x)$ satisfies:
\begin{equation}\label{eq:V-nls}
  H_V\psi - |\psi|^2\psi = E\psi,\quad \psi\in H^1(\mathbb R),\quad\textrm{where $H_V:=-\partial_x^2+V(x)$}\ . 
\end{equation}
For $V = 0$, \eqref{eq:V-nls} reduces to $-\psi''-|\psi|^2\psi=E\psi$ with $\psi\in H^1$. In this case, all non-trivial solutions are given, up to spatial and phase translates, by the {\it free soliton}:
\begin{equation}\label{eq:free-soliton-intro}
 \psi(\xi) =  \Scal(\xi;E) := \sqrt{2|E|}\,\sech(\sqrt{|E|}\,\xi)\ ,
\end{equation}
where $E<0$.
A nonlinear orbital Lyapunov stability theory  of nonlinear bound states is discussed in Definition~\ref{sec:orbital}; see also, for example,  \cite{weinstein1986lyapunov,fibich2015nonlinear,sulemsulem1999,cazenave2003semilinear}.

For the case where $V(x)$ is a potential well with discrete spectrum, stable nonlinear bound states were shown in \cite{rose1988bound} to bifurcate from the zero solution at $L^2(\mathbb R)$-eigenpairs of $H_V=-\partial_x^2 + V(x)$. Their role in long-time dynamics, including asymptotic stability of nonlinear ground states and radiation damping, is discussed further in Section~\ref{sec:discussion}.



Recently, in \cite{turner2026resonance}, the authors demonstrated the bifurcation of nonlinear bound states from certain \emph{scattering} (SR) and \emph{transmission resonances} (TR) of the underlying linear operator $H_V$; see Theorem~\ref{thm:bifurcation} below.  We refer to such nonlinear bound states as {\it resonance-induced nonlinear bound states}. 

Scattering and transmission resonances are nontrivial solutions of $H_V\psi=E\psi$ subject to (i) outgoing radiation conditions (SR) or (ii) transmission conditions (TR). The corresponding resonance modes are $L^2_{\rm loc}$  but spatially unbounded. Resonance-induced nonlinear states bifurcate at a strictly positive {\it excitation threshold} in the $L^2(\mathbb R)$ norm (optical power \cite{fibich2015nonlinear,sulemsulem1999} or particle number \cite{erdHos2010derivation}). %
Let
 \[\mathcal N[\psi] := \| \psi \|_{L^2}^2\quad .\] 
 It is common to visualize families of nonlinear bound states as bifurcation curves in the $(\mathcal{N},E)$-plane. See, for example, the right column of Figure~\ref{fig:warmup-bifurcation}, which contrasts  SR and TR bifurcations with $L^2(\mathbb R)$-bound-state bifurcations studied in~\cite{rose1988bound}.

{\it In this paper we study the dynamical stability and instability properties of such resonance (SR- and TR-) induced  nonlinear bound states.} 
The main results are summarized just below in Section \ref{ssec:main-results} and stated in detail in Theorem~\ref{thm:intro-stability-sharp}, Corollary \ref{cor:stability-TR} (TR) and Corollary~\ref{cor:stability-SR} (SR).

\subsection{A brief review of resonance-induced nonlinear bound states; see \cite[Theorem~4.1]{turner2026resonance}}
Consider the following two (SR and TR) resonance spectral problems for the operator $H_V$: 
\begin{align}
H_V U_\star & =E_\star U_\star\ ,\ \textrm{where}\  
E_\star\ =\ -\gamma_\star^2<0\quad\textrm{and}\quad U_\star\ne0, \label{eq:res-prob}\end{align}
%
%
corresponding to the following asymptotic behaviors at infinity:
\begin{subequations}\label{eq:SR-TR-cases}
\begin{align}
&\textrm{Scattering resonances (SR):\ $0<x<\infty$,\quad  $U_\star(0)=0$\ and\ $U_\star$\ ``outgoing'' as $x\to+\infty$}\label{eq:SR-half}\\
&\textrm{Transmission resonances (TR): \  $x\in\mathbb R$,\ $U_\star$\ ``incoming'' at $-\infty$ and ``outgoing'' as $x\to+\infty$}.
\end{align}
\end{subequations}

We review the basics of 1D scattering theory, reflection and transmission coefficients, etc.\ in Section~\ref{ssec:resonance-def}.
In the (SR) case, resonances are poles  of the reflection coefficients; we focus in particular on poles of the form $-i\gamma_\star$, with $\gamma_\star>0$ (the $\gs <0$ states are the bound states covered in \cite{rose1988bound}). In the (TR) case, resonances are 
zeros  of the reflection coefficients of the form $-i\gamma_\star$, with $\gamma_\star \in\mathbb R \setminus \{0\}$.

 Although resonance modes $U_\star$ are not $L^2(\R)$ (indeed $U_\star$ grows exponentially), our resonance-induced nonlinear bound states are  solutions of 
    \begin{equation} H_V\psi -|\psi|^2\psi = E\psi,\quad \psi\in X, \label{eq:X-nl-evp}
    \end{equation}
    where
    \begin{equation}\label{eq:Xspace}
    X = \begin{cases} \Hd & \text{(SR) case },\\ \Hf & \text{(TR) case}.\end{cases}
    \end{equation}

\begin{rem} If $V(x)=V(-x)$, then in the (SR) case we may equivalently take $X=H^1_{\rm odd}(\mathbb R)$, since any resonance in $H^1_0(\mathbb R_+)$ can be extended to a function in $H^1(\mathbb R)$ which is odd.\end{rem}

We next recall our bifurcation result  \cite[Theorem~4.1]{turner2026resonance}; see also Theorem~\ref{thm:bifurcation} below. Introduce the following quantity, which depends on the resonance eigenpair $(U_\star,-\gamma_\star^2)$:
\begin{equation}\label{eq:Bstar}
     \mathcal B_\star
      :=\begin{cases}
      \int_{-b}^b U_\star^2\,dx \;-\; \frac{U_\star(b)^2}{2\gs}+\frac{1}{2\gs}, &\ \textrm{(TR) case / $X=H^1(\mathbb R)$}\\
      & \\
      \int_0^b U_\star^2\,dx \;-\; \frac{U_\star(b)^2}{2\gs}, &\ \textrm{(SR) case / $X=H_0^1(\mathbb R_+)$.}
      \end{cases}
\end{equation}
When referring to $\mathcal{B}_\star$  specifically in the (TR) or (SR) case, we shall write: $\mathcal{B}^{(TR)}_\star$ or  $\mathcal{B}^{(SR)}_\star$.

\begin{thm}[Resonance-induced nonlinear bound states; Theorem 4.1 of \cite{turner2026resonance}]
\label{thm:intro-bifurcation}
 Assume the non-degeneracy condition: $\mathcal B_\star\ne0$.

 Then, the following holds:
\begin{enumerate} 
\item   Each (SR) and (TR) resonance eigenpair $(\gs, U_\star)$ seeds a curve of nonlinear bound state solutions to \eqref{eq:V-nls}:
\[
\eps \longmapsto(E(\eps),\psi_\eps)\in\R\times X,
\]
which bifurcates at energy $E_\star=-\gamma_\star^2$: \[ E(\eps)=E_\star+\mathcal O(\eps), \quad \textrm{as $\eps\to0$}.\]
\item 
These bifurcations occur  at a strictly positive threshold $L^2$ norm
;  as $\eps\to0$, 
\begin{align*}
     \|\psi_\eps\|_{L^2(\R_+)}^2 &\to \mathcal N_\star=4\gs\qquad   \textrm{for the SR case}\\
     \|\psi_\eps\|_{L^2(\R)}^2 &\to \mathcal N_\star=4\gs\qquad \textrm{for the TR case.}
\end{align*}
(Viewed on all $\R$, the limiting scattering resonance (SR) state consists of a soliton  near $+\infty$ and one near $-\infty$. Hence, the $L^2(\R)$-threshold is at $2\times (4\gs)=8\gs$; see \cite{turner2026resonance}.)

\item For $0<\eps\ll1$, the spatial structure of the TR- induced nonlinear bound state  is \begin{equation}
  \psi_\eps(x) \approx
  \begin{cases}
      \sqrt\eps\ U_\star(x), & x\in {\rm supp(V)}\\
      \mathcal{S}(x-y_R(\eps);-\gs^2) &   x\notin {\rm supp}(V)\ \textrm{and}\ x\gg1 , 
  \end{cases}
\end{equation}
with $y_R(\eps)\sim \log(1/\eps)$.
\end{enumerate}
\end{thm}
%
An analogous statement applies for SR-induced nonlinear bound states; see Theorem \ref{thm:bifurcation}, where an expanded form of Theorem \ref{thm:intro-bifurcation} is stated for use in proofs.

Figure~\ref{fig:warmup-bifurcation} displays the different bifurcation scenarios; Top right panel: bifurcation from point spectrum, corresponding to a purely imaginary pole of $r_\pm(k)$ in the upper half plane, Middle right panel: bifurcation from a scattering resonance, corresponding to a purely imaginary pole of $r_\pm(k)$ in the lower half plane, and 
 Bottom right panel: bifurcation from a transmission resonance, corresponding to a purely imaginary zero of $r_\pm(k)$ (either upper or lower half plane).

    \subsection{Summary of results}\label{ssec:main-results}

Our main results concern the nonlinear dynamical stability of resonance-induced nonlinear bound states. We first review an appropriate notion of dynamical stability.

\begin{definition}[Orbital stability]
\label{sec:orbital}
Let $e^{-iEt}\psi( \, x \, ;E)$, where $x\mapsto\psi(x;E)\in X$, denote a nonlinear bound state spatial profile of \eqref{eq:X-nl-evp}. We say that this state is \emph{orbitally stable in $X$} if the following holds:  For every $\delta > 0$, there exists $\eta=\eta(\delta) > 0$, such that if $\|\Psi_0 - \psi(\cdot;E)\|_{X} < \eta$, then the solution $\Psi(t)$ of~\eqref{eq:nlse} satisfies
\begin{equation}
  \inf_{\theta \in \R}\|\Psi(t) - e^{i\theta}\psi(\cdot;E)\|_{X} < \delta \qquad \text{for all } t \ge 0. 
\end{equation}
\end{definition}

\noindent We next formulate our main theorem on the stability and instability properties of resonance-induced nonlinear bound states, near the bifurcation point ($\eps$ small); see Theorem \ref{thm:intro-bifurcation}. \medskip

Let $(U_\star,-\gs^2)$ denote a scattering or transmission resonance eigenpair for which the non-degeneracy condition $\mathcal{B}_\star\ne0$ holds. Formulating our stability theorem requires us to introduce additional $(U_\star,-\gs^2)$ and $\mathcal{B}_\star$ dependent quantities:  Let 
\begin{align}\label{eq:Omega}
  \Omega_\star \;:=\; 
  \begin{cases}
     \;-\; 2\!\int_{-b}^b U_\star^4\,dx + \frac{ U_\star(b)^4 }{2\gs} + \frac{1}{4\gs}, &\ \textrm{TR case}\\
     &\\
     \;-\; 2\!\int_0^b U_\star^4\,dx +  \frac{U_\star(b)^4}{2\gs} ,&\ \textrm{SR case\ .}
  \end{cases} 
\end{align}
In our review of the Lyapunov stability theory of NLS/GP  (Section \ref{sec:stability-theory}) we shall see the important role played by the slope, in the $(E,\mathcal N)$-plane, of the bifurcating branch $\eps\mapsto\bigl(E(\eps),\mathcal N[\psi(\cdot;E(\eps))]\bigr)$ at the bifurcation point $\eps=0$. We have, by Proposition \ref{prop:slope} below that
\begin{align}\label{eq:nu}
\Nslope \; :=\;  \frac{d\mathcal N}{dE}\Big|_{\eps=0^+}\  =\
  \begin{cases}
      \frac{2(\mathcal B_\star^{\rm TR})^2}{\Omega_\star^{\rm TR} - 3/(4\gs)} \;-\; \frac{2}{\gs}\  &\ \textrm{(TR case)}\\
     &\\
     \frac{2(\mathcal B_\star^{\rm SR})^2}{\Omega_\star^{\rm SR}} \;-\; \frac{2}{\gs}\ ,&\ \textrm{(SR case).}
  \end{cases} 
\end{align}
We assume that this slope is finite;
\begin{equation}\label{eq:non-vert}
\Omega_\star^{\rm TR} \ne \frac{3}{4\gs}\quad \textrm{(TR case)} \qquad{\rm and}\qquad  
\Omega_\star^{\rm SR} \ne 0\quad 
\textrm{(SR case)}.
\end{equation}

For solutions sufficiently near to and emanating from the bifurcation point at $(E_\star,\mathcal N_\star)$, stability or instability is determined by our main theorem:
{
\begin{thm}[Sharp stability criterion]\label{thm:intro-stability-sharp}
For the SR case, we set $[a,b]=[0,b]$ and for the TR case we set $[a,b]=[-b,b]$.

Let $(\gs,U_\star)$, with $E_\star=-\gamma_\star^2$, be a scattering or transmission resonance pair determined by the boundary value problem on $[a,b]$; equations \eqref{eq:resonance-eq}--\eqref{eq:SR-BC} for the SR case and equations \eqref{eq:resonance-eq}--\eqref{eq:TR-BC} for the TR case. Assume:
  \begin{enumerate}[leftmargin=2em,label=\textup{(H\arabic*)}]
    \item\label{H:positivity} \emph{Positivity:} $U_\star(x)\geq 0$ on $[a,b]$,
    \item\label{H:Bstar} \emph{Non-degeneracy:} $\mathcal B_\star\neq 0$,
    \item\label{H:Omega} \emph{Non-vertical tangency:} $\Nslope:=\dfrac{d\mathcal N}{dE}\Big|_{\eps=0^+}$ is finite; see \eqref{eq:non-vert}.
  \end{enumerate}
For all $\eps>0$ sufficiently small, let $\eps\mapsto \big(E(\eps),\psi_\eps(x;E(\eps))\big)\in \R\times X$ denote the resonance-induced nonlinear bound state obtained from Theorem~\ref{thm:intro-bifurcation}. Then:
  \begin{enumerate}
    \item the branch $\eps\mapsto \big(E(\eps),\psi_\eps(\cdot;E(\eps))\big)$ is orbitally stable in $X$ if
    \begin{equation}\label{eq:Omega-criterion}
      \Omega_\star > 0 \qquad \text{and} \qquad \Nslope < 0,
    \end{equation}
    \item the branch is unstable if
    \[ \textrm{either\qquad $\Omega_\star < 0$\qquad or\qquad $\Nslope > 0$.}\]
  \end{enumerate}
\end{thm}
\begin{rem}[Stability for $|\gs|$ small]\label{rem:gs-small}
  Via asymptotic analysis and estimates of terms comprising  the expressions for $\Omega_\star$ in \eqref{eq:Omega} and $\nu_\star$ in \eqref{eq:nu}, it can be proved that
$\Omega_\star>0$ and $\nu_\star<0$ provided $|\gs|$ is sufficiently small. Corollary ~\ref{cor:small-gamma-stability} states that if $U_\star>0$ and $|\gs|$ is small then, there exists $\eps_0>0$ such that for all $0<\eps<\eps_0$,  the resonance-induced nonlinear bound states $\eps\mapsto \big(E(\eps),\psi_\eps(\cdot;E(\eps))\big)\in \R\times X$ of Theorem \ref{thm:intro-stability-sharp} are Lyapunov stable in $X$.
\end{rem}
}




\begin{rem}[On the proof of Theorem \ref{thm:intro-stability-sharp}]\label{rem:on_the_pf} A stability theory, summarized in Theorem \ref{thm:W-GSS} below, provides criteria that ensure that the given nonlinear bound state is a local minimum of the Hamiltonian energy subject to fixed $\mathcal{N}$ ($\mathcal{N}[\psi]=\|\psi\|_{L^2}^2$). Stability is controlled by a constrained variational problem for a pair of quadratic forms: $u\mapsto \left\langle L_+u,u\right\rangle$ 
 and $v\mapsto \left\langle L_-v,v\right\rangle$,  where $u+iv\in X$. Very succinctly, condition (i) on the positivity of the resonance eigenstate of $H_V$ (an extension of the property  of positivity of the ground state in the soliton theory) implies positivity of $\left\langle L_-v,v\right\rangle$, and conditions (ii) and (iii) are used to obtain the positivity of $\left\langle L_+u,u\right\rangle$. The most delicate point arises from NLS/GP not being invariant under spatial translation. This induces the zero energy eigenvalue of NLS ($V=0$) to perturb away from zero. Proposition~\ref{prop:mu} gives, for  $\eps>0$ and small, $\mu(\eps)=\tfrac{3\eps^2}{4\gs}\Omega_\star+\mathcal O(\eps^3)$.  Hence, the condition (ii) $\Omega_\star>0$ ensures that $\mu(\eps)$ does not contribute a negative direction for the quadratic form
  $\left\langle L_+u,u\right\rangle$, which would be a direction of instability. Our proof uses that $L_\pm$ are Schr\"odinger operators with P\"oschl--Teller type potentials for which there is an explicit spectral theory. 
  \end{rem}
  \bigskip\bigskip

\medskip

\subsection{Outline of the paper }\label{ssec:organization}

Section~\ref{sec:scattering} provides background: a quick review of 1D linear scattering from a potential well, and the special case of P\"oschl--Teller potentials. We also state a version of our result in  Theorem \ref{thm:intro-bifurcation} (see \cite{turner2026resonance}) on resonance-induced nonlinear bound state bifurcation, which we apply in our stability analysis. Section~\ref{sec:stability-theory} reviews the orbital stability theory, and how the analysis of the appropriate $L_\pm$ operators plays out for NLS/GP with $V=0$ (stability of the 1-soliton) and NLS/GP, where the bifurcation is from a simple $L^2$-eigenvalue of $H_V$. Section \ref{sec:prelim-rmk} contains preliminary remarks on the stability analysis of  resonance-induced nonlinear bound states.  Section~\ref{sec:main-thm} proves our main result,  Theorem~\ref{thm:intro-stability-sharp}, and derives Corollary~\ref{cor:small-gamma-stability}, and explains how the proofs are  reduced to Propositions~\ref{prop:mu},~\ref{prop:slope},~\ref{prop:Lp-count}. The proof is given in Section~\ref{sec:proofs}; the transmission resonance case (TR) is presented with all details, and the required modifications for the scattering resonance case (SR) are discussed in Remark~\ref{rem:SR-stability}.
In Section~\ref{ssec:stab-regimes} we present a  geometric interpretation of our stability and instability criteria (Corollaries~\ref{cor:stability-TR} and~\ref{cor:stability-SR})  along with illustrative numerical examples. In Section \ref{ssec:mu-verification} we  provide a numerical corroboration of the key perturbation expression (Proposition~\ref{prop:mu}) for the translation eigenvalue $\mu(\eps)$, and a construction (Proposition~\ref{prop:threshold-construction}) of stable TR branches for small $\gs > 0$ (Corollary~\ref{cor:small-gamma-stability}).  In Section~\ref{sec:dynamics} we present two cases of the long-time simulations for the time-dependent  NLS-GP initial value problem with initial data in a neighborhood of an unstable resonance-induced nonlinear bound state.  Section~\ref{sec:discussion} discusses extensions and open questions.

\subsection*{Acknowledgments} 
MIW acknowledges the very stimulating collaboration with Harvey Rose \cite{rose1988bound} during the 1980s, where the questions explored in \cite{turner2026resonance} and the current work originate.  The authors thank Henri Berestycki,  Panayotis Kevrekidis, Eduard Kirr, Mikael Rechtsman and the research group of Sebastian Will for stimulating discussions. 
This work was supported in part by NSF grants: DMS-1908657, DMS-1937254 and DMS-2510769, and Simons Foundation Math+X Investigator Award \#376319 (MIW). Part of this research was carried out during the 2023-24 academic year, when MIW was a Visiting Member in the School of Mathematics, Institute of Advanced Study, Princeton, supported by the Charles Simonyi Endowment, and a Visiting Fellow in the Department of Mathematics at Princeton University.

\section{Linear scattering theory and resonance induced nonlinear bound states}\label{sec:scattering}

This section reviews scattering and transmission resonances, summarizes the bifurcation results from \cite{turner2026resonance}, and provides background on the spectral theory of Schr\"odinger operators with potentials of  P\"oschl--Teller type. Figure~\ref{fig:warmup-bifurcation} provides a schematic  of the three bifurcation scenarios for nonlinear bound states of NLS/GP.

\begin{figure}[p]
\centering
\makebox[0.47\linewidth][c]{\textbf{Linear} ($k$-plane)}%
\hfill%
\makebox[0.50\linewidth][c]{\textbf{Nonlinear} $(\mathcal{N},E)$-plane}\\[2pt]

\begin{minipage}[t]{0.47\linewidth}\centering
\vspace{0pt}%
\begin{tikzpicture}[scale=0.65,transform shape]
  \fill[panelbgRW, rounded corners=4pt] (-3.7,-3.7) rectangle (3.7,3.7);
  \node[font=\normalsize\bfseries, anchor=south, rotate=90] at (-3.55,0) {(a) Bound state};
  \draw[->] (-3.5,0)--(3.5,0) node[below,font=\footnotesize]{Re$(k)$};
  \draw[->] (0,-3.5)--(0,3.5) node[right,font=\footnotesize]{Im$(k)$};
  \foreach \x/\y in {0/2.75, -1.2/-1, 1.2/-1, -2.2/-1.5, 2.2/-1.5, -3.2/-2, 3.2/-2}{
    \filldraw[figgreen!55, draw=black!35, thin] (\x,\y) circle (2pt);
  }
  \foreach \x/\y in {0/1.5, 0/-1.75, -1.2/0, 1.2/0, -2.2/0, 2.2/0, -3.2/0, 3.2/0}{
    \draw[zeroblue!50, line width=1pt]
      (\x-0.17,\y-0.17)--(\x+0.17,\y+0.17)
      (\x-0.17,\y+0.17)--(\x+0.17,\y-0.17);
  }
  \filldraw[figgreen, draw=black, line width=0.6pt] (0,0.7) circle (3.8pt)
    node[right=8pt, font=\footnotesize, text=black, fill=white, draw=gray!30, thin, inner sep=1pt]{$k_\star = i\kappa_\star$};
  \draw[<->, gray!55, line width=0.6pt] (1,2)--(3,2) node[below,font=\footnotesize,gray!70]{$x$};
  \bifwell{2}{2}{0.4}{0.55}
  \begin{scope}[shift={(2,2)}, xscale=0.5, yscale=0.35]
    \draw[black, line width=1.4pt, domain=-2:2, samples=80, smooth]
      plot(\x, {0.8*bifwarmTanh(4*\x)*exp(-0.6*abs(\x))});
  \end{scope}
  \node[font=\scriptsize, black] at (2.5,2.58) {$U_\star$};
\end{tikzpicture}
\end{minipage}%
\hfill%
\begin{minipage}[t]{0.50\linewidth}\centering
\vspace{0pt}%
\begin{tikzpicture}[scale=0.65,transform shape]
  \fill[panelbgRW, rounded corners=4pt] (-12,-0.7) rectangle (0.7,6.2);
  \draw[->] (0,0)--(-11.5,0) node[below]{$E$};
  \draw[->] (0,0)--(0,6) node[right]{$\mathcal{N}[\psi_E]$};
  \draw[rwgray, line width=1.8pt] plot[smooth, tension=1]
    coordinates{(-1,0)(-3,2)(-6,3.5)(-9,4.6)(-11,5.2)};
  \filldraw[figgreen, draw=black, line width=0.5pt] (-1,0) circle (2.8pt);
  \node[below=4pt, font=\footnotesize, text=black, fill=white, draw=gray!30, thin, inner sep=1pt]
    at (-1,0) {$E_\star = k_\star^2$};
  \foreach \nx/\ny/\amp/\tran in {-1.4/0.9/.1/.1, -5.0/3.0/.5/.1, -9.5/4.7/.8/.1}{
    \draw[<->, gray!55, line width=0.6pt] (\nx-0.9,\ny)--(\nx+0.9,\ny) node[right,font=\footnotesize,gray!70]{$x$};
    \draw[black, line width=1.2pt, domain=-0.8:0.8, samples=50, smooth, shift={(\nx,\ny)}]
      plot(\x,{0.8*\amp*(bifwarmSech(12*(\x-\tran))-bifwarmSech(12*(\x+\tran)))});
    \bifwell{\nx}{\ny}{0.13}{0.4}
  }
\end{tikzpicture}
\end{minipage}

\vspace{2pt}

\begin{minipage}[t]{0.47\linewidth}\centering
\vspace{0pt}%
\begin{tikzpicture}[scale=0.65,transform shape]
  \fill[panelbgNew, rounded corners=4pt] (-3.7,-3.7) rectangle (3.7,3.7);
  \node[font=\normalsize\bfseries, anchor=south, rotate=90] at (-3.55,0) {(b) Scattering resonance};
  \draw[->] (-3.5,0)--(3.5,0) node[below,font=\footnotesize]{Re$(k)$};
  \draw[->] (0,-3.5)--(0,3.5) node[right,font=\footnotesize]{Im$(k)$};
  \foreach \x/\y in {0/1.5, -1/-1, 1/-1, -2/-1.5, 2/-1.5, -3/-2, 3/-2}{
    \filldraw[figgreen!55, draw=black!35, thin] (\x,\y) circle (2pt);
  }
  \foreach \x/\y in {-1.2/0, 1.2/0, -2.2/0, 2.2/0, -3.2/0, 3.2/0}{
    \draw[zeroblue!50, line width=1pt]
      (\x-0.17,\y-0.17)--(\x+0.17,\y+0.17)
      (\x-0.17,\y+0.17)--(\x+0.17,\y-0.17);
  }
  \filldraw[figgreen, draw=black, line width=0.6pt] (0,-0.7) circle (3.8pt)
    node[right=8pt, font=\footnotesize, text=black, fill=white, draw=gray!30, thin, inner sep=1pt]{$k_\star={-}i\gs$};
  \draw[<->, gray!55, line width=0.6pt] (1,2)--(3,2) node[below,font=\footnotesize,gray!70]{$x$};
  \bifwell{2}{2}{0.4}{0.2}
  \begin{scope}[shift={(2,2)}, xscale=0.5, yscale=0.35]
    \draw[black, line width=1.4pt, domain=-2:2, samples=80, smooth]
      plot(\x, {0.3*bifwarmTanh(3*\x)*exp(0.6*abs(\x))});
  \end{scope}
  \node[font=\scriptsize, black] at (2.5,2.58) {$U_\star$};
\end{tikzpicture}
\end{minipage}%
\hfill%
\begin{minipage}[t]{0.50\linewidth}\centering
\vspace{0pt}%
\begin{tikzpicture}[scale=0.65,transform shape]
  \fill[panelbgNew, rounded corners=4pt] (-12,-0.7) rectangle (0.7,6.2);
  \draw[->] (0,0)--(-11.5,0) node[below]{$E$};
  \draw[->] (0,0)--(0,6) node[right]{$\mathcal{N}[\psi_E]$};
  \draw[refgray, line width=0.9pt, dotted, domain=0.05:8.4, samples=100]
    plot({-\x},{2.0*sqrt(\x)});
  \node[refgray, font=\footnotesize, anchor=south] at (-5,4.65)
    {$2\mathcal{N}[\mathcal{S}_E]$};
  \draw[srbrown, line width=1.8pt] plot[smooth, tension=1]
    coordinates{(-1,2)(-3,3.0)(-6,3.8)(-9,4.4)(-11,4.8)};
  \filldraw[figgreen, draw=black, line width=0.6pt] (-1,2) circle (3.8pt);
  \draw[black!70] (-1,0)--(-1,-0.15)
    node[below, font=\footnotesize, text=black, fill=white, draw=gray!30, thin, inner sep=1pt]{$E_\star = k_\star^2$};
  \draw[black!70] (0,2)--(0.12,2)
    node[right, font=\footnotesize, text=black, fill=white, draw=gray!30, thin, inner sep=1pt]{$\mathcal{N}_\star$};
  \foreach \nx/\ny/\amp/\tran/\scal in {-2.5/3.3/.5/.6/1.0, -6.0/4.4/.5/.3/1.4, -9.5/5.3/1/.1/2.0}{
    \draw[<->, gray!55, line width=0.6pt] (\nx-0.9,\ny)--(\nx+0.9,\ny) node[right,font=\footnotesize,gray!70]{$x$};
    \draw[black, line width=1.2pt, domain=-0.9:0.9, samples=200, smooth, shift={(\nx,\ny)}]
      plot(\x,{0.8*\amp*(bifwarmSech(12*(\scal*\x-\tran))-bifwarmSech(12*(\scal*\x+\tran)))});
    \bifwell{\nx}{\ny}{0.13}{0.2}
    \ifdim \nx pt = -2.5pt
      \node[font=\scriptsize] at ({\nx-\tran/\scal-.18},{\ny-0.44}) {$-\infty\!\leftarrow$};
      \node[font=\scriptsize] at ({\nx+\tran/\scal+.05},{\ny+0.43}) {$\to\!\infty$};
    \fi
  }
\end{tikzpicture}
\end{minipage}

\vspace{2pt}

\begin{minipage}[t]{0.47\linewidth}\centering
\vspace{0pt}%
\begin{tikzpicture}[scale=0.65,transform shape]
  \fill[panelbgNew, rounded corners=4pt] (-3.7,-3.7) rectangle (3.7,3.7);
  \node[font=\normalsize\bfseries, anchor=south, rotate=90] at (-3.55,0) {(c) Transmission resonance};
  \draw[->] (-3.5,0)--(3.5,0) node[below,font=\footnotesize]{Re$(k)$};
  \draw[->] (0,-3.5)--(0,3.5) node[right,font=\footnotesize]{Im$(k)$};
  \foreach \x/\y in {0/.2, 0/1.5, -1/-1, 1/-1, -2/-1.5, 2/-1.5, -3/-2, 3/-2}{
    \filldraw[figgreen!55, draw=black!35, thin] (\x,\y) circle (2pt);
  }
  \foreach \x/\y in {0/.7, -1.2/0, 1.2/0, -2.2/0, 2.2/0, -3.2/0, 3.2/0}{
    \draw[zeroblue!50, line width=1pt]
      (\x-0.17,\y-0.17)--(\x+0.17,\y+0.17)
      (\x-0.17,\y+0.17)--(\x+0.17,\y-0.17);
  }
  \draw[zeroblue, line width=2.2pt, line cap=round]
    (-0.25,-0.95)--(0.25,-0.45) (-0.25,-0.45)--(0.25,-0.95);
  \node[right=8pt, font=\footnotesize, text=black, fill=white, draw=gray!30, thin, inner sep=1pt] at (0,-0.7)
    {$k_\star={-}i\gs$};
  \draw[<->, gray!55, line width=0.6pt] (1,2)--(3,2) node[below,font=\footnotesize,gray!70]{$x$};
  \bifwell{2}{2}{0.4}{0.4}
  \begin{scope}[shift={(2,2)}, xscale=0.5, yscale=0.35]
    \draw[black, line width=1.4pt, domain=-2:2, samples=80, smooth]
      plot(\x, {0.7*bifwarmTanh(6*\x)*exp(0.4*\x)});
  \end{scope}
  \node[font=\scriptsize, black] at (2.5,2.58) {$U_\star$};
\end{tikzpicture}
\end{minipage}%
\hfill%
\begin{minipage}[t]{0.50\linewidth}\centering
\vspace{0pt}%
\begin{tikzpicture}[scale=0.65,transform shape]
  \fill[panelbgNew, rounded corners=4pt] (-12,-0.7) rectangle (0.7,6.2);
  \draw[->] (0,0)--(-11.5,0) node[below]{$E$};
  \draw[->] (0,0)--(0,6) node[right]{$\mathcal{N}[\psi_E]$};
  \draw[refgray, line width=0.9pt, dashed, domain=0.05:11.5, samples=120]
    plot({-\x},{1.0*sqrt(\x)});
  \node[refgray, font=\footnotesize, anchor=south] at (-5,2.38)
    {$\mathcal{N}[\mathcal{S}_E]$};
  \draw[trblue, line width=1.8pt] plot[smooth, tension=0.6]
    coordinates{(-1,1)(-2,2.1)(-3.5,3.2)(-5.5,4.2)(-7.5,4.9)(-9.5,5.4)(-11,5.7)};
  \filldraw[white, draw=black, thin] (-1,1) circle (2.5pt);
  \draw[zeroblue, line width=2.2pt, line cap=round]
    (-1.25,0.75)--(-0.75,1.25) (-1.25,1.25)--(-0.75,0.75);
  \draw[black!70] (-1,0)--(-1,-0.15)
    node[below, font=\footnotesize, text=black, fill=white, draw=gray!30, thin, inner sep=1pt]{$E_\star = k_\star^2$};
  \draw[black!70] (0,1)--(0.12,1)
    node[right, font=\footnotesize, text=black, fill=white, draw=gray!30, thin, inner sep=1pt]{$\mathcal{N}_\star$};
  \draw[<->, gray!55, line width=0.6pt] (-9.5-0.9,5.55)--(-9.5+0.9,5.55) node[right,font=\footnotesize,gray!70]{$x$};
  \draw[black, line width=1.2pt, domain=-0.8:0.8, samples=120, smooth, shift={(-9.5,5.55)}]
    plot(\x,{0.8*0.7*(bifwarmSech(28*(\x-0.1))-.8*bifwarmSech(28*(\x+0.1)))});
  \bifwell{-9.5}{5.55}{0.13}{0.3}
  \draw[<->, gray!55, line width=0.6pt] (-6.0-0.9,4.55)--(-6.0+0.9,4.55) node[right,font=\footnotesize,gray!70]{$x$};
  \draw[black, line width=1.2pt, domain=-0.8:0.8, samples=60, smooth, shift={(-6.0,4.55)}]
    plot(\x,{0.8*0.5*(1.1*bifwarmSech(25*(\x-0.3))-.2*bifwarmSech(12*(\x+0.07)))});
  \bifwell{-6.0}{4.55}{0.13}{0.3}
  \pgfmathsetmacro{\xc}{-2.5+0.6}
  \pgfmathsetmacro{\yc}{2.8}
  \draw[<->, gray!55, line width=0.6pt] (-2.5-0.9,2.8)--(-2.5+0.9,2.8) node[right,font=\footnotesize,gray!70]{$x$};
  \draw[black, line width=1.2pt, domain=-1.4:0.3, samples=50, smooth, shift={(\xc,\yc)}]
    plot(\x,{0.8*0.5*bifwarmSech(15*\x)});
  \node[font=\scriptsize] at (\xc,\yc+0.47) {$\to\!\infty$};
  \bifwell{-2.5}{2.8}{0.13}{0.3}
\end{tikzpicture}
\end{minipage}

\vspace{1pt}
\setlength{\abovecaptionskip}{3pt}
\captionsetup{font=footnotesize}
\caption{%
  Schematic bifurcation diagrams for the three nonlinear bound-state bifurcation scenarios.
  Left: the $k$-plane, with the bifurcating pole (\specpole) or zero (\speczero) highlighted.
  Right: the corresponding $(\mathcal{N},E)$ branches, with reference soliton curves and inset profiles.
  Panel \textbf{(a)} is the $L^2$ bound-state case~\cite{rose1988bound}; panels \textbf{(b)} and \textbf{(c)} are the scattering- and transmission-resonance cases from~\cite{turner2026resonance}.%
}
\label{fig:warmup-bifurcation}
\end{figure}
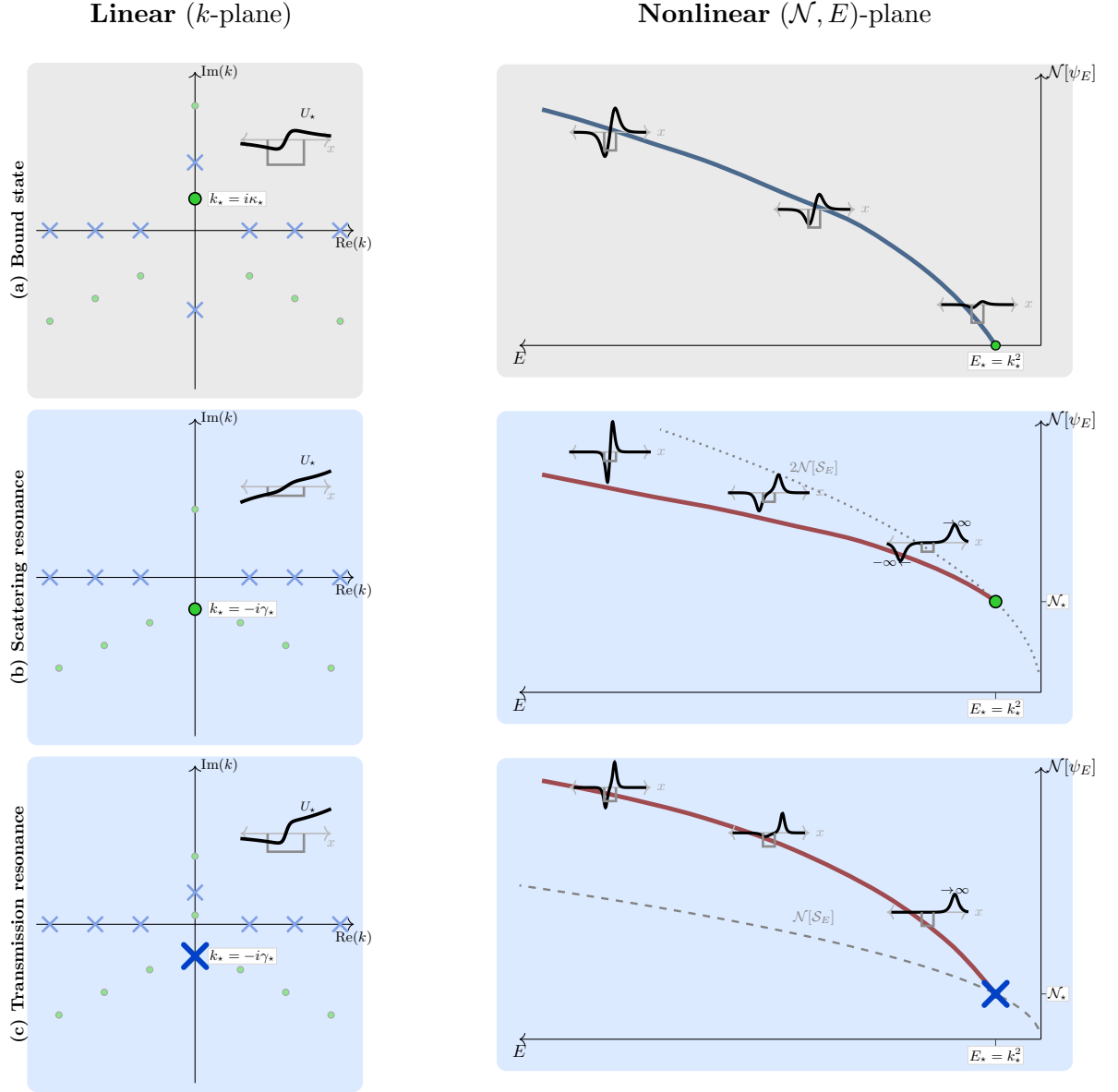
\subsection{Scattering and transmission resonances}\label{ssec:resonance-def}

Let $V \in L^1(\R)$ and be compactly supported in $[-b,b]$. Set $H_V = -\partial_x^2 + V$. For $k \in \C$, the \emph{Jost solutions} $f_\pm(x,k)$ are the unique solutions of $H_V\phi = k^2\phi$ satisfying $f_\pm(x,k)\sim e^{\pm ikx}$ as $x\to\pm\infty$; see  \cite{deift1979inverse,Newton1982,korotyaev2005inverse}. For any two solutions $f, g$ of $H_V\phi = k^2\phi$, their \emph{Wronskian} $W[f,g] := fg' - f'g$ is independent of $x$. 
Define the \emph{left reflection coefficient}
\begin{equation}\label{eq:r-minus-def}
  r_-(k) := \frac{W\bigl[f_+(\,\cdot\,,k),\, f_-(\,\cdot\,,-k)\bigr]}{W\bigl[f_+(\,\cdot\,,k),\, f_-(\,\cdot\,,k)\bigr]}, \qquad k \in \R\ .
\end{equation}
On the real axis, $\overline{f_-(x,k)} = f_-(x,-k)$. Both Jost solutions $f_\pm(x,k)$ extend analytically in $k$ to all of $\C$, since $V$ has compact support. The same is true of the two Wronskians appearing in \eqref{eq:r-minus-def} and hence $r_-(k)$ extends meromorphically to $\C$. 
The \emph{right reflection coefficient} is
\begin{equation}\label{eq:r-plus-def}
  r_+(k) \;:=\; \frac{W\bigl[f_+(\,\cdot\,,-k),\, f_-(\,\cdot\,,k)\bigr]}{W\bigl[f_+(\,\cdot\,,k),\, f_-(\,\cdot\,,k)\bigr]}, \qquad k \in \R\ .
\end{equation}
Since $r_\pm(k)$ share the denominator $W[f_+(\cdot,k),\,f_-(\cdot,k)]$, their poles (bound states and scattering resonances) coincide.
Their zeros (transmission resonances) are related by $k$-reflection: $r_-(k_\star)=0$ if and only if $r_+(-k_\star)=0$, since the numerator of $r_+$ at $k$ equals the numerator of $r_-$ at $-k$.

Poles of $r_-(k)$ in $\{\operatorname{Im} k > 0\}$ correspond to $L^2$ bound states of $H_V$; by self-adjointness, these lie on $i\R_+$. Poles $k_\star$ in $\{\operatorname{Im} k < 0\}$ are \emph{scattering resonances}, and those resonances which have zero real part correspond to \emph{anti-bound states}. Zeros $k_\star$ of $r_-(k)$ are \emph{transmission resonances}. These three cases correspond to the left-column panels of Figure~\ref{fig:warmup-bifurcation}.

Nonlinear bound states are time-harmonic solutions $e^{-iEt}\psi(x;E)$ and  hence $E\in\mathbb R$ must be real. Further, since $\psi(x;E)$ decays at spatial infinity, $E<0$. Our resonance-induced nonlinear bound states bifurcate from scattering or transmission 
resonances $(k_\star,U_\star)$; see Theorem \ref{thm:intro-bifurcation}  and Theorem \ref{thm:bifurcation} below.  Thus, at the bifurcation energy, $E_\star$, we require $k_\star^2=E_\star<0$, which forces $k_\star\in i\R$. 
We write 
\[ k_\star := -i\gs\qquad \textrm{with}\qquad  \gs\in\R\quad \textrm{and resonance energy $E=-\gs^2$.}\]
We delineate three cases:
\begin{itemize}[leftmargin=2em]
\item $L^2$ bound states / poles of $r_-(k)$ in the $\C_+$,  $\gs<0$. These bifurcations were studied in \cite{rose1988bound};
\item Scattering resonances / poles of $r_-(k)$ in the $\C_-$,  $\gs>0$;
\item Transmission resonances / zeros of $r_-(k)$ on $i\R$, $\gs\in\R$. By Remark~\ref{rem:gamma-sign}, one may reduce to the case $\gs>0$.
\end{itemize}
Thus we take $\gs>0$.
For both  \emph{scattering resonances} (SR) and \emph{transmission resonances} (TR) we work with $\gs>0$.

For scattering resonances, we work on the half-line $X = \Hd$; see Remark~\ref{rem:full-line-instability}. For transmission resonances, we work on the full line $X = \Hf$.

A resonance pair  $(\gs,U_\star)$ is a nontrivial solution of
\begin{equation}\label{eq:resonance-eq}
  (H_V + \gs^2)U_\star = 0, \qquad E_\star = -\gs^2.
\end{equation}
 subject to scattering resonance (SR) or transmission resonance (TR) boundary conditions:
\begin{itemize}[leftmargin=2em]
\item \emph{Scattering resonance (SR) case:} ${\rm domain}(H_V) = \Hd$ and
\begin{equation}\label{eq:SR-BC}
  U_\star(0) = 0, \qquad U_\star'(b) = \gs\, U_\star(b), 
\end{equation}
\item \emph{Transmission resonance (TR) case:} ${\rm domain}(H_V) = \Hf$ and
\begin{equation}\label{eq:TR-BC}
  U_\star'(-b) = \gs\, U_\star(-b), \qquad U_\star'(b) = \gs\, U_\star(b).
\end{equation}
\end{itemize}
The boundary conditions in \eqref{eq:SR-BC} and \eqref{eq:TR-BC} correspond to the relevant outgoing condition at $\pm b$ for $k_\star=-i\gs$ on the imaginary axis.

Note that if $(\gs,U_\star)$ is a resonance pair, then $U_\star(b)\ne0$. Indeed, if $U_\star(b)=0$, then resonance boundary condition at $x=b$ (for both TR and SR cases) implies $U_\star'(b)=0$. Hence,  $U_\star\equiv0$, a contradiction. 

For a resonance eigenpair $(\gs, U_\star)$, we choose the normalizations:
\begin{equation}\label{eq:resonance-normalization}
\begin{aligned}
\text{(SR):}\quad & U_\star'(0) = 1\quad 
{\rm and}\quad \text{(TR):}\quad U_\star(-b) = 1.
\end{aligned}
\end{equation}
Outside the support of $V$, the scattering resonance profile satisfies $U_\star(x) \propto e^{\gs|x|}$, while the transmission profile satisfies $U_\star(x) \propto e^{\gs x}$ (with $\gs>0$, exponentially growing to the right and decaying to the left).

\subsection{A brief review of resonance-induced nonlinear bound states \cite{turner2026resonance}}



The following result is essentially Theorem~4.1 of \cite{turner2026resonance}; see also Theorem \ref{thm:intro-bifurcation}.

\begin{thm}[Resonance-induced nonlinear bound states {\cite[Theorem~4.1]{turner2026resonance}}]\label{thm:bifurcation}
  Let $(\gs, U_\star)$ be a scattering or transmission resonance pair with $\gs > 0$ as defined in Section~\ref{ssec:resonance-def}.  Assume that $(\gs, U_\star)$ satisfies the non-degeneracy condition $\mathcal{B}_\star\ne0$.  

  Then, there exists $\eps_0 > 0$ such that for $0 < \eps < \eps_0$ the following hold:
  \begin{enumerate}[leftmargin=2em, label=\textup{(\roman*)}]
  \item There is a curve $\eps\mapsto(E(\eps),\psi_\eps)\in\R\times X$, with 
  \[\psi_\eps(x) = \sqrt\eps\,\phi_\eps(x)\quad  \textrm{and}\quad \phi_0(x) = U_\star(x), \] of solutions to \eqref{eq:V-nls} which bifurcates from the trivial solution at $E=E_\star$. The maps 
  \begin{align*}
    \eps&\mapsto E(\eps)\quad\textrm{and}\quad 
    \eps\mapsto \phi_\eps\big|_{[a,b]} 
  \end{align*} 
  extend to real analytic functions in an open neighborhood of $\eps=0$. Any non-trivial solution   with $E$ near $E_\star$ is equal to $\psi_\varepsilon$ up to a multiplicative constant  of modulus one.
  \item The bifurcation occurs at the $L^2(\mathbb R)$-threshold $\mathcal N_\star = 4\gs$; see the right-column panels of Figure~\ref{fig:warmup-bifurcation}.
  \item \emph{Nonlinear bound state eigenpair expansion:} 
  \begin{align*}
   E(\eps) &= -\gs^2 + \eps\dot E + \mathcal O(\eps^2), \quad  \dot E := \left.dE/d\eps\right|_{\eps=0}; \textrm{ see \eqref{eq:Edot-solved}.}\\
    \phi_\eps\big|_{[a,b]} &= U_\star + \mathcal{O}(\eps)
  \end{align*} 
  \item\emph{SR-induced mode spatial structure on $\R_+$:}
  \begin{equation}
  \psi_\eps(x) =
  \begin{cases}
      \sqrt\eps\,[U_\star + \mathcal O(\eps)], & 0\le x\le b\\
      \mathcal{S}(x-y_R(\eps);E(\eps)) &   x>b, 
  \end{cases}
\end{equation}
where $y_R$ (see \eqref{eq:yRL}) satisfies $y_R(\eps)\to +\infty$ as $\eps\to0$.
  \item\emph{TR-induced mode spatial structure on $\mathbb R$ ($\gs>0$):}
  \begin{equation}\label{eq:psi_eps-expand}
  \psi_\eps(x) =
  \begin{cases}
  \mathcal{S}(x-y_L(\eps);E(\eps)) &   x<-b, \\
      \sqrt\eps\,[U_\star + \mathcal O(\eps)], & -b\le x\le b\\
      \mathcal{S}(x-y_R(\eps);E(\eps)) &   x>b, 
  \end{cases}
\end{equation}
where $y_R(\eps)$ and $y_L(\eps)$  tend to $ +\infty$ as $\eps\to0$, 
   and are given by
  \begin{subequations}\label{eq:yRL}
  \begin{align}\label{eq:yRL-tanh}
    y_R(\eps) &\;=\; b + \frac{1}{|E(\eps)|^{1/2}}\,\mathrm{arctanh}\!\sqrt{1 - \frac{\psi^2_\eps(+b^-)}{2\,|E(\eps)|}}; \ \ \psi^2_\eps(+b^-) = \eps U_\star^2(b)  + \mathcal{O}(\eps^2). \\
\label{eq:yL-tanh}
    y_L(\eps) &\;=\; -b + \frac{1}{|E(\eps)|^{1/2}}\,\mathrm{arctanh}\!\sqrt{1 - \frac{\psi^2_\eps(-b^+)}{2\,|E(\eps)|}};\ \ \psi^2_\eps(-b^+) = \eps + \mathcal{O}(\eps^2).
  \end{align}
  \end{subequations}
  \item 
  \begin{equation}\label{eq:Edot-solved}
    \dot E \;=\;
    \begin{cases}
      \dfrac{\,(U_\star(b)^4 - 1)/(2\gs) \;-\; 2\!\int_{-b}^b U_\star^4(s)\,ds\,}{2\,\mathcal B_\star^{\rm TR}} & \text{(TR)}, 
      \\[12pt]
      \dfrac{\,U_\star(b)^4/(2\gs) \;-\; 2\!\int_0^b U_\star^4(s)\,ds\,}{2\,\mathcal B_\star^{\rm SR}}
      & \text{(SR)}\ .
    \end{cases}
\end{equation}

  \end{enumerate}
 Note: In the construction of $\psi_\varepsilon(x)$,  $E(\eps)$, $\psi_\eps(-b^+)$ and $\psi_\eps(b^-)$, which appear in \eqref{eq:yRL} are first determined by the ``inner'' resonance boundary value problems \eqref{eq:resonance-eq}-\eqref{eq:SR-BC} on $[0,b]$, respectively \eqref{eq:resonance-eq}-\eqref{eq:TR-BC} on $[-b,b]$, and then matched to a soliton in the region outside ${\rm supp}(V)$.
\end{thm}



Theorem \ref{thm:bifurcation} is essentially Theorem~4.1 of \cite{turner2026resonance}, but with more explicit detail concerning $y_L(\eps)$, $y_R(\eps)$, and the expression for $\dot E$. The proof of~\eqref{eq:Edot-solved} is in Section~\ref{ssec:proof-Edot}; the proof of~\eqref{eq:yRL-tanh},~\eqref{eq:yL-tanh} is in Section~\ref{ssec:proof-centers}.

\begin{rem}[Sign of $\gs$]\label{rem:gamma-sign}
Without loss of generality, we work with $\gs > 0$  throughout; $k_\star=-i\gs\in i\mathbb R_-$ for both SR and TR. Scattering resonances always live on $\{\operatorname{Im}(k_\star)< 0 \}$. The upper half-plane transmission case $k_\star \in i\R_+$ can be mapped to the lower half-plane via the reflection $V(x) \mapsto V(-x)$. By the same reflection, the TR formulas~\eqref{eq:Bstar},~\eqref{eq:Edot-solved} persist verbatim if one works with the mirrored normalization $U_\star(x) \mapsto U_\star(-x)$ with normalization $U_\star(b) = 1$. In this case,  the roles of $U_\star(b)$ and $U_\star(-b)$ are swapped throughout, but $\mathcal B_\star$ and $\dot E$ are unchanged. In the soliton-center expressions~\eqref{eq:yRL-tanh},~\eqref{eq:yL-tanh}  the $\mathrm{arctanh}$ terms acquire a minus sign, implying that $y_{R,L}(\eps) \to -\infty$ (rather than $+\infty$) as $\eps \to 0^+$.
\end{rem}

\medskip

\subsection{The P\"oschl--Teller potential}\label{ssec:PT}

As detailed in Section \ref{sec:stability-theory}, the stability of a nonlinear bound state depends crucially on the spectral properties of the linearized operators $L_-$ and $L_+$, each a  Schr\"odinger operator with a localized potential derived from the nonlinear bound state, $\psi_\varepsilon$. For the case  of cubic (indeed any homogeneous power) nonlinearity, $L_-$ and $L_+$ are operators with {\it P\"oschl--Teller} type  potentials. We shall make extensive use of the explicit spectral properties of such operators.

\begin{definition}[P\"oschl--Teller potentials and Schr\"odinger operators]\label{def:PT}
For parameters: $\ell = 1,2,3,\dots$ and $\la_0 > 0$, we introduce the  \emph{P\"oschl--Teller potential}  \[V_{\rm PT}(\xi;\la_0) = -\ell(\ell+1)\la_0^2\sech^2(\la_0\xi), \] and 
the P\"oschl--Teller Schr\"odinger operator:
\begin{equation}\label{eq:PTpot} H_{\rm PT}(\la_0,\ell) = -\partial_\xi^2 + V_{\rm PT}(\xi;\lambda_0,\ell).\end{equation}
The P\"oschl--Teller Schr\"odinger operator with potential centered at $\xi=y$ is denoted by:
\begin{equation}\label{eq:PTpot_y} H_{\rm PT}(\la_0,\ell,y) = -\partial_\xi^2 + V_{\rm PT}(\xi-y;\lambda_0,\ell).\end{equation}
\end{definition}

We next construct a class of scattering states for $H_{\rm PT}$ in terms of the Jost solutions,  $f_\pm(\xi, k)$, of Section~\ref{ssec:resonance-def}. This will lead us to a determination of the point spectrum of $H_{\rm PT}$.

\begin{lem}[Jost solutions of $H_{\rm PT}(\la_0,\ell)$, $\ell = 1, 2$]\label{lem:Jost}
Let $\ell\in\{1,2\}$. For $\la>0$, define  
\begin{subequations}
\label{eq:m_and_t}
    \begin{align}
    m &:= \la/\la_0 \label{eq:m-def}\\
    t(\xi) &:= \tanh(\la_0\xi).\label{eq:t-def}
    \end{align}
\end{subequations}
The unique solution of
\begin{equation} H_{\rm PT}(\lambda_0,\ell)\phi = -\la^2\phi\quad \textrm{ with }\quad \phi(\xi)\sim e^{-\la\xi}  \quad \textrm{as $\xi\to+\infty$}\end{equation}
is given in terms of the Jost solution $f_+$:
\begin{equation}\label{eq:Jost}
f_+(\xi, i\la; \la_0) = \begin{cases}
\displaystyle e^{-\la\xi}\,\frac{m + t(\xi)}{m + 1} & (\ell = 1),\\[6pt]
\displaystyle e^{-\la\xi}\,\frac{p\bigl(m, t(\xi)\bigr)}{(m+1)(m+2)} & (\ell = 2)\ ,
\end{cases}
\end{equation}
where, for $\ell=2$,
\begin{equation}\label{eq:p-def}
  p(m,t) \;:=\; m^2 - 1 + 3mt + 3t^2.
\end{equation}
Further,  the unique solution with asymptotic behavior $\phi(\xi)\sim e^{+\la\xi}$ as $\xi\to-\infty$ is given by
\begin{equation} \label{eq:f_pm}
 f_-(\xi, i\la; \la_0) = f_+(-\xi, i\la; \la_0).\end{equation}
\end{lem}

\begin{proof}
See \cite{Flugge1971}, \cite{LandauLifshitzQM}.
\end{proof}

\begin{lem}[Spectrum of $H_{\rm PT}$, $\ell = 1, 2$]\label{prop:PT-spectrum}
Let $\ell\in\{1,2\}$ and $\la_0>0$. Then $\sigma_{\rm cont}(H_{\rm PT}) = [0,\infty)$, and the discrete spectrum is
\begin{equation}\label{eq:PT-spectrum-disc}
\begin{aligned}
\ell=1: \quad & \sigma_{\rm disc}(H_{\rm PT}) = \{-\la_0^2\}, && \text{eigenfunction } \sech(\la_0\xi);\\
\ell=2: \quad & \sigma_{\rm disc}(H_{\rm PT}) = \{-4\la_0^2,\,-\la_0^2\}, && \text{eigenfunctions } \sech^2(\la_0\xi),\;\sech(\la_0\xi)\tanh(\la_0\xi).
\end{aligned}
\end{equation}

\end{lem}

\begin{proof}
See \cite{Flugge1971}, \cite{LandauLifshitzQM}.




\end{proof}

\section{Review of orbital stability theory for NLS/GP and some applications}\label{sec:stability-theory}

For any initial condition $\Psi_0\in H^1(\R)$, NLS/GP~\eqref{eq:nlse} has a unique global-in-time solution $\Psi\in C\bigl([0,\infty);H^1(\R)\bigr)$, for which the Hamiltonian and $L^2$ norm are conserved in $t$:
\begin{equation}\label{eq:H-N-conserved}
\begin{aligned}
    \mathcal H[\Psi] &:= \int |\partial_x\Psi(x,t)|^2 + V(x) |\Psi(x,t)|^2 - \tfrac 12 | \Psi(x,t)|^4 \, dx, \\
     \mathcal N[\Psi] &:= \int |\Psi(x,t)|^2\,dx\ ;
\end{aligned}
\end{equation}
see, e.g.\ \cite{cazenave2003semilinear,sulemsulem1999,fibich2015nonlinear}.
For symmetric potentials, $V(-x)=V(x)$, the theory applies to data $\Psi_0\in X$ giving rise to solutions in $\Psi\in C\bigl([0,\infty);X\bigr)$, where $X=H^1(\mathbb R)$ or $X=H^1_0(\mathbb R_+)$

The stability theory of nonlinear bound states of NLS/GP, developed in \cite{weinstein1986lyapunov,grillakis1987stability} (see also \cite{grillakis1990stability,weinstein-survey:15}), is based on a variational principle, which we now explain.  Since any nonlinear bound state solution of \eqref{eq:V-nls} on $X$ is a critical point  $\psi(\cdot ;E)\in X$ of the time-invariant energy functional
\begin{equation} \mathcal{E}[\psi]\ :=\ \mathcal H[\psi] - E\,\mathcal N[\psi]\ , \qquad E < 0,\end{equation}
a sufficient condition for orbital stability (Definition~\ref{sec:orbital}) is positive definiteness of the second variation of $\mathcal E$ at $\psi( \, \cdot \, ;E)$, restricted to the tangent space of the constraint $\{\mathcal N=\mathcal N(\psi)\}$. 

The second variation is given by:
\begin{equation}
\frac{d^2}{d\tau^2} \mathcal{E}\bigl[\psi( \, \cdot \, ;E) + \tau(u+iv)\bigr]\Big|_{\tau=0}\ =\ \langle L_+ u, u\rangle\ +\ \langle L_- v, v\rangle\ ,
\label{eq:hessian}\end{equation}
where
\begin{equation}\label{eq:Lpm-def}
  L_+ := -\partial_x^2 + V - 3\psi( \, \cdot \, ;E)^2 - E, \qquad L_- := -\partial_x^2 + V - \psi( \, \cdot \, ;E)^2 - E,
\end{equation}
acting on the relevant subspace $X\subset H^1$; see \eqref{eq:Xspace}. We note that 
$ \ker L_-\big|_X = \spn\{\psi( \, \cdot \, ;E)\}$.

 Let $A$ denote a self-adjoint operator, all of  whose essential spectrum is strictly positive,   and denote by $n_-(A)$  the number of negative eigenvalues of an operator $A$. Then, we have the following:

\begin{thm}[Stability and stability criteria \cite{weinstein1986lyapunov,grillakis1988linearized,grillakis1990stability,jones1988instability}]\label{thm:W-GSS}
Consider the initial value problem for NLS/GP in the space $C(\mathbb R; X)$; see \eqref{eq:Xspace}. 
  Let $\psi( \, \cdot \, ;E)e^{-iEt}\in X$ be a bound state of \eqref{eq:V-nls}, 
   and let $L_\pm$ act on $X$ as in~\eqref{eq:Lpm-def}. 
  \begin{enumerate}[leftmargin=2em, label=\textup{(\roman*)}]
    \item \emph{Sufficient for orbital stability \cite{weinstein1986lyapunov,grillakis1987stability}:}
      Assume
      \begin{subequations}
      \label{eq:stab-conditions}
      \begin{align}\label{eq:W-GSS-hypotheses}
      L_-\Big|_X &\geq 0,\ \textrm{i.e. $L_-$ has no negative point spectrum in the space $X$}\\
        n_-\!\left(L_+\Big|_X\right) &= 1, \ \textrm{i.e. $L_+$ has exactly one negative eigenvalue in the space $X$, and }  \label{eq:nminusLminus} \\ \qquad \frac{d\mathcal N[\psi(\cdot,E)]}{dE} &< 0.\label{eq:Nslope-1}
      \end{align}
      \end{subequations}
     Then, the nonlinear bound state $e^{-iEt}\psi( \, \cdot \, ;E)$ is orbitally stable in $X$.
    \item \emph{Sufficient for instability:}
     The nonlinear bound state  $e^{-iEt}\psi( \, \cdot \, ;E)$ is  unstable in $X$ if any one of the following holds:
      \begin{enumerate}[leftmargin=2em, label=\textup{(\alph*)}]
        \item $d\mathcal N/dE > 0$ \cite{grillakis1988linearized}; \label{wgss-a}
        \item $n_-(L_+\big|_X) - n_-(L_-\big|_X) \neq 0,\,1$ \cite{jones1988instability,grillakis1988linearized}; \label{wgss-b}
      \end{enumerate}
  \end{enumerate}
\end{thm}

We shall study the stability and instability properties of resonance-induced nonlinear bound states by studying $L_\pm$ in the relevant space, $X$, and applying Theorem~\ref{thm:W-GSS}. 

Note, from  Theorem~\ref{thm:bifurcation}, that the states $(\psi_\eps, E_\eps)$ are such that, outside the support of $V$, $\psi_\eps$ is identically equal to an $\eps$-dependent translate of the free soliton profile \eqref{eq:free-soliton-intro}. Therefore, for $x > \sup \supp(V)$ (and $x < \inf \supp(V)$ for TR), the corresponding $L^\eps_\pm$ operators are both Schr\"odinger operators with a translated P\"oschl--Teller potential. For $\eps>0$ and sufficiently small the P\"oschl--Teller potential is translated  far from ${\rm supp}(V)$ (see Theorem \ref{thm:bifurcation}) and so the P\"oschl--Teller operator figures centrally in the  spectral properties of:
\begin{equation}\label{eq:Lpm-eps}
L_+^\eps := -\partial_x^2 + V - 3\psi_\eps^2 - E(\eps), \qquad L_-^\eps := -\partial_x^2 + V - \psi_\eps^2 - E(\eps).
\end{equation}
For NLS/GP with $V=0$, we denote the analogous linearized  operators about the free soliton  by $L_\pm^{\rm sol}$.

\subsection{NLS/GP with $V=0$; stability of solitons of cubic (translation invariant)  NLS}\label{ssec:free-soliton}

The free soliton profile
\begin{equation}\label{def:sech} \Scal(\xi; E) = \sqrt{2}\la_0\sech(\la_0\xi), \quad \la_0 = \sqrt{-E}\end{equation}
solves $-\Scal'' - \Scal^3 = E\Scal$ on $\R$. The linearized operators about $\Scal$ are
\begin{equation}\label{eq:Lpm-free}
  L_+^{\rm sol} = -\partial_\xi^2 - 6\la_0^2\sech^2(\la_0\xi) + \la_0^2, \qquad L_-^{\rm sol} = -\partial_\xi^2 - 2\la_0^2\sech^2(\la_0\xi) + \la_0^2,
\end{equation}
which are P\"oschl--Teller operators (Definition~\ref{def:PT}) with $\ell = 2$ and $\ell = 1$, respectively, shifted by the constant $+\la_0^2$. Lemma~\ref{prop:PT-spectrum}, applied at $\ell=2$ and $\ell=1$ and shifted by $+\la_0^2$, gives
\begin{align}
\sigma_{\rm disc}(L_-^{\rm sol}) &= \{0\}, &\quad \sigma_{\rm cont}(L_-^{\rm sol}) &= [\la_0^2,\infty). \label{eq:Lm-spectrum}\\
  \sigma_{\rm disc}(L_+^{\rm sol}) &= \{-3\la_0^2,\; 0\}, &\quad \sigma_{\rm cont}(L_+^{\rm sol}) &= [\la_0^2,\infty), \label{eq:Lp-spectrum}
\end{align}

We next review the application of Stability / Instability Theorem~\ref{thm:W-GSS}, part (i), to the free soliton, $\Scal(x;E)$. Let us verify the hypotheses for orbital stability.

Note that $L_-^{\rm sol}\Scal=0$, and since $\Scal > 0$,  the ground state energy of $L_-^{\rm sol}$ is $0$. (Note also that it is the only point spectrum as displayed in \eqref{eq:Lm-spectrum}). Therefore, $L_-\ge0$.

Turning to  $L^{\rm sol}_+$, the zero eigenvalue of $L_+^{\rm sol}$ ($L_+^{\rm sol}\Scal'=0$,\ $\Scal'(\xi) \propto \sech(\lambda_0\xi)\tanh(\lambda_0\xi)$ odd) is associated with the spatial translation invariance of NLS in the case $V=0$, and is called the  \emph{translation mode}. Since $\Scal'$ changes sign exactly once (at $\xi=0$) on $\mathbb R$, $\Scal'$ is the second eigenvalue of $L^{\rm sol}_+$, and hence $n_-(L^{\rm sol}_+)=1$, as displayed in \eqref{eq:Lp-spectrum}. (We remark that the ground state eigenfunction of $L_+^{\rm sol}$ is $\sech^2(\la_0\xi)$ .)


Finally, the squared $L^2$ norm of the free  soliton is  $\mathcal N[\Scal_E] = 4\sqrt{-E}$. Therefore,  $d\mathcal N[\Scal_E]/dE = -2/\sqrt{-E} < 0$. All conditions of Theorem~\ref{thm:W-GSS}(i) hold and hence  the free soliton is orbitally stable.

\subsection{NLS/GP with $V \neq 0$; stability of nonlinear bound states induced by the $L^2$ ground state of $H_V$}\label{ssec:V0}

 Rose and Weinstein~\cite{rose1988bound}\label{ssec:Rose-W} proved the $H^1$-orbital Lyapunov stability of small amplitude nonlinear bound states which bifurcate from the zero solution  at simple eigenvalues of  $H_V$. Due to the presence of a non-trivial potential, $V(x)$, NLS/GP is not translation invariant.
 The key step in determining whether or not solutions on the ground state bifurcation curve are  orbitally stable is  a computation of $n_-(L_+)$. For energies $E$ near the linear ground state energy, the nonlinear bound state has uniformly small amplitude, and
standard analytic perturbation theory  yields that $L_+$ has one negative eigenvalue; $n_-(L_+)=1$. Finally, a verification of the slope condition in Theorem \ref{thm:W-GSS} yields stability.

\subsection{Preliminary remarks on stability analysis of resonance-induced nonlinear bound states}\label{sec:prelim-rmk}

As in the ground-state bifurcation of Section~\ref{ssec:V0}, the central
analytical task is the computation of $n_-(L_+^\eps)$.  The difficulty is
that resonance-induced states are \emph{not} uniformly small in amplitude:
they bifurcate at a positive $L^2$ threshold (Theorem~\ref{thm:bifurcation}),
so the standard perturbative argument used in~\cite{rose1988bound} does
not apply.

For $\eps>0$ small, most of the mass of $\psi_\eps$ is concentrated in a
free soliton profile centered at distance $\sim\log(1/\eps)$ from
$\supp(V)$.  One therefore expects $\Lop$ to possess an eigenvalue
$\mu(\eps)$ near zero, inherited from the translation mode
$\Scal'\propto\sech\tanh$ of $L_+^{\rm sol}$ (which has eigenvalue
$\mu=0$ for the free soliton; see~\eqref{eq:Lp-spectrum}).  The
leading-order expression for the perturbed translation eigenvalue,
$\mu(\eps)$, is obtained in Proposition~\ref{prop:mu} via matched
asymptotics and the scattering theory of P\"oschl--Teller potentials;
the calculations are delicate.  Proposition~\ref{prop:Lp-count} then
shows that, under the positivity hypothesis $U_\star\ge0$, $\mu(\eps)$
is the \emph{second} eigenvalue of $\Lop$; consequently,
\[
  n_-(L_+^\eps)=1 \quad\iff\quad \mu(\eps)>0.
\]
Finally, Proposition~\ref{prop:slope} determines $d\mathcal N/dE$ along
the resonance-induced branch near the bifurcation point.

\begin{rem}[Half-line restriction for scattering resonances]\label{rem:full-line-instability}
For symmetric $V(x) = V(-x)$, the scattering-resonance profiles in~\eqref{eq:SR-BC} are either odd or even. For $\varepsilon>0$ and small, the corresponding $\Hf$ resonance-induced (odd) nonlinear states have a spatial structure of free solitons (of opposite sign) centered at $y_R \gg b$ and $y_L \ll -b$. Therefore, the potential of $L_\pm^\eps$ consists of $V(x)$, centered at $x=0$, plus two identical wells which are concentrated far from $x=0$. The latter should contribute a \underline{pair} of discrete eigenvalues near zero. We believe therefore that it can be shown that our resonance-induced nonlinear bound states are dynamically \underline{unstable} in $H^1(\mathbb R)$, as an application of the instability criteria of Theorem~\ref{thm:W-GSS}(ii). In contrast, by restricting to the half-line case,  $X=\Hd$,  stability is restored.
\end{rem}

\section{Main stability theorem}\label{sec:main-thm}

We now turn to the proof of Theorem~\ref{thm:intro-stability-sharp} and derive Corollary~\ref{cor:small-gamma-stability}. We state three propositions: Proposition~\ref{prop:mu} (perturbed translation eigenvalue), Proposition~\ref{prop:Lp-count} (spectral count of $\Lop$ under positivity), and Proposition~\ref{prop:slope} (limiting mass slope), which are the key steps in the proof. These propositions are proved in Section~\ref{sec:proofs}.

Recall the function space $X$ from~\eqref{eq:Xspace} and the linearized operator $\Lop$ from~\eqref{eq:Lpm-eps}. Also recall the resonance pair $(\gs,U_\star)$-dependent parameters:
\[ \textrm{$\mathcal{B}_\star $ in \eqref{eq:Bstar}, $\Omega_\star$ in \eqref{eq:Omega}, and $\nu_\star$ in \eqref{eq:nu}.}\]
\begin{equation}\label{eq:assume-gs-pos}
\textrm{
Henceforth we assume, without loss of generality (see Remark~\ref{rem:gamma-sign}), that $\gs>0$ .
}
\end{equation}

For $\eps>0$ and small, let $\eps\mapsto (E(\eps),\psi_\eps(\cdot;E(\eps)))$ denote the family of transmission-resonance (TR) induced or scattering-resonance (SR) induced nonlinear bound states provided by Theorem~\ref{thm:bifurcation}.

The stability and instability of this family is determined by Theorem~\ref{thm:intro-stability-sharp}. The proof of Theorem~\ref{thm:intro-stability-sharp}, outlined below in Section~\ref{ssec:proof-main-thm},  is based on verifying the hypotheses of the  Stability Theorem \ref{thm:W-GSS}. The technical details are in subsequent sections. The assumptions in Theorem~\ref{thm:intro-stability-sharp} target the hypotheses of Theorem \ref{thm:W-GSS}: \ref{H:Bstar} ensures the existence of a resonance-induced branch, \ref{H:positivity} is  used to prove that $L_-$ has no negative point spectrum; $\Omega_\star>0$ ensures that the perturbed  ``translation eigenvalue'' is strictly positive and hence $n_-(L_+)=1$; \ref{H:Omega} and $\nu_\star<0$ imply the slope condition \eqref{eq:Nslope-1} for nonlinear bound states near the bifurcation.


A consequence of Theorem~\ref{thm:intro-stability-sharp} is the following characterization of stability and instability, purely in terms of the single parameter $\Omega_\star$ or in terms of $\nu_\star$.

\begin{cor}[Stability intervals: TR case]\label{cor:stability-TR}
\begin{enumerate}
\item The transmission-resonance induced nonlinear bound state $(E(\eps),\psi_\eps(\cdot;E(\eps)))$  is orbitally stable if:
\[
  0<\Omega_\star^{\rm TR}<\tfrac{3}{4\gs}
  \quad\text{or}\quad
  \Omega_\star^{\rm TR} > \tfrac{3}{4\gs} + \gs(\mathcal B_\star^{\rm TR})^2,
\]
and unstable if 
\[
  \Omega_\star^{\rm TR} < 0
  \quad\text{or}\quad
   \tfrac{3}{4\gs}<\Omega_\star^{\rm TR} < \tfrac{3}{4\gs} + \gs(\mathcal B_\star^{\rm TR})^2\ .
\]

\item Equivalently, in terms of $\NslopeT$, where $\NslopeT=d\mathcal N/dE\big|_{\eps=0^+}$, the transmission-resonance induced nonlinear bound state  is stable if:
\[
  \NslopeT < -\tfrac{2}{\gs}-\tfrac{8\gs(\mathcal B_\star^{\rm TR})^2}{3}
  \quad\text{or}\quad
  -\tfrac{2}{\gs}<\NslopeT<0\ .
\]
and unstable if 
\[
  -\tfrac{2}{\gs}-\tfrac{8\gs(\mathcal B_\star^{\rm TR})^2}{3}<\NslopeT<-\tfrac{2}{\gs}
  \quad\text{or}\quad
  \NslopeT>0.
\]
\end{enumerate}
\end{cor}

\begin{cor}[Stability intervals: SR case]\label{cor:stability-SR}
\begin{enumerate}
\item The scattering-resonance induced nonlinear bound state $(E(\eps),\psi_\eps(\cdot;E(\eps)))$ is 
\[
\textrm{orbitally stable if}\qquad  \Omega_\star^{\rm SR} > \gs(\mathcal B_\star^{\rm SR})^2
  \qquad\text{and unstable if}\qquad
  \Omega_\star^{\rm SR} < \gs(\mathcal B_\star^{\rm SR})^2.
\]
\item Equivalently, in terms of $\NslopeS$, where $\NslopeS=d\mathcal N/dE\big|_{\eps=0^+}$, the scattering-resonance induced nonlinear bound state is
\[
\textrm{orbitally stable if}\qquad  
  -\tfrac{2}{\gs}<\NslopeS<0\qquad 
  \textrm{and unstable if}\qquad \NslopeS>0\ \text{ or }\ \NslopeS<-\tfrac{2}{\gs}.
\]

\end{enumerate}
\end{cor}

The proofs of Corollaries~\ref{cor:stability-TR} and~\ref{cor:stability-SR} are below in Section~\ref{ssec:proofof4243}.

\begin{rem}[Geometric perspective]\label{rem:geometric}
Consider the bifurcation curve which emanates from $(E_\star,\mathcal{N}_\star)$.
Since $\nu_\star$ is the limiting slope of the resonance-induced bifurcation curve at $(E_\star,\mathcal{N}_\star)$, the $\nu_\star$-stability and instability intervals displayed in Corollaries~\ref{cor:stability-TR} and~\ref{cor:stability-SR} induce a local partition of the $(E,\mathcal N)$-plane in a neighborhood of $(E_\star,\mathcal N_\star)$. In particular, the horizontal tangent $\nu_\star=0$ and the free-soliton tangent $\nu_\star=-2/\gs$ bound the free-soliton wedge. Entering this wedge, namely
\[
  -\frac{2}{\gs} < \nu_\star < 0,
\]
places the branch in the stable region $\Sreg_2^{\rm TR}$ in the TR case and in the stable region $\Sreg^{\rm SR}$ in the SR case. This geometric picture is highlighted in Section~\ref{ssec:stab-regimes}.
\end{rem}

%

\begin{cor}[Stability for small $\gs$; see Remark \ref{rem:gs-small}]\label{cor:small-gamma-stability}
Assume $\supp V\subseteq[-b,b]$ and $\|V\|_{L^1}\le M$. Let $(\gs, U_\star)$ be a resonance pair satisfying the hypotheses
\ref{H:positivity}--\ref{H:Omega} of Theorem~\ref{thm:intro-stability-sharp}.


\smallskip
\noindent\underline{\emph{SR case}}:\quad
There exists $\Gamma_{\rm SR}(b,M) > 0$, depending only on $b$ and $M$,
such that if $\gs < \Gamma_{\rm SR}(b,M)$, then $\psi_\eps$ is orbitally
stable for all $\eps > 0$ sufficiently small.

\smallskip
\noindent\underline{\emph{TR case}}:\quad
Suppose additionally that either
\begin{equation}\label{eq:TR-Vside-dichotomy}
  U_\star(b) \le 1,
  \quad\text{or}\quad
  U_\star(b)^2 - 1 \;\ge\; A_0 \quad\text{for some } A_0 > 0.
\end{equation}
In the first alternative, there exists $\Gamma_{\rm TR}(b,M) > 0$,
depending only on $b$ and $M$, such that $\gs < \Gamma_{\rm TR}(b,M)$
implies orbital stability of $\psi_\eps$ for all $\eps > 0$ sufficiently
small. In the second alternative, the same conclusion holds with
$\Gamma_{\rm TR}$ replaced by $\Gamma_{\rm TR}(b,M,A_0)$.
\end{cor}
\begin{rem}[Proof of Corollary~\ref{cor:small-gamma-stability}]
A full proof, including explicit formulas for $\Gamma_{\rm SR}$ and $\Gamma_{\rm TR}$, is given in~\cite{turner2026columbia-thesis}. The key point is that as $\gs\to0^+$ the positive $\gs^{-1}$ terms dominate the order-one contributions in $\Omega_\star$ and $\nu_\star$. In the SR case this yields $\Omega_\star^{\rm SR}>\gs(\mathcal B_\star^{\rm SR})^2$, while in the TR case, under the dichotomy~\eqref{eq:TR-Vside-dichotomy}, it yields either $0<\Omega_\star^{\rm TR}<\tfrac{3}{4\gs}$ or $\Omega_\star^{\rm TR}>\tfrac{3}{4\gs}+\gs(\mathcal B_\star^{\rm TR})^2$. Hence the branch enters a stable region of Corollaries~\ref{cor:stability-TR} and~\ref{cor:stability-SR}.
\end{rem}
In Section~\ref{ssec:cor-nonvacuity}  we display explicit families (half-line SR square well, and full-line TR $V_\flat$, asymmetric) which satisfy the hypotheses of Corollary \ref{cor:small-gamma-stability}. A 
perspective on global minimizers of $\mathcal H$ at fixed mass is developed in the companion paper~\cite{tw-variational} and treated in full in~\cite{turner2026columbia-thesis}. That material has been deferred from this paper.

\subsection{Proof of Corollaries~\ref{cor:stability-TR} and~\ref{cor:stability-SR}}\label{ssec:proofof4243}

\begin{proof}
Setting $\widetilde\Omega_\star := \Omega_\star - \tfrac{3}{4\gs}$ in the TR case (or $\widetilde\Omega_\star := \Omega_\star$ in the SR case), the slope functional~\eqref{eq:nu} reads
\begin{equation}
  \Nslope \;=\; f(\widetilde\Omega_\star),
  \qquad
  f(x) := \frac{2(\mathcal B_\star)^2}{x} - \frac{2}{\gs}.
\end{equation}
Figure~\ref{fig:f-stability} shows a plot of  $x\mapsto f(x)$ in the TR (panel a) and SR (panel b) cases. By Theorem~\ref{thm:intro-stability-sharp}, stability occurs exactly where $\Omega_\star>0$ and $\Nslope<0$, and instability elsewhere. Marking those sign conditions on the plotted graph yields the labeled regions (four in the TR case, two in the SR case). The complete tables of $\Nslope$-intervals, broken out by case, are:

\noindent\textbf{Transmission resonance (TR):} $\Sreg$ and $\Ureg$ denote stable and unstable regions, respectively.
\begin{center}
\footnotesize
\normalfont
\renewcommand{\arraystretch}{1.35}
\setlength{\tabcolsep}{4pt}
\begin{tabular}{|>{\centering\arraybackslash}m{0.18\linewidth}|>{\centering\arraybackslash}m{0.58\linewidth}|}
\hline
\rowcolor{headerbg}
\textbf{Region} & $\nu_\star^{\rm TR}$ interval \\
\hline
\rowcolor{unstablebg}
$\Ureg_1^{\rm TR}$ & $-\tfrac{2}{\gs} - \tfrac{8\gs(\mathcal B_\star^{\rm TR})^2}{3} < \nu_\star^{\rm TR} < -\tfrac{2}{\gs}$ \\
\hline
\rowcolor{stablebg}
$\Sreg_1^{\rm TR}$ & $\nu_\star^{\rm TR} < -\tfrac{2}{\gs} - \tfrac{8\gs(\mathcal B_\star^{\rm TR})^2}{3}$ \\
\hline
\rowcolor{unstablebg}
$\Ureg_2^{\rm TR}$ & $\nu_\star^{\rm TR} > 0$ \\
\hline
\rowcolor{stablebg}
$\Sreg_2^{\rm TR}$ & $-\tfrac{2}{\gs} < \nu_\star^{\rm TR} < 0$ \\
\hline
\end{tabular}
\end{center}
\bigskip

\medskip

\noindent\textbf{Scattering resonance (SR):}

\begin{center}
\footnotesize
\normalfont
\renewcommand{\arraystretch}{1.35}
\setlength{\tabcolsep}{4pt}
\begin{tabular}{|>{\centering\arraybackslash}m{0.18\linewidth}|>{\centering\arraybackslash}m{0.58\linewidth}|}
\hline
\rowcolor{headerbg}
\textbf{Region} & $\nu_\star^{\rm SR}$ interval \\
\hline
\rowcolor{unstablebg}
$\Ureg^{\rm SR}$ & $\nu_\star^{\rm SR} > 0$ \ or\ $\nu_\star^{\rm SR} < -\tfrac{2}{\gs}$ \\
\hline
\rowcolor{stablebg}
$\Sreg^{\rm SR}$ & $-\tfrac{2}{\gs} < \nu_\star^{\rm SR} < 0$ \\
\hline
\end{tabular}
\end{center}
\noindent The endpoint $-2/\gs$ is the slope of the full-line free-soliton mass curve $\mathcal N[\mathcal S(\,\cdot\,;E)]$ at $E = E_\star$; see Remark~\ref{rem:geometric}. In the transition cases $\Omega_\star=0$ or $\Nslope=0$ the leading-order criterion is inconclusive.\medskip
\end{proof}

\begin{figure}[htbp]
\centering
\begin{subfigure}[t]{0.48\textwidth}
\centering
\begin{tikzpicture}
\begin{axis}[
  width=\linewidth, height=6.5cm,
  axis lines=middle,
  xlabel={$\Omega_\star$}, ylabel={$\Nslope$},
  xlabel style={anchor=west},
  ylabel style={anchor=south},
  xmin=-1.5, xmax=4.2,
  ymin=-8, ymax=4.5,
  xtick={0, 1.75},
  xticklabels={$0$, $\tfrac{3}{4\gs}{+}\gs\mathcal B_\star^2$},
  ytick={0},
  yticklabels={$0$},
  extra y ticks={-4.667},
  extra y tick labels={$-\tfrac{2}{\gs}{-}\tfrac{8\gs\mathcal B_\star^2}{3}$},
  tick label style={font=\footnotesize, text=black!55},
  major tick length=4.5pt,
  tick style={line width=0.7pt, black},
  samples=200,
  clip=false,
  restrict y to domain=-8:4.5,
  legend style={
    at={(0.97,0.97)},
    anchor=north east,
    font=\footnotesize,
    draw=black!50,
    fill=white,
    legend cell align=left,
  },
]
  \addplot[dashed, gray, domain=-1.5:4.2, samples=2, forget plot] {-2};
  \node[anchor=east, font=\footnotesize, text=black!55, inner sep=3pt] at (axis cs:-1.5,-2) {$-\tfrac{2}{\gs}$};
  \draw[dashed, gray] (axis cs:0.75,-8) -- (axis cs:0.75,4.5);
  \node[anchor=south, font=\footnotesize, text=black!55, inner sep=3pt] at (axis cs:0.75,4.5) {$\tfrac{3}{4\gs}$};
  \fill (axis cs:0,0) circle (1.5pt);
  \fill (axis cs:0.75,0) circle (1.5pt);
  \fill (axis cs:1.75,0) circle (1.5pt);

  \addplot[green!60!black, very thick, domain=0.02:0.72] {2/(x-0.75) - 2};
  \addlegendentry{stable}
  \addplot[red, dashed, very thick, domain=-1.5:-0.02] {2/(x-0.75) - 2};
  \addlegendentry{unstable}
  \addplot[red, dashed, very thick, domain=0.78:1.749, forget plot] {2/(x-0.75) - 2};
  \addplot[green!60!black, very thick, domain=1.751:4.2, forget plot] {2/(x-0.75) - 2};

  \node[font=\small, fill=white, inner sep=1.5pt] at (axis cs:-1.05, -3.07) {$\Ureg_1$};
  \node[font=\small, fill=white, inner sep=1.5pt] at (axis cs: 0.36, -6.90) {$\Sreg_1$};
  \node[font=\small, fill=white, inner sep=1.5pt] at (axis cs: 1.20, 2.44) {$\Ureg_2$};
  \node[font=\small, fill=white, inner sep=1.5pt] at (axis cs: 3.00, -1.11) {$\Sreg_2$};
\end{axis}
\end{tikzpicture}
\subcaption{Transmission-resonance (TR) case. The graph of $\nu_\star$ has a vertical asymptote at $\Omega_\star=\tfrac{3}{4\gs}$. Reading the $\Omega_\star$-axis from left to right gives $\Ureg_1^{\rm TR},\Sreg_1^{\rm TR},\Ureg_2^{\rm TR},\Sreg_2^{\rm TR}$.}\label{fig:f-stability-TR}
\end{subfigure}\hfill
\begin{subfigure}[t]{0.48\textwidth}
\centering
\begin{tikzpicture}
\begin{axis}[
  width=\linewidth, height=6.5cm,
  axis lines=middle,
  xlabel={$\Omega_\star$}, ylabel={$\Nslope$},
  xlabel style={anchor=west},
  ylabel style={anchor=south},
  xmin=-2, xmax=4.2,
  ymin=-8, ymax=4.5,
  xtick={1},
  xticklabels={$\gs\mathcal B_\star^2$},
  ytick={0},
  yticklabels={$0$},
  tick label style={font=\footnotesize, text=black!55},
  major tick length=4.5pt,
  tick style={line width=0.7pt, black},
  samples=200,
  clip=false,
  restrict y to domain=-8:4.5,
  legend style={
    at={(0.97,0.97)},
    anchor=north east,
    font=\footnotesize,
    draw=black!50,
    fill=white,
    legend cell align=left,
  },
]
  \addplot[dashed, gray, domain=-2:4.2, samples=2, forget plot] {-2};
  \node[anchor=east, font=\footnotesize, text=black!55, inner sep=3pt] at (axis cs:-2,-2) {$-\tfrac{2}{\gs}$};
  \fill (axis cs:0,0) circle (1.5pt);
  \fill (axis cs:1,0) circle (1.5pt);

  \addplot[green!60!black, very thick, domain=1.001:4.2] {2/x - 2};
  \addlegendentry{stable}
  \addplot[red, dashed, very thick, domain=-2:-0.02] {2/x - 2};
  \addlegendentry{unstable}
  \addplot[red, dashed, very thick, domain=0.02:0.999, forget plot] {2/x - 2};

  \node[font=\small, fill=white, inner sep=1.5pt] at (axis cs:-1.00, -4.00) {$\Ureg^{\rm SR}$};
  \node[font=\small, fill=white, inner sep=1.5pt] at (axis cs: 0.50, 2.00) {$\Ureg^{\rm SR}$};
  \node[font=\small, fill=white, inner sep=1.5pt] at (axis cs: 2.60, -1.23) {$\Sreg^{\rm SR}$};
\end{axis}
\end{tikzpicture}
\subcaption{Scattering-resonance (SR) case. The vertical asymptote is $\Omega_\star=0$. There are only two regions: the unstable region $\Ureg^{\rm SR}$ (a single $\Omega_\star$-interval, covered by the $2$ distinct branches of $f$) and the stable region $\Sreg^{\rm SR}$.}\label{fig:f-stability-SR}
\end{subfigure}
\caption{Schematic plots of $\nu_\star$ as a function of $\Omega_\star$ in the TR and SR formulas, shown here with $\gs=1$ and $\mathcal B_\star=1$. Green curve segments satisfy the stability sign conditions of Theorem~\ref{thm:intro-stability-sharp}, namely $\Omega_\star>0$ and $\nu_\star<0$, while red segments violate at least one of them. Each segment is labeled by the region $\Sreg,\Ureg$ of Corollaries~\ref{cor:stability-TR} and~\ref{cor:stability-SR} it represents (four in the TR case, two in the SR case).}\label{fig:f-stability}
\end{figure}
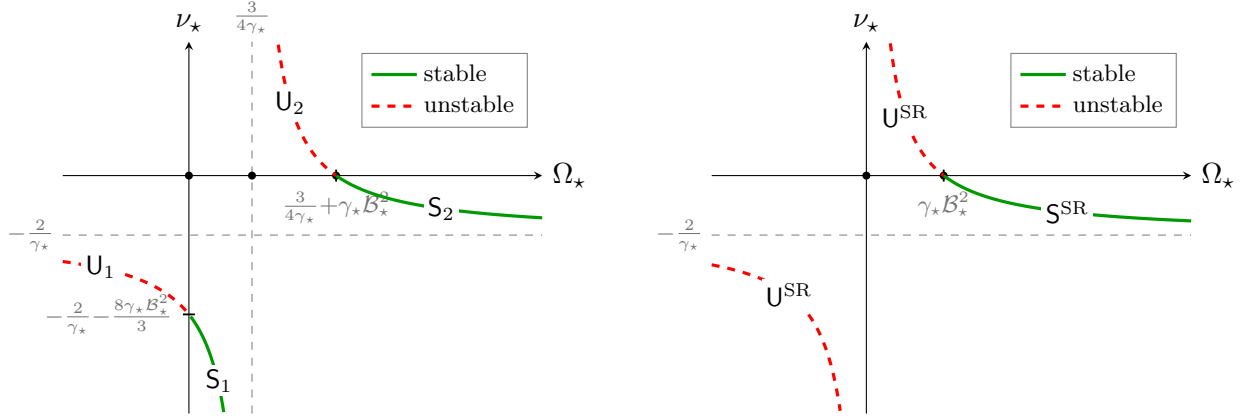


\subsection{Strategy for the proof of Theorem~\ref{thm:intro-stability-sharp}}
The proof of Theorem~\ref{thm:intro-stability-sharp} makes use of the following three propositions in order to verify the hypotheses of Theorem~\ref{thm:W-GSS}.

\begin{prop}[Perturbed ``translation mode'']\label{prop:mu}
For $\eps > 0$ sufficiently small, the eigenvalue problem
\begin{equation} \Lop\phi = \mu \phi,\quad \phi\in  X\end{equation} has exactly one real eigenvalue $\mu(\eps)$ in a neighborhood of zero. Its corresponding  eigenfunction $\phi(\,\cdot\,;\eps) \in X$ has the following spatial structure:
\begin{enumerate}[leftmargin=2em, label=\textup{(\roman*)}]
  \item For both the TR and SR cases and   $| x | \leq b$, $\phi(x) = U_\star(x) + h(x)$, where  $\|h\|_{C^1} = \mathcal O(\eps)$.
  \item  For both the TR and SR cases and $x > b$, $\phi$ is a scalar multiple of the $\ell=2$ P\"oschl--Teller Jost solution $f_+$ of Lemma~\ref{lem:Jost}, shifted by $y_R(\eps)$ . 
  \\  For the TR case, and $x < -b$, $\phi$   is  a multiple of $f_-$ shifted by $y_L(\eps)$.
  \item The perturbed ``translation eigenvalue''  has the following behavior for $\eps>0$ and small:
  \begin{equation}\label{eq:mu-final}
    \mu(\eps) \;=\; \frac{3\eps^2}{4\gs}\,\Omega_\star + \mathcal O(\eps^3),
  \end{equation} 
  where $\Omega_\star$ is the expression in  \eqref{eq:Omega}.
\end{enumerate}
\end{prop}

\begin{prop}[Point spectrum of $\Lop$]\label{prop:Lp-count}
  Assume $U_\star\geq0$. For $\eps > 0$ sufficiently small, 
  \begin{enumerate}
      \item $\Lop$ has exactly one negative eigenvalue which is $\mathcal{O}(\eps)$-distant from $\mu_0^{\rm sol}:=-3\gs^2$, the ground state eigenvalue of $L^{\rm sol}_+$.
      \item  $\mu(\eps)$, whose small $\eps$ asymptotics in \eqref{eq:mu-final}, is the second eigenvalue of $\Lop$. Thus,
      \item  $n_-(\Lop) = 1$ if and only if $\mu(\eps) > 0$, that is if $\Omega_\star>0$;
       recall $\gs>0$ by \eqref{eq:assume-gs-pos}.
  \end{enumerate}  
\end{prop}

\begin{prop}[Slope of $\mathcal{N}$ vs. $E$]\label{prop:slope}
  For sufficiently small $\eps > 0$,
  \begin{equation}\label{eq:slope-asymp}
    \frac{d\mathcal N}{dE}(\eps) \;=\; \Nslope + \mathcal O(\eps),
  \end{equation}

  where $\Nslope$ is given by the expression in \eqref{eq:nu}.
\end{prop}

\subsection*{Proof of Theorem~\ref{thm:intro-stability-sharp}}\label{ssec:proof-main-thm}
\smallskip
\noindent\emph{TR case.}
By hypothesis~\ref{H:positivity}, $U_\star\ge0$ on $[-b,b]$ and
\[
  U_\star'' \;=\; (V+\gs^2)U_\star
\]
there. A nontrivial nonnegative solution of a second-order linear ODE cannot have an interior zero: if $U_\star(x_0)=0$ for some $x_0\in(-b,b)$, then $x_0$ is a local minimum, so $U_\star'(x_0)=0$, and uniqueness for the ODE would force $U_\star\equiv0$, contradicting the normalization $U_\star(-b)=1$. Hence $U_\star>0$ on $[-b,b]$. By Theorem~\ref{thm:bifurcation}, on $[-b,b]$ we have $\psi_\eps=\sqrt{\eps}\,[U_\star+\mathcal O(\eps)]$, so $\psi_\eps>0$ there for $\eps>0$ sufficiently small. On the two exterior regions $(-\infty,-b)$ and $(b,\infty)$, the same theorem identifies $\psi_\eps$ with translates of the free soliton $\Scal(\,\cdot\,;E(\eps))$, which is strictly positive. Therefore $\psi_\eps>0$ on all of $\R$.

Since $L_-^\eps\psi_\eps=0$ and $\psi_\eps$ is nodeless, Sturm oscillation implies that $\psi_\eps$ is the ground-state eigenfunction of $L_-^\eps$. Thus
\[
  L_-^\eps \;\ge\; 0,\qquad \ker L_-^\eps \;=\; \spn\{\psi_\eps\},\qquad n_-(L_-^\eps) \;=\; 0.
\]
Turning to $L_+^\eps$, Proposition~\ref{prop:Lp-count} shows that for $\eps>0$ small its point spectrum consists of the negative soliton ground-state eigenvalue near $-3\gs^2$, together with the small eigenvalue $\mu(\eps)$ from Proposition~\ref{prop:mu}, and that there are no other eigenvalues between them. Consequently,
\[
  n_-(L_+^\eps) \;=\; 1 \iff \mu(\eps) \;>\; 0,
\]
while $\mu(\eps) \;<\; 0$ gives $n_-(L_+^\eps) \;=\; 2$. Proposition~\ref{prop:mu} gives
\[
  \sgn \mu(\eps) \;=\; \sgn \Omega_\star
\]
for $\eps>0$ sufficiently small, and Proposition~\ref{prop:slope} gives
\[
  \sgn\!\left(\frac{d\mathcal N}{dE}(\eps)\right) \;=\; \sgn \Nslope .
\]
Therefore:
\begin{itemize}
\item If $\Omega_\star \;>\; 0$ and $\Nslope \;<\; 0$, then $n_-(L_+^\eps) \;=\; 1$, $L_-^\eps \;\ge\; 0$ with $\ker L_-^\eps \;=\; \spn\{\psi_\eps\}$, and $\frac{d\mathcal N}{dE}(\eps) \;<\; 0$. The hypotheses of Theorem~\ref{thm:W-GSS}(i) hold, so $\psi_\eps$ is orbitally stable.
\item If $\Omega_\star \;<\; 0$, then $\mu(\eps) \;<\; 0$, so $n_-(L_+^\eps) \;=\; 2$ while $n_-(L_-^\eps) \;=\; 0$. Hence
\[
  n_-(L_+^\eps)-n_-(L_-^\eps) \;=\; 2 \;\notin\; \{0,1\},
\]
and Theorem~\ref{thm:W-GSS}(ii)(\ref{wgss-b}) yields instability.
\item If $\Nslope \;>\; 0$, then $\frac{d\mathcal N}{dE}(\eps) \;>\; 0$, and Theorem~\ref{thm:W-GSS}(ii)(\ref{wgss-a}) yields instability.
\end{itemize}

\smallskip
\noindent\emph{SR case.}
The argument is the same, with the left exterior replaced by the Dirichlet endpoint $x=0$. Hypothesis~\ref{H:positivity} gives $U_\star\ge0$ on $[0,b]$, and the boundary conditions $U_\star(0)=0$, $U_\star'(0)=1$ imply $U_\star>0$ on $(0,b]$. Theorem~\ref{thm:bifurcation} then gives $\psi_\eps>0$ on $(0,\infty)$, with $\psi_\eps(0)=0$. Hence $\psi_\eps$ is again the ground state of $L_-^\eps$ on $X=H_0^1(\R_+)$, so $n_-(L_-^\eps)=0$. Proposition~\ref{prop:Lp-count} reduces the count of negative eigenvalues of $L_+^\eps$ to the sign of $\mu(\eps)$, Proposition~\ref{prop:mu} identifies the sign of $\mu(\eps)$ with the sign of $\Omega_\star$, and Proposition~\ref{prop:slope} identifies the sign of $\frac{d\mathcal N}{dE}$ with the sign of $\Nslope$. The same three applications of Theorem~\ref{thm:W-GSS} therefore prove Theorem~\ref{thm:intro-stability-sharp} in the SR case as well.
\qed



\section{Proofs of Propositions~\ref{prop:mu},~\ref{prop:Lp-count} and ~\ref{prop:slope}}\label{sec:proofs}

We present the proofs in detail for the transmission resonance (TR) case; hence the support of $V(x)$ is contained in $-b\le x\le b$. In Remark~\ref{rem:SR-stability} we indicate the modifications that are required to adapt the argument to the scattering resonance (SR) case.

\subsection{The perturbed translation eigenvalue; proof of
Proposition~\ref{prop:mu}}\label{ssec:proof-mu}
Recall $\Lop$, given in ~\eqref{eq:Lpm-eps}, is the linearized operator about the resonance-induced nonlinear bound state branch $\eps\mapsto (E(\eps),\psi_\eps(\cdot;E(\eps)))$; see Theorem~\ref{thm:bifurcation}.

For $\eps$ small and positive, we seek solutions of the eigenvalue problem 
\begin{equation}\label{eq:X-evp}
    \Lop\phi = \mu\ \phi,\quad \phi\in X=H^1(\R),\quad \textrm{with spectral parameter $\mu$ near zero.}
\end{equation}

By parts (ii) and (iii) of Theorem~\ref{thm:bifurcation}, the transmission (TR) resonance-induced nonlinear bound state, $(\psi_\varepsilon,E(\varepsilon))$, has the form:
\begin{align}\label{eq:psi-piecewise}
  \psi_\eps(x) &= \begin{cases}
    \Scal\bigl(x - y_L(\eps);\,E(\eps)\bigr) & x < -b, \\[2pt]
    \sqrt\eps\bigl[\,U_\star(x) + \mathcal O(\eps)\bigr] & -b\le x\le b, \\[2pt]
    \Scal\bigl(x - y_R(\eps);\,E(\eps)\bigr) & x > b,
  \end{cases}\\
  E(\eps) &= -\gs^2 + \eps\dot E + \mathcal O(\eps^2)\ . \label{eq:E-expand}
\end{align}
Here, $\mathcal{S}(\cdot;E)$ is given by \eqref{eq:free-soliton-intro}, $\dot E$ by~\eqref{eq:Edot-solved}, and the free soliton centerings $y_{R}(\eps)$ and $y_{L}(\eps)$ by~\eqref{eq:yRL-tanh} and~\eqref{eq:yL-tanh}. Note, for the (TR) case, both $y_{R}(\eps)$ and $y_{L}(\eps) \to +\infty$ as $\eps \to 0^+$.

For $x<-b$ and $x>b$, $\psi_\eps$ is equal to a translated free soliton, and hence the potential in the linearized operator $\Lop = -\partial_x^2 + V - 3\psi_\eps^2 - E(\eps)$ is easily seen to be equal to  a P\"oschl--Teller potential (Section~\ref{ssec:PT}) with parameter $\ell=2$ and centerings at:  $x=y_L(\eps)$, for $x<-b$, and $x=y_R(\eps)$, for $x>b$. Further, on the interval $-b\le x\le b$, $\Lop$ is equal to $H_V + \gs^2$ plus  an order $\eps$ correction:
\begin{equation}\label{eq:Lop-pieces}
  \Lop \;=\;
  \begin{cases}
    -\partial_x^2 - 6\la_0(\eps)^2\sech^2\!\bigl(\la_0(\eps)(x - y_L(\eps))\bigr) + \la_0(\eps)^2
      & x < -b, \\[4pt]
    H_V -3\psi_\eps^2(x) +\la_0^2(\eps)
      & -b\le x\le b, \\[4pt]
    -\partial_x^2 - 6\la_0(\eps)^2\sech^2\!\bigl(\la_0(\eps)(x - y_R(\eps))\bigr) + \la_0(\eps)^2
      & x > b\ ,
  \end{cases}
\end{equation}
where, by \eqref{eq:E-expand}, 
\begin{equation}\label{eq:la0}
  \la_0(\eps) \;:=\; \sqrt{|E(\eps)|} \;=\; \gs - \frac{\eps\dot E}{2\gs} + \mathcal O(\eps^2)\ .
\end{equation}
 On the interior interval, $-b\le x\le b$, we substitute $\psi_\eps^2(x) = \eps U_\star(x)^2 + \mathcal O(\eps^2)$ from~\eqref{eq:psi-piecewise} and $\la_0^2(\eps) = \gs^2 - \eps\dot E + \mathcal O(\eps^2)$ from~\eqref{eq:la0} to obtain 
\begin{equation}\label{eq:Lop-interior-expanded}
  \Lop\Big|_{[-b,b]} \;=\; (H_V + \gs^2) - \eps\bigl(3U_\star^2 + \dot E\bigr) + \mathcal O(\eps^2)\ .
\end{equation}
Hence, on $[-b,b]$, $\Lop$ is a small perturbation of the operator $H_V + \gs^2$, whose kernel is spanned by the resonance eigenmode $U_\star$.

Noting now that for $|x|>b$, we have that $\Lop$ is equal to $H_{\rm PT}(\la_0(\eps),2,y) + \la_0(\eps)^2$ with $y=y_{R}(\eps)$ or $y_{L}(\eps)$;  see \eqref{eq:PTpot_y}. So the eigenvalue problem $\Lop\phi=\mu\phi$ with $\phi\in L^2(\R)$ takes the form
\begin{subequations}
\label{eq:evp-mu}
\begin{align}
    H_{\rm PT}\big(\la_0(\eps),2,y_L(\eps)\big)\ \phi &= -\lambda^2(\eps,\mu)\ \phi,\quad x<-b\ {\rm and}\ \phi(x)\to0\ \textrm{ as $x\to-\infty$}\label{eq:evp-mu-lt-}\\
   \big(\ H_V -3\psi_\eps^2(x)  \big)\  \phi &= -\lambda^2(\eps,\mu)\ \phi \quad -b\le x\le b\ {\rm and}\
    \textrm{$\phi(x)$, $\phi'(x)$ continuous at $x=\pm b$}\label{eq:evp-mu-lt0}\\
   H_{\rm PT}\big(\la_0(\eps),2,y_R(\eps)\big)\ \phi &= -\lambda^2(\eps,\mu)\ \phi,\ \quad x>b\ {\rm and}\ \phi(x)\to0\ \textrm{ as $x\to+\infty$}\ . \label{eq:evp-mu-lt+}
\end{align}
\end{subequations}
In \eqref{eq:evp-mu}, by \eqref{eq:la0}, we have
\begin{equation}\label{eq:la-expand}
  \la(\eps,\mu) \;:=\; +\sqrt{\la_0(\eps)^2-\mu} \;=\; \gs - \frac{\eps\dot E}{2\gs} - \frac{\mu}{2\gs} + \mathcal O(\eps^2+\mu^2)\ .
\end{equation}
Let us now solve the eigenvalue problem \eqref{eq:evp-mu} (equivalently \eqref{eq:X-evp}) for $\mu$ near zero.

In the region where $|x|>b$, we use the Jost solutions for  P\"oschl--Teller potentials ($f_\pm$ in Lemma~\ref{lem:Jost}), and impose the spatial decay  conditions in~\eqref{eq:evp-mu-lt-} and \eqref{eq:evp-mu-lt+}:
\begin{equation}\label{eq:phi-cases}
  \phi(x) = \begin{cases}
    c_-\,f_-\!\bigl(x - y_L(\eps);\,i\la(\eps,\mu);\,\la_0(\eps)\bigr) & x < -b, \\[4pt]
    c_+\,f_+\!\bigl(x - y_R(\eps);\,i\la(\eps,\mu);\,\la_0(\eps)\bigr) & x > b.
  \end{cases}
\end{equation}

Next, to construct the solution on the interval $-b\le x\le b$, we let
\begin{equation} \phi(x)=\phi_{\rm in}(x)=\phi_{\rm in}(x;\mu,\eps,\alpha)\label{eq:phi_in}
\end{equation}
denote the unique solution of the initial-value problem with initial data prescribed at $x=-b$:
\begin{subequations}
\label{eq:phi-ivp}
\begin{align}
  \Lop\phi_{\rm in} &= \mu\phi_{\rm in} \quad  -b<x<b \label{eq:phi-ivp-1}\\
\phi_{\rm in}(-b) &= 1,\quad \phi_{\rm in}'(-b) = \gs + \alpha\ .\label{eq:phi-ivp-2}
\end{align}
\end{subequations}
Using \eqref{eq:Lop-interior-expanded} we see that   $\phi_{\rm in}(x;0,0,0)=U_\star(x)$ and that the mapping $(\mu,\eps,\alpha)\mapsto \phi_{\rm in}(x;\mu,\eps,\alpha)$ is real analytic in a neighborhood of $(0,0,0)$.

Summarizing the previous discussion, we have the following characterization of solutions of the eigenpairs $(\mu,\phi)$ for $\mu$ near zero:
\begin{prop}
The pair $(\phi(x),\mu)$, with $\mu\in\R$ small and $\phi$ defined by \eqref{eq:phi-cases} for $|x|>b$ and by \eqref{eq:phi-ivp} for $|x|\le b$, is a solution of the  $X=H^1(\mathbb R)$ eigenvalue problem \eqref{eq:X-evp} for $\Lop$ if and only if $[\phi]_{x=\pm b}=0$ and $[\phi']_{x=\pm b}=0$. 

Equivalently, $(\mu, c_-, c_+, \alpha)$ is such that the following matching conditions at $x=\pm b$ hold:
\begin{subequations}\label{eq:continuity}
\begin{align}
  c_-\,f_-\!\bigl(-b-y_L(\eps);\,i\la;\,\la_0\bigr)  &= \phi_{\rm in}(-b) \;=\; 1,                  \label{eq:continuity-Lval}\\
  c_-\,f_-'\!\bigl(-b-y_L(\eps);\,i\la;\,\la_0\bigr) &= \phi_{\rm in}'(-b) \;=\; \gs + \alpha,       \label{eq:continuity-Lder}\\
  c_+\,f_+\!\bigl(b-y_R(\eps);\,i\la;\,\la_0\bigr)   &= \phi_{\rm in}(b),                            \label{eq:continuity-Rval}\\
  c_+\,f_+'\!\bigl(b-y_R(\eps);\,i\la;\,\la_0\bigr)  &= \phi_{\rm in}'(b).                           \label{eq:continuity-Rder}
\end{align}
\end{subequations}
In \eqref{eq:continuity}, $\lambda=\lambda(\mu,\alpha)$, $\lambda_0=\lambda_0(\eps)$, and $\phi_{\rm in}=\phi_{\rm in}(\alpha,\mu,\eps)$.
\end{prop}

The system~\eqref{eq:continuity} consists of four non-homogeneous equations for the four unknowns $c_-,\,c_+,\,\alpha,\,\mu$. For $\eps$ positive and small, we solve this system  for a unique eigenvalue $\mu(\eps)$ near zero,  together with $c_\mp(\eps)$ and $\alpha(\eps)$. These in turn determine the eigenfunction $\phi_\eps(x)$ on all $\R$ via the piecewise representation~\eqref{eq:phi-cases},~\eqref{eq:phi_in}.

The proof of Proposition~\ref{prop:mu} is in two steps.
\begin{itemize}
  \item \textbf{Step~1.} Assume that $\mu$ is an eigenvalue for $0<\eps,|\mu|<C_1$. We prove that $D(\eps,\mu)=0$, where $(\eps,\mu)\mapsto D(\eps,\mu)$ is a scalar function defined and analytic on the open domain $\mathcal{Q}=\{(\eps,\mu):|\eps|<C_1,\,|\mu|<C_1\}$ of $(0,0)\in\R^2$.
  \item \textbf{Step~2.} We prove that for $:|\eps|<C_1$, the equation $D(\eps,\mu)=0$ has a unique solution curve $\eps\mapsto\mu(\eps)$ in $\mathcal{Q}$. Moreover, for each $\eps$ such that $0<\eps<C_1$, $\mu(\eps)$ is an eigenvalue of $\Lop$ with corresponding eigenfunction $\phi_\eps\in X$ reconstructed from  $c_\mp(\eps)$, $\alpha(\eps)$.
\end{itemize}

\bigskip
\noindent\textbf{Implementation of Step~1.} Suppose   for $0<\eps<C_1$ that $\mu$, with $|\mu|<C_1$, a sufficiently small constant, is an eigenvalue of \eqref{eq:X-evp}, equivalently \eqref{eq:evp-mu}.
%
%

\noindent\textbf{Claim (1a).}\claimlabel{1a} The continuity conditions at $x=-b$:
\[ [\phi]_{x=-b}=0\quad \textrm{ and}\quad  [\phi']_{x=-b}=0\] imply the relations:
 \begin{align}
 \label{eq:fminus-ne0}
 &f_-\!\bigl(-b-y_L(\eps);\,i\la;\,\la_0\bigr) \ne 0\\
  &c_-(\eps,\mu) \;=\; \frac{1}{f_-\!\bigl(-b-y_L(\eps);\,i\la;\,\la_0\bigr)},
 \label{eq:c-minus}\end{align}
Furthermore, $(\mu,\eps)\mapsto \alpha(\mu,\eps)$ is analytic in a neighborhood of $(0,0)$, with
\begin{equation}\label{eq:alpha}
  \alpha \;=\; \alpha(\eps,\mu) \;=\; -\frac{\eps\,\dot E + \mu}{2\gs} + \mathcal O(\eps^2 + \mu^2) \;=\; \mathcal O(\eps + |\mu|).
\end{equation}
It follows from the analyticity of 
$(\mu,\alpha,\eps)\mapsto \phi_{\rm in}(x;\mu,\alpha,\eps)$, that 
\[ (\mu,\alpha,\eps)\mapsto \phi_{\rm in}(x,\eps,\mu):=\phi_{\rm in}(x;\mu,\eps,\alpha(\mu,\eps))\] 
is analytic in an $(\mu,\eps)$ neighborhood of $(0,0)$.
\begin{proof}
We first prove the non-vanishing property \eqref{eq:fminus-ne0}. We use Lemma~\ref{lem:Jost}; first, by \eqref{eq:f_pm},  $f_-(\xi)=f_+(-\xi)$,   and  hence by \eqref{eq:Jost} and \eqref{eq:p-def}:
\begin{equation}
  f_-(-b-y_L;\,i\la;\,\la_0) \;=\; e^{\la(-b-y_L)}\,\frac{p(m,-t_L)}{(m+1)(m+2)},
\end{equation}
where   \[ p(m,t):=m^2-1+3mt+3t^2\quad {\rm and}\quad m = m(\mu,\eps) := \la(\mu,\eps)/\la_0(\eps).\] 
 By \eqref{eq:m_and_t} and ~\eqref{eq:yL-tanh}, for $\eps$ small
\begin{equation}
  t_L(\eps) := t(\xi)\big|_{\xi=-b-y_L(\eps)} 
  = -\sqrt{1 - \tfrac{\eps+\mathcal{O}(\eps^2)}{2\la_0^2(\eps)}} = -1 + \frac{1}{4\gs^2}\,\eps + \mathcal O(\eps^2) \ .
\label{eq:1ptL} \end{equation}
Further, ~\eqref{eq:la0} and ~\eqref{eq:la-expand} imply
\begin{equation}\label{eq:m-expand-b}
  m(\mu,\eps) \;=\; 1 - \frac{\mu}{2\gs^2} + \mathcal O(\eps\mu + \mu^2).
\end{equation}
Substituting into the expression for $p(m,t)$ yields
\begin{equation}\label{eq:wout-den-L}
  p(m, -t_L) \;=\; p(1,1) + \mathcal O(\eps + \mu) = 6 + \mathcal O(\eps + \mu).
\end{equation}
This proves the non-vanishing property \eqref{eq:fminus-ne0} for  $0<\eps,|\mu|<C_1$ sufficiently small. 

Now using $\phi_{\rm in}(-b) = 1$ (eqn \eqref{eq:phi-ivp-2}) and \eqref{eq:fminus-ne0}, we have that $[\phi]_{x=-b}$ (eqn \eqref{eq:continuity-Lval}) implies the expression \eqref{eq:c-minus} for $c_-(\eps,\mu)$. Next, substitute this expression into the condition $[\phi']_{x=-b}$ (eqn \eqref{eq:continuity-Lder}) and we obtain
\begin{equation}\label{eq:wout-def}
    \woutL(\eps,\mu)   := \frac{f_-'\!\bigl(-b-y_L;\,i\la(\eps,\mu);\,\la_0(\eps)\bigr)}{f_-\!\bigl(-b-y_L;\,i\la(\eps,\mu);\,\la_0(\eps)\bigr)}\ = \gs + \alpha.
\end{equation}
The relation \eqref{eq:wout-def} determines $\alpha$ as a function of $\eps$ and $\mu$.
To obtain an expansion of  $\alpha(\eps,\mu)$ observe from~\eqref{eq:Jost} that
\begin{equation}\label{eq:outerlogform}
  \frac{f_-'(\xi;i\la;\la_0)}{f_-(\xi;i\la;\la_0)}
  \;=\; \la + \la_0\big(1-t^2(\xi)\big)\,\frac{-3m + 6t(\xi)}{m^2 - 1 - 3mt(\xi) + 3t^2(\xi)}. 
\end{equation}
Evaluating \eqref{eq:outerlogform} at  $\xi = -b - y_L(\eps)$  we obtain using  ~\eqref{eq:la0}, ~\eqref{eq:la-expand},~\eqref{eq:1ptL} and~\eqref{eq:m-expand-b}:
\begin{equation}\label{eq:wout-mb}
  \woutL(\eps,\mu) \;=\; \gs - \frac{\eps\,\dot E}{2\gs} - \frac{\mu}{2\gs} + \mathcal O(\eps^2 + \mu^2).
\end{equation}
Finally, substituting \eqref{eq:wout-mb} into \eqref{eq:wout-def} yields the expansion of $\alpha(\eps,\mu)$ in \eqref{eq:alpha}. 
 This completes the proof of \claimref{1a}.
\end{proof}

We continue under the assumption that $\mu$ is an eigenvalue for $0<\eps<C_1$ and $|\mu|<C_1$, where $C_1$ is a sufficiently small constant.

\noindent\textbf{Claim (1b).}\claimlabel{1b}
The continuity conditions at $x=+b$ \[ [\phi]_{x=+b}=0\quad  \textrm{and}\quad  [\phi']_{x=+b}=0
\] imply the relations: 
\begin{subequations}
\label{eq:c+f+}
\begin{align}
\label{eq:f+ne0}
&f_+(b-y_R(\eps);i\la(\eps,\mu),\la_0(\eps))\neq 0\\
&{\ }\nonumber\\
&  c_+(\eps,\mu) \;=\; \frac{\phi_{\rm in}\big(b;\eps,\mu\big)}{f_+\!\big(b-y_R(\eps);i\la(\eps,\mu),\la_0(\eps)\big)}\ \ne\ 0.
\label{eq:c-plus}
\end{align}
\end{subequations}
Here, $\phi_{\rm in}(x;\eps,\mu,\alpha)$ is the solution of the IVP \eqref{eq:phi-ivp} and $\alpha=\alpha(\eps,\mu)$ is given by \eqref{eq:alpha}.

\noindent\textbf{Claim (1c).}\claimlabel{1c}
There is a function $(\eps,\mu)\mapsto D(\eps,\mu)$ (see \eqref{eq:Ddef} below), defined and analytic for
\[ \mathcal{Q}:\ |\eps|<C_1,\ |\mu|<C_1,\]
with $C_1>0$ chosen sufficiently small, such that  
\begin{equation}\label{eq:D-parabolic}
  D(\eps,\mu) \;=\; -\,4\gs\,\mu \;+\; 3\,\Omega_\star\,\eps^2 \;+\; \mathcal O(\eps^3 + \eps\mu + \mu^2)\quad \textrm{for}\quad (\eps,\mu)\in\mathcal Q,
\end{equation}
and  
\[ \textrm{if $\mu$ is an eigenvalue, then $D(\eps,\mu)=0$}.\]

We next present the proofs of \claimref{1b} and \claimref{1c}.

{
\begin{proof}[Proof of \claimref{1b}]
We must verify the relations \eqref{eq:c+f+}. Suppose  $0<\eps,|\mu|<C_1$ are sufficiently small, as above.  
 We begin by writing 
\begin{align}\label{eq:h-expand}
\phi_{\rm in} &:= U_\star + h,\quad -b\le x\le b
\end{align}
Since $U_\star$ satisfies $(H_V+\gs^2)U_\star = 0$ with $U_\star(-b)=1$ and $U_\star'(-b)=\gs$ (see \eqref{eq:TR-BC}), we have
\begin{align}
 h(-b)&=0\quad {\rm and}\quad h'(-b)=\alpha(\eps,\mu)\ ,\label{eq:h-bc}
 \end{align}
 since $\phi_{\rm in}(-b;\eps,\mu)=1$ and $\phi_{\rm in}'(-b;\eps,\mu)=\gs+\alpha(\eps,\mu)$.
 Substituting \eqref{eq:h-expand} into~\eqref{eq:phi-ivp-1} and using the expansion~\eqref{eq:Lop-interior-expanded} of $\Lop$, we obtain that $h=h(x;\eps,\mu)$ solves the second order ODE 
\begin{equation}\label{eq:delta-eq}
  \big(H_V + \gs^2 \big)\,h \;=\; \bigl[\eps(3U_\star^2 + \dot E + O(\eps)\ ) + \mu\bigr]U_\star  + \mathcal{O}(|\eps| + |\mu|) h
\end{equation}
with initial conditions \eqref{eq:h-bc}.  We note that the mapping $(\eps,\mu)\mapsto h(x;\eps,\mu)$ is real-analytic in a neighborhood of $(0,0)$.

We claim that 
\begin{equation}\label{eq:hbound}
\|h(\cdot,\eps,\mu)\|_{C^1[-b,b]} \lesssim  \eps + |\mu|,
\end{equation} and hence, via \eqref{eq:h-expand}, that 
\begin{equation}\label{eq:phi-in-pm-b}
  \phi_{\rm in}(b;\eps,\mu) = \Ub + {\mathcal O}(\eps + |\mu|) \ne0.
\end{equation}
To prove \eqref{eq:hbound}, let $H=(h,h')^\top$. Note that $H(-b)=(0,\alpha(\eps,\mu))^\top$.  $H(x)$ satisfies a system of the form $H'=A_\star(x)H + \mathcal{O}_{C^0}(\eps+|\mu|)H +  G_\eps$, where $\|H(-b)\|\lesssim \eps+|\mu|$ and  $\|G_\eps\|_{C^0[-b,b]}\lesssim \eps+|\mu|$. 
 Therefore, $H(x)= H(-b)+\int_{-b}^x U(x,y)[G_\eps(y)+ \mathcal{O}_{C^0}(\eps+|\mu|)H(y) dy$, where  $U'(x,y)=A_\star(x)U(x,y)$ and $U(y,y)={\rm Id}$. Hence, $\|H\|_{C^0[-b,b]}\lesssim \eps+|\mu| + \|G_\eps\|_{C^0[-b,b]}+ (\eps+|\mu|)\|H\|_{C^0[-b,b]}$, implying \eqref{eq:hbound}.

Since $\Ub\neq 0$, by  further constraining $C_1$ we ensure, via \eqref{eq:phi-in-pm-b},  that  $\phi_{\rm in}(b;\eps,\mu)\neq 0$ for $0<\eps,|\mu|<C_1$. Then, since $\mu$ is assumed to be an eigenvalue, the continuity condition  $[\phi]_{x=b}=0$ in~\eqref{eq:continuity-Rval} holds and gives:  \[ c_+(\eps,\mu)\,f_+(b-y_R(\eps);i\la(\eps,\mu),\la_0(\eps)) = \phi_{\rm in}(b;\eps,\mu)\neq 0.\]
 This proves the relations \eqref{eq:c+f+}, and in particular
 \begin{equation}\label{f+phin-ne0}
  f_+(b-y_R(\eps);i\la(\eps,\mu),\la_0(\eps))\ne0\quad \textrm{and}\quad 
  \phi_{\rm in}(b;\eps,\mu)\neq 0.
 \end{equation}
The proof of \claimref{1b} is complete.\end{proof}
}

\begin{proof}[Proof of \claimref{1c}]

{
Substituting the expression for  $c_+(\eps,\mu)$ in ~\eqref{eq:c-plus}. into~\eqref{eq:continuity-Rder} (continuity of the derivatives at $x=b$) and simplifying yields
\begin{equation}\label{eq:log-match-at-b}
  \frac{\phi_{\rm in}'(b;\eps,\mu)}{\phi_{\rm in}(b;\eps,\mu)} \;=\; \frac{f_+'\!\bigl(b-y_R;\,i\la(\eps,\mu);\,\la_0(\eps)\bigr)}{f_+\!\bigl(b-y_R;\,i\la(\eps,\mu);\,\la_0(\eps)\bigr)},
\end{equation}
where we recall (\claimref{1a}) that 
 $\phi_{\rm in}(x,\eps,\mu):=\phi_{\rm in}(x;\mu,\eps,\alpha(\mu,\eps))$. 
We denote the two ratios in \eqref{eq:log-match-at-b} by
\begin{equation}\label{eq:wR-defs}
  \win^R(\eps,\mu) \;:=\; \frac{\phi_{\rm in}'(b;\eps,\mu)}{\phi_{\rm in}(b;\eps,\mu)}\quad\textrm{and}
  \quad
  \woutR(\eps,\mu) \;:=\; \frac{f_+'\!\bigl(b-y_R(\eps);\,i\la(\eps,\mu);\,\la_0(\eps)\bigr)}{f_+\!\bigl(b-y_R(\eps);\,i\la(\eps,\mu);\,\la_0(\eps)\bigr)}\ .
\end{equation}
Using \eqref{f+phin-ne0}, we see that  $ \win^R(\eps,\mu) $ and $\woutR(\eps,\mu)$ are well-defined analytic functions in a neighborhood of $(0,0)$.
The continuity of $\phi$ and $\phi'$ at $x=b$  now boils down to the single equation:
\begin{equation}
\win^R(\eps,\mu) = \woutR(\eps,\mu).
    \label{eq:winR_eq_woutR}
\end{equation}

Let us now expand the left and right hand sides of \eqref{eq:winR_eq_woutR} for $\eps$ and $|\mu|$ small.}

\noindent\underline{Expansion of  $\win^R(\eps,\mu)$:}\  Writing $\phi_{\rm in} = U_\star + h$, substituting into the first expression in \eqref{eq:wR-defs}, and using $U_\star'(b) = \gs\Ub$ from~\eqref{eq:TR-BC}, we obtain
\begin{equation}\label{eq:winR-decompose}
  \win^R(\eps,\mu) \;=\; \frac{\gs\Ub + h'(b)}{\Ub + h(b)}
  \;=\; \gs + \frac{h'(b) - \gs\,h(b)}{\Ub} + \mathcal O(( \eps+|\mu|)^2 ),
\end{equation}
where we have used the bound 
$\|h\|_{C^1[-b,b])}\lesssim \eps+|\mu|$ from \eqref{eq:hbound}. Note that $\win^R(0,0)=\gs$.

We obtain an expression for $h'(b,\eps,\mu) - \gs\,h(b,\eps,\mu)$ using the following identity, proved below:

\begin{lem}\label{prop:boundary-integral}
  Let $f$ solve $(H_V + \gs^2) f = g$ on $[-b,b]$ with $f(-b) = 0$, $f'(-b) = \alpha$. Then
  \begin{equation}\label{eq:boundary-integral}
    f'(b) - \gs\,f(b) \;=\; \frac{1}{\Ub}\Bigl[\alpha - \int_{-b}^b g(y)\,U_\star(y)\,dy\Bigr].
  \end{equation}
\end{lem}

\begin{proof}[Proof of Lemma \ref{prop:boundary-integral}] Let
\begin{equation}
  W(x) := W[U_\star,f](x) = U_\star(x)f'(x) - U_\star'(x)f(x).
\end{equation}
Differentiate, using $(H_V+\gs^2)U_\star=0$ from~\eqref{eq:resonance-eq} and $(H_V+\gs^2)f=g$ by hypothesis:
\begin{equation}
  W'(x)
  \;=\; U_\star(x)f''(x) - U_\star''(x)f(x)
  \;=\; -\,U_\star(x)g(x).
\end{equation}
Integrating from $-b$ to $b$,
\begin{equation}
  W(b) - W(-b) \;=\; -\int_{-b}^b U_\star(x)g(x)\,dx.
\end{equation}
Evaluate $W$ at the endpoints. At $x=b$, the resonance boundary condition $U_\star'(b)=\gs U_\star(b)$ from~\eqref{eq:TR-BC} gives
\begin{equation}
  W(b) \;=\; U_\star(b)\bigl(f'(b)-\gs f(b)\bigr).
\end{equation}
At $x=-b$, the hypotheses $f(-b)=0$ and $f'(-b)=\alpha$ give
\begin{equation}
  W(-b) \;=\; \alpha\,U_\star(-b) = \alpha,
\end{equation}
due to the normalization $U_\star(-b) =1$.
Substituting 
\begin{equation}
  U_\star(b)\bigl(f'(b)-\gs f(b)\bigr) - \alpha \;=\; -\int_{-b}^b U_\star\,g\,dx,
\end{equation}
which implies~\eqref{eq:boundary-integral}.

Finally, we remark that for the scattering resonance (SR) case, the same identity holds with $-b$ replaced by $0$; the zero boundary condition at $x=0$ kills the $\alpha\,U_\star(0)$ term. 
\end{proof}

\medskip

Now apply the identity  \eqref{eq:boundary-integral} with $f(x)=h(x,\eps,\mu)$, the solution of ~\eqref{eq:delta-eq} with 
$g = [\eps(3U_\star^2 + \dot E) + \mu]U_\star + \mathcal O(\eps + |\mu|)h$, and the bound \eqref{eq:hbound}. We obtain


%
\begin{equation}\label{eq:hp-gamh}
  h'(b) - \gs\,h(b) \;=\; \frac{1}{\Ub}\Bigl[\alpha - \eps\!\int_{-b}^b\!(3U_\star^4 + \dot E\,U_\star^2)\,dx - \mu\!\int_{-b}^b U_\star^2\,dx\Bigr] + \mathcal O(\eps^2 + \mu^2 + \eps\mu).
\end{equation}
Substituting \eqref{eq:hp-gamh} into~\eqref{eq:winR-decompose} and inserting $\alpha(\eps,\mu)$ from~\eqref{eq:alpha} yields
\begin{equation}\label{eq:win-expand}
\begin{aligned}
  \win^R(\eps,\mu) \;=\; \gs
  &- \frac{\eps}{\Ub^2}\Bigl(3\!\int_{-b}^b U_\star^4\,dx + \dot E\Bigl[\int_{-b}^b U_\star^2\,dx + \frac{1}{2\gs}\Bigr]\Bigr) \\
  &- \frac{\mu}{\Ub^2}\Bigl(\int_{-b}^b U_\star^2\,dx + \frac{1}{2\gs}\Bigr)
  + \mathcal O(\eps^2 + \mu^2).
\end{aligned}
\end{equation}

\noindent

\underline{Expansion of  $\woutR(\eps,\mu)$:} Recall that $\woutR(\eps,\mu)$ is the logarithmic-derivative expression in \eqref{eq:wR-defs}. It can be evaluated using the following identity, using the expression for $f_+$ in \eqref{eq:Jost}:
\begin{equation}\label{eq:fplus-logderiv-b}
  \frac{f_+'(\xi;i\la;\la_0)}{f_+(\xi;i\la;\la_0)}
  \;=\; -\la + \la_0(1-t^2(\xi))\,\frac{3m + 6t(\xi)}{p\big(m,t(\xi)\big)} \ .
\end{equation}
Here, $p(m,t(\xi)) = m^2 - 1 + 3mt(\xi) + 3t^2(\xi)$. Note, by \eqref{eq:wR-defs}, that $\woutR(\eps,\mu)$ is equal to the left-hand side of \eqref{eq:fplus-logderiv-b}  evaluated at $\xi=b-y_R(\eps)$, $\la=\la(\eps,\mu)$, $\la_0=\la_0(\eps)$, and $m=m(\mu,\eps)=\la(\eps,\mu)/\la_0(\eps)$.
Therefore, the right hand side of \eqref{eq:fplus-logderiv-b} yields an explicit expression for $\woutR(\eps,\mu)$.


Next, we claim that 
\begin{equation}\label{eq:fplus-p-equiv}
  p\big(m(\eps,\mu),t_R(\eps)\big) \neq 0.
\end{equation}

By Lemma~\ref{lem:Jost} (with $\ell=2$):
\begin{equation}\label{eq:f_+}
  f_+\big(b-y_R(\eps);\,i\la(\eps,\mu);\,\la_0(\eps)\big) \;=\; e^{-\la(\eps,\mu)\big(b-y_R(\eps)\big)}\,\frac{p\big(m(\eps,\mu),t_R(\eps)\big)}{\big(m(\eps,\mu)+1\big)\big(m(\eps,\mu)+2\big)}.
\end{equation}
The exponential factor is nonzero for $\eps$ positive and small, and $(m+1)(m+2) \to 6$ as $(\eps,\mu)\to(0,0)$. It follows that if $\eps$ and $|\mu|$ are small and $\mu=\mu(\eps)$ is an eigenvalue \eqref{eq:f_+} then since $f_+\big(b-y_R(\eps);\,i\la(\eps,\mu);\,\la_0(\eps)\big)  \neq 0$ we have \eqref{eq:fplus-p-equiv}.

Therefore, \eqref{eq:fplus-logderiv-b} evaluated at $\xi=b-y_R(\eps)$, etc. implies
\begin{equation}\label{eq:p-woutR-identity}  p\big(m(\eps,\mu),t_R(\eps)\big)\,\bigl(\woutR(\eps,\mu) + \la(\eps,\mu)\bigr) \;=\; \la_0(\eps)\big(1-t_R^2(\eps)\big)\,\big(3m(\eps,\mu) + 6t_R(\eps)\big) \ .
\end{equation}
{
Therefore, for $\eps>0$ and $\mu$ small:
\[\textrm{$\mu$ is an eigenvalue}\quad\implies\quad\quad  \wout^R(\eps,\mu) = \win^R(\eps,\mu)\quad\implies\quad D(\eps,\mu)=0,\]
where
\begin{equation}\label{eq:Ddef}
  D(\eps,\mu):=  -4\gs^2\ \Bigl[\, p\big(m(\eps,\mu),t_R(\eps)\big)\,\bigl[\win^R(\eps,\mu) + \la(\eps,\mu) \bigr] \ -\ \la_0(\eps)\big(1-t_R^2(\eps)\big)\,\big(3m(\eps,\mu) + 6t_R(\eps)\big) \,\Bigr].
\end{equation}
The non-zero factor, $-4\gs^2$, is included so that $D$ has the expansion~\eqref{eq:D-parabolic} below. Alternatively, we have
\begin{equation}\label{eq:D-factorization}
  D(\eps,\mu) \;=\; -4\gs^2\,p\big(m(\eps,\mu),t_R(\eps)\big)\,\bigl(\win^R(\eps,\mu) - \woutR(\eps,\mu)\bigr).
\end{equation}
 }
By the previous discussion, $D(\eps,\mu)$ is clearly analytic in a neighborhood of $(0,0)$.
{
To complete the proof of \claimref{1c}, it remains to prove the expansion of $D(\eps,\mu)$, for $\eps$ and $\mu$ small, displayed in \eqref{eq:D-parabolic}. The derivation  of this expansion is carried out in Proposition~\ref{prop:D-expand} of Appendix~\ref{sec:D-expand}. }
\end{proof}

\bigskip
\noindent\textbf{Implementation of Step~2.}


We begin by showing that the equation $D(\eps,\mu)=0$ has a unique solution $\eps\mapsto\mu(\eps)$ defined in a two-sided neighborhood $(-C_3,C_3)$ of $\eps=0$. We then show that if  $\eps\in(0,\eps_0)$, where $\eps_0$ is sufficiently small, then  $\mu(\eps)$ is an eigenvalue of $\Lop$ with a corresponding eigenfunction $\phi_\eps\in X$ constructed explicitly via the expressions for $\phi_\eps$, restricted to $|x|>b$ and to $|x|\le b$,  in Step 1. 

By \claimref{1c}, $D$ is analytic on $\mathcal{Q}$ with $D(0,0)=0$ and $\partial_\mu D(0,0)=-4\gs\neq0$ (see eqn \eqref{eq:D-parabolic}). By the implicit function theorem, there exists $C_3>0$ and a unique analytic map $\mu:(-C_3,C_3)\to\R$ with $\mu(0)=0$ and $D(\eps,\mu(\eps))=0$, the unique zero of $D$ in $\mathcal{Q}\cap\{|\eps|<C_3\}$. By  \eqref{eq:D-parabolic}, $\partial_\eps D(0,0)=0$, and hence 
\begin{equation}
  \mu'(0) \;=\; -\,\frac{\partial_\eps D(0,0)}{\partial_\mu D(0,0)} \;=\; 0;
\end{equation}
so $\mu(\eps)=\mathcal O(\eps^2)$. Evaluating \eqref{eq:D-parabolic} at $\mu=\mu(\eps)$ then gives $0 = -4\gs\,\mu(\eps) + 3\Omega_\star\eps^2 + \mathcal O(\eps^3)$, and therefore
\begin{equation}\label{eq:mu-branch}
  \mu(\eps) \;=\; \frac{3\Omega_\star}{4\gs}\eps^2 + \mathcal{O}(\eps^3).
\end{equation}


It remains to construct an eigenfunction with corresponding eigenvalue  $\mu(\eps)$. Since any eigenvalue in a neighborhood of $\mu=0$ is a solution of $D(\mu,\eps)=0$, by the local uniqueness of solutions of $D(\mu,\eps)=0$, $\mu(\eps)$ is the only eigenvalue in a neighborhood of zero.

To construct  $\phi_\eps\in H^1(\R)$ such that  $L_+^\eps\phi_\eps=\mu(\eps)\phi_\eps$, we define $\phi_\eps$ via the expressions  \eqref{eq:phi-cases} and \eqref{eq:phi_in}, which satisfies the differential equation for $x\in\R\setminus\{-b,b\}$. It suffices, to check that with $\mu=\mu(\eps)$ satisfying $D(\mu(\eps),\eps)=0$, the jump conditions $[\phi_\eps]_{x=\pm b}=0$ and $[\phi'_\eps]_{x=\pm b}=0$ hold. This requires reversing the steps in \claimref{1a} and \claimref{1b} of Step 1. Along the way, it will be necessary to establish the following claim:

{
{\it Let $\mu=\mu(\eps)$, having the expansion \eqref{eq:mu-branch} denote the unique solution of $D(\mu,\eps)=0$ in a neighborhood of $(0,0)$. Then, 
\begin{subequations}
\begin{align}
   &f_-\big(-b-y_L(\eps);i\la(\eps,\mu(\eps)\big),\la_0(\eps))\ne0 \label{eq:fminus_at_mu}\\ 
  &\phi_{\rm in}\big(b;\eps,\mu(\eps)\big)\neq 0\label{eq:phi-in_at_mu}\\
  &f_+\big(b-y_R(\eps);i\la(\eps,\mu(\eps)\big),\la_0(\eps))\ne0
\label{eq:fplus_at_mu}\end{align}
\end{subequations}
}
}

{We first verify  \eqref{eq:fplus_at_mu} and then use \eqref{eq:c-plus} to define $c_+(\mu(\eps),\eps)$. By \eqref{eq:Jost}, to prove \eqref{eq:fplus_at_mu} it suffices to show that 
$p\bigl(m(\eps,\mu(\eps)),\,t_R(\eps)\bigr)\ne0$. We have, for small $\eps$ using \eqref{eq:mu-branch} and \eqref{eq:wout-den-R} that 
\begin{equation}
  p\bigl(m(\eps,\mu(\eps)),\,t_R(\eps)\bigr) \;=\; -\,\frac{3\Ub^2}{4\gs^2}\,\eps\,\bigl(1+\mathcal O(\eps)\bigr) \ne 0 . 
\end{equation}

We now verify \eqref{eq:fminus_at_mu}. Here, by \eqref{eq:Jost}, we require that $ p(m(\mu(\eps),\eps), -t_L(\eps))\ne0$. This is a consequence of the \eqref{eq:wout-den-L} of Step 1, which holds for all $(\mu,\eps)$ near $(0,0)$. Thus, we define $c_-(\mu(\eps),\eps)$ via \eqref{eq:c-minus}.

Finally, we observe that  \eqref{eq:phi-in_at_mu} holds by \eqref{eq:phi-in-pm-b}.

For $\eps>0$, we  define
\begin{equation}\label{eq:pertzeromode}
  \phi(x) \;:=\; \begin{cases}
    c_-(\eps,\mu(\eps))\,f_-\!\bigl(x-y_L(\eps);\,i\la(\mu(\eps),\eps);\,\la_0(\eps)\bigr) & x<-b,\\
    \phi_{\rm in}(x;\eps,\mu(\eps)) & -b\le x\le b,\\
    c_+(\eps,\mu(\eps))\ \,f_+\!\bigl(x-y_R(\eps);\,i\la(\mu(\eps),\eps);\,\la_0(\eps)\bigr) & x>b\ ,
  \end{cases}
\end{equation}
which satisfies $L_+^\eps\phi=\mu(\eps)\phi$.

Further, with $c_+(\mu(\eps),\eps)$ now given by \eqref{eq:c-plus}, we see that $[\phi]_{x=b}=0$,  \eqref{eq:continuity-Rval},  is satisfied. Further, since $D(\mu(\eps),\eps)=0$, we have $[\phi']_{x=b}=0$.

Thus, for $\eps>0$ and sufficiently small, $(\mu(\eps),\phi)$ is an eigenpair of $L_+^\eps$.
 This completes the implementation of Step 2 and hence the proof of Proposition \ref{prop:mu} for the case of transmission resonance-induced nonlinear bound states.


}

\begin{rem}[Adapting the proof to scattering resonance-induced nonlinear bound states]\label{rem:SR-stability}
The scattering-resonance (SR) branch on the half-line is treated in a manner analogous to transmission-resonance (TR) case. Recall that in the SR case, the underlying Hilbert space is $H^1_0(\R_+)$. (Equivalently, we may work in  $H^1_{\rm odd}(\R)$.) The following is a list of modifications required in the proof.
\begin{enumerate}[leftmargin=2em,label=(\roman*)]
  \item \emph{Domain.} $\R_+=[0,b]\cup (b,\infty)$, where $\sup\{x\in\supp(V)\}\le b$. Hence, there are no contributions from $(-\infty,0)$, $y_L$, and $N_{\rm out}^L$.
  \item \emph{Boundary data.} The Dirichlet condition $U_\star(0) = 0$ from~\eqref{eq:SR-BC} replaces the Robin condition $U_\star'(-b) = \gs U_\star(-b)$ of TR, normalized instead to $U_\star'(0) = 1$. The boundary conditions at $x=0$,  $h(0)=0$ and $h'(0)=\alpha(\mu,\eps)$ in~\eqref{eq:delta-eq},  then give $\phi(0) = 0$ and $\phi'(0) = 1 + \alpha(\mu, \eps)$.
  \item $\Omega_\star$ and $ \mathcal B_\star$: Setting $U_\star(0) = 0$ in the TR formulas~\eqref{eq:Omega} yields $\Omega_\star=\Omega_\star^{\rm SR}$ and $ \mathcal B_\star=\mathcal B_\star^{\rm SR}$ from~\eqref{eq:Bstar}, and  the expressions for the translation eigenvalue, $\mu(\eps)$, and $\Nslope$ reduce to \eqref{eq:mu-final} and ~\eqref{eq:nu}, respectively.
  \item The $L^2(\R_+)$-threshold for the SR bifurcation is $\mathcal N_\star = 4\gs$, as in the TR case; for $\eps>0$ and small, the dominant part of the ``mass'' ($L^2$-norm) of $\psi_\eps$ is concentrated in a single soliton centered at $y_R(\eps)\gg b$.
\end{enumerate}
The conclusions of Theorem~\ref{thm:intro-stability-sharp} and Corollary~\ref{cor:small-gamma-stability} hold verbatim with the SR functionals.
\end{rem}

This completes the proof of Proposition~\ref{prop:mu}.\qed

\bigskip

\subsection{Point spectrum of $\Lop$; proof of Proposition~\ref{prop:Lp-count}}\label{ssec:proof-Lp-count}

{
By Proposition~\ref{prop:mu}, for $\eps$ small, $\Lop$ has exactly one simple eigenvalue $\mu(\eps)$ near zero with corresponding eigenspace spanned by $\phi(x)$, displayed in \eqref{eq:pertzeromode}. Proposition~\ref{prop:Lp-count} follows if we can prove that $\mu(\eps)$ is the second eigenvalue of $\Lop$.

We begin by noting that by an analogous construction to that in the proof of Proposition~\ref{prop:mu}, that for $\eps>0$ and small, $\Lop$ has a strictly negative ground state eigenvalue 
$\mu_g(\eps) = -3\gs^2 + o(\eps)$,  which is order one distant from $\mu(\eps)$. Here, $-3\gs^2$ is the ground state energy of the linearized operator $L_+^{\rm sol}$, associated with the free soliton; we omit the details, which run parallel to Section~\ref{ssec:proof-mu}. Further, note that the one-dimensional eigenspace corresponding to $\mu_g(\eps)$ is spanned by $\phi_g(x)$, which can be taken to be strictly positive on $\R$.

To prove Proposition~\ref{prop:Lp-count}, it suffices to prove that $\mu(\eps)$ is the second eigenvalue of $\Lop$. By Sturm oscillation theory, this follows if we can prove, for all $\eps>0$ small, that $\phi(x)$ vanishes at exactly one point in $\R$.
 For this, we split $\R$ into the three regions $\R= (-\infty,-b)\cup [-b,b]\cup (b,\infty)$ and count zeros of $\phi(x)$ in each subregion. 

Note first that for  $x\in (-b,b)$, the interior interval,  $\phi(x) = \phi_{\rm in}(x) = U_\star(x) + \mathcal O_{C^0}(\eps)$, which is strictly positive for all $\eps>0$ and small because $U_\star$ is positive on $(-b,b)$.

We turn now to the exterior intervals $(-\infty,-b]$ and $[b,\infty)$, where $\phi(x)$ is given by the first and third of the Jost solution expressions in  \eqref{eq:pertzeromode}, respectively:
By Lemma \ref{lem:Jost} with $\ell=2$, the zeros of $\phi(x)$ on these intervals are determined by the zeros of:
\begin{subequations}
\begin{align}
 &p\big(m(\mu,\eps),-t(\xi)\big)\Big|_{\xi = x-y_L}\quad \textrm{for $x\in(-\infty,-b)$, and}\label{eq:p_left}\\
    &p\big(m(\mu,\eps),t(\xi)\big)\Big|_{\xi = x-y_R}\quad \textrm{for $x\in(b,\infty)$} \label{eq:p_right}
\end{align}
\end{subequations}
where (from \eqref{eq:p-def}):
\begin{equation}
  p\big(m(\mu,\eps),t(\xi)\big) \;=\; m^2(\mu,\eps) - 1 + 3m(\mu,\eps)t(\xi) + 3t^2(\xi),\quad t(\xi) := \tanh(\la_0\xi),
\end{equation}
with $\mu=\mu(\eps)$ in these expressions. To obtain \eqref{eq:p_left}
we used that $f_-(\xi) = f_+(-\xi)$.

In order to complete the proof of Proposition~\ref{prop:Lp-count}, we establish the following two assertions:

\begin{subequations}
\begin{align}
  &\textrm{$p\big(m(\mu,\eps),-t(\xi)\big)\Big|_{\xi = x-y_L}$ and therefore $\phi(x)$ has no zeros in the interval $-\infty<x<-b$} \label{eq:phi_left}\\ 
&\textrm{$p\big(m(\mu,\eps),t(\xi)\big)\Big|_{\xi = x-y_R}$ and therefore $\phi(x)$ has exactly one  zero in the interval $b<x<\infty$.}
\label{eq:phi_right}
\end{align}
\end{subequations}

Note that since $\mu(\eps)=\mathcal{O}(\eps^2)$, we have by \eqref{eq:m-expand-b} that $m\big(\mu(\eps),\eps\big)= 1 + \mathcal O(\eps^2)$.

Since $x\mapsto t(\xi(x))$ is monotone in $x$, counting zeros of $\phi$ in a sub-interval of an exterior reduces to counting roots of the relevant polynomial in the corresponding $t$-interval. 

\noindent\underline{\it Proof of \eqref{eq:phi_left}:} For $-\infty<x<-b$, set
$  t_x \;:=\; \tanh\!\bigl(\la_0(\eps)(x - y_L(\eps))\bigr)$,
which is monotone increasing and sweeps out the interval 
$ -1=t_{b=-\infty}<t_x< t_{x=-b}= t_L(\eps)$, 
 where by~\eqref{eq:1ptL},  $t_L(\eps) = -1 + \tfrac{\eps}{4\gs^2} + \mathcal O(\eps^2)$.  The zeros of $\phi(x)$ in the interval $-\infty<x<-b$, correspond to solutions $t_x$ of $p(m,-t_x) \;=\; 3t_x^2 - 3m\,t_x + (m^2-1)=0$
 with $t_x\in \big(-1, t_L(\eps)\big)$. These solutions $t_x^\pm$ are given by
 \begin{equation}
  t_x^\pm \;=\; \frac{3m \pm \sqrt{12-3m^2}}{6}.
\end{equation}
Using that $m = 1+\mathcal O(\eps^2)$ we have 
\begin{equation}
  t_x^- \;=\; \mathcal O(\eps^2), \qquad t_x^+ \;=\; 1 + \mathcal O(\eps^2),
\end{equation}
neither of which is in the interval 
$(-1, t_L(\eps)] = \left(-1,-1+\mathcal O(\eps)\right]$. Hence, $\phi(x)$ is non-zero for $-\infty<x<-b$. This completes the proof of \eqref{eq:phi_left}.

\noindent\underline{\it Proof of \eqref{eq:phi_right}:} For $b<x<\infty$, set 
$t_x \;:=\; \tanh\!\bigl(\la_0(\eps)(x - y_R(\eps))\bigr)$, which is monotonically increasing and sweeps out the interval $t_R(\eps)=t_{x=b}<t_x<t_{x=\infty}=1$. Since $t_R = -1 + \frac{\Ub^2\eps}{4\gs^2} + \mathcal O(\eps^2)$ by~\eqref{eq:1ptR}, 
\[\textrm{$x\mapsto t_x$ sweeps out an interval $\big[-1+\mathcal O(\eps),1\big)$}. \]
The zeros of $\phi(x)$ in the interval $b<x<\infty$, correspond to solutions $t_x$ of $p(m,t_x) \;=\; 3t_x^2 + 3m\,t_x + (m^2-1)=0$
 with $t_x\in \big[t_R(\eps), 1\big)= \big[-1+\mathcal O(\eps),1\big)$. These solutions $t_x^\pm$ are given by
 \begin{equation}
  t_x^\pm \;=\; \frac{-3m \pm \sqrt{12 - 3m^2}}{6}\ .
\end{equation}
Again we use that $m = 1 + \mathcal O(\eps^2)$ (see ~\eqref{eq:m-expand-b}) to obtain 
\begin{equation}
  t_x^- \;=\; -1 + \mathcal O(\eps^2)\qquad {\rm and} \qquad t_x^+ \;=\; \mathcal O(\eps^2).
\end{equation}
For $\eps>0$ and small, $t_x^-= -1 + \mathcal O(\eps^2)<t_R(\eps)=-1+\mathcal O(\eps)$. On the other hand, $t_x^+\in \big[t_R, 1\big)$. Hence, there exists $x_0\in (b,\infty)$ such that $\phi(x_0)=0$. This value of $x_0$ is unique by monotonicity of $t_x$. This completes the proof of \eqref{eq:phi_right} and therewith Proposition~\ref{prop:Lp-count} for the 
TR case.  The proof is easily adapted to cover the SR case.

}

\subsection{Slope of $\mathcal{N}$ vs $E$ near the bifurcation point; proof of Proposition~\ref{prop:slope}}\label{ssec:proof-slope}

Split the mass as $\mathcal N(\eps) = N_{\rm in} + N_{\rm out}^R + N_{\rm out}^L$.

By Theorem~\ref{thm:bifurcation}, $\psi_\eps = \sqrt\eps[U_\star + \mathcal O(\eps)]$ on $[-b,b]$, so
\begin{equation}\label{eq:Nin}
  N_{\rm in} \;=\; \int_{-b}^b \psi_\eps(x)^2\,dx
  \;=\; \eps\!\int_{-b}^b U_\star(x)^2\,dx + \mathcal O(\eps^2).
\end{equation}

Substitute $\xi = x - y_R$ in $\psi_\eps = \Scal(x - y_R;E(\eps)) = \sqrt 2\,\la_0\sech(\la_0\xi)$:
\begin{equation}
  N_{\rm out}^R
  \;=\; \int_b^\infty 2\la_0^2 \sech^2(\la_0(x-y_R))\,dx
  \;=\; 2\la_0 \int_{\la_0(b-y_R)}^\infty \sech^2 \eta\,d\eta
  \;=\; 2\la_0\bigl(1-\tanh(\la_0(b-y_R))\bigr).
\end{equation}
With $t_R = \tanh(\la_0(b-y_R))$ and $1+t_R = \Ub^2\eps/(4\gs^2) + \mathcal O(\eps^2)$ from~\eqref{eq:1ptR},
\begin{equation}
  1 - \tanh(\la_0(b-y_R)) \;=\; 1 - t_R \;=\; 2 - (1 + t_R) \;=\; 2 - \frac{\Ub^2\,\eps}{4\gs^2} + \mathcal O(\eps^2).
\end{equation}
Hence
\begin{equation} \label{eq:NR}
  N_{\rm out}^R \;=\; 2\la_0\Bigl[2 - \frac{\Ub^2\eps}{4\gs^2} + \mathcal O(\eps^2)\Bigr]
  \;=\; 4\la_0 - \frac{\Ub^2\,\eps}{2\gs} + \mathcal O(\eps^2).
\end{equation}

The same computation at $-b$, using $1 + t_L = \eps/(4\gs^2) + \mathcal O(\eps^2)$ from~\eqref{eq:1ptL}, gives
\begin{equation}\label{eq:NL}
  N_{\rm out}^L
  \;=\; 2\la_0\bigl(1 - \tanh(\la_0(y_L + b))\bigr)
  \;=\; \frac{\eps}{2\gs} + \mathcal O(\eps^2).
\end{equation}

Summing the three contributions of \eqref{eq:Nin}--\eqref{eq:NL},
\begin{equation}
  \mathcal N(\eps) \;=\; 4\la_0 + \eps\!\int_{-b}^b U_\star^2\,dx - \frac{\Ub^2\eps}{2\gs} + \frac{\eps}{2\gs} + \mathcal O(\eps^2).
\end{equation}
Substituting $\la_0 = \gs - \eps\dot E/(2\gs) + \mathcal O(\eps^2)$ from~\eqref{eq:la0} and grouping the $\mathcal O(\eps)$ terms via $\mathcal B_\star = \int_{-b}^b U_\star^2\,dx - (\Ub^2 - 1)/(2\gs)$ from~\eqref{eq:Bstar},
\begin{equation}\label{eq:N-eps}
  \mathcal N(\eps) \;=\; 4\gs + \eps\Bigl(\mathcal B_\star - \frac{2\dot E}{\gs}\Bigr) + \mathcal O(\eps^2).
\end{equation}

With $E(\eps) = -\gs^2 + \eps\dot E + \mathcal O(\eps^2)$ and $\dot E \neq 0$ (hypothesis~\ref{H:Omega}), the chain rule gives
\begin{equation}
  \frac{d\mathcal N}{dE}(\eps) \;=\; \frac{d\mathcal N/d\eps}{dE/d\eps} \;=\; \frac{\mathcal B_\star}{\dot E} - \frac{2}{\gs} + \mathcal O(\eps).
\end{equation}
Substituting $\dot E = (\Omega_\star - 3/(4\gs))/(2\mathcal B_\star)$ from~\eqref{eq:Edot-solved},
\begin{equation}
  \frac{d\mathcal N}{dE}(\eps) \;=\; \frac{2\mathcal B_\star^2}{\Omega_\star - 3/(4\gs)} - \frac{2}{\gs} + \mathcal O(\eps).
\end{equation}
Evaluating at $\eps = 0$ gives~\eqref{eq:nu}. For SR, $N_{\rm out}^L$ is absent (Remark~\ref{rem:SR-stability}); the same computation with $U_\star(0) = 0$ yields the SR slope formula~\eqref{eq:nu}. \qed


\section{Stability and instability regions in the $(E,\mathcal{N})$-plane near the bifurcation point}
\label{ssec:stab-regimes}

By Theorem~\ref{thm:W-GSS}, the slope of the nonlinear bound-state curve
\[
  E \mapsto \mathcal N[\psi(\cdot;E)]
\]
encodes information about orbital stability under the NLS/GP flow. As discussed in Remark~\ref{rem:geometric}, the $\nu_\star$-stability and instability intervals displayed in Corollaries~\ref{cor:stability-TR} and~\ref{cor:stability-SR} induce a local partition of the $(E,\mathcal N)$-plane near $(E_\star,\mathcal N_\star)$. The tables above list the corresponding regions, and Figure~\ref{fig:schematic} visualizes them.
\begin{figure}[htbp]
\centering
\setlength{\abovecaptionskip}{6pt}
\begin{subfigure}[t]{0.48\linewidth}
\centering
\begin{tikzpicture}[>=Stealth, font=\normalsize]
  \draw[->] (-4.3,3.2) -- (0.4,3.2) node[right] {$E$};
  \draw[->] (0,3.1) -- (0,6.5) node[above] {$\mathcal{N}$};
  \draw[orange, dashed, thick, domain=0.04:4.0, samples=80]
    plot ({-\x}, {3.2+1.5*sqrt(\x)});
  \node[orange, font=\footnotesize, anchor=east] at (-3.6, 6.35) {$\mathcal{N}[\mathcal{S}_E]$};
  \fill (-1, 4.7) circle (2pt);
  \draw (-1,3.2) -- (-1,3.1);
  \node[font=\footnotesize, anchor=north] at (-1,3.1) {$E_\star$};
  \draw (0,4.7) -- (0.1,4.7);
  \node[font=\footnotesize, anchor=west] at (0.1,4.7)
    {$\mathcal{N}_\star = \mathcal{N}[\mathcal{S}_{E_\star}]$};
  \draw[thick, domain=0:2.6, samples=40]
    plot ({-1-\x}, {4.7 + 0.5*\x - 0.015*\x*\x});
  \draw[dashed, thick] (-1.55, 4.2) rectangle (-0.45, 5.2);

  \pgfmathsetmacro{\bx}{-2.3}
  \pgfmathsetmacro{\by}{1.35}
  \pgfmathsetmacro{\Lx}{-4.45}
  \pgfmathsetmacro{\Rx}{-0.15}
  \pgfmathsetmacro{\By}{0.0}
  \pgfmathsetmacro{\Ty}{2.95}
  \pgfmathsetmacro{\msol}{0.55}
  \pgfmathsetmacro{\soly}{\by + (\bx - \Lx)*\msol}
  \pgfmathsetmacro{\solyr}{\by - (\Rx - \bx)*\msol}

  \begin{scope}
    \clip (\Lx,\By) rectangle (\Rx,\Ty);
    \fill[red!15] (\bx,\by) rectangle (\Rx,\Ty);
    \fill[red!15] (\bx,\by) rectangle (\Lx,\By);
    \fill[red!15] (\bx,\by) -- (\Lx,\soly) -- (\Lx,\Ty) -- (\bx,\Ty) -- cycle;
    \fill[red!15] (\bx,\by) -- (\Rx,\solyr) -- (\Rx,\By) -- (\bx,\By) -- cycle;
    \fill[green!25] (\bx,\by) -- (\Lx,\soly) -- (\Lx,\by) -- cycle;
    \fill[green!25] (\bx,\by) -- (\Rx,\by) -- (\Rx,\solyr) -- cycle;
  \end{scope}

  \begin{scope}
    \clip (\Lx,\By) rectangle (\Rx,\Ty);
    \draw[thin, gray] (\Lx,\by) -- (\Rx,\by); 
    \draw[orange, dashed, thick] (\Lx,\soly) -- (\Rx,\solyr); 
    \draw[gray, thin] (\bx,\By) -- (\bx,\Ty); 
  \end{scope}

  \begin{scope}
    \clip (\Lx,\By) rectangle (\Rx,\Ty);
    \draw[very thick, black, ->, domain=0:1.0, samples=30]
      plot ({\bx - 1.4*\x}, {\by + 0.38*\x - 0.02*\x*\x});
  \end{scope}

  \node[font=\small] at (-4.00, 1.84) {$\Sreg$};   
  \node[font=\small] at (-0.95, 2.52) {$\Ureg$};   
  \node[font=\small] at (-3.65, 0.18) {$\Ureg$};   
  \node[font=\small] at (-0.60, 0.86) {$\Sreg$};   

  \node[orange, font=\footnotesize, rotate=-28, anchor=center, fill=white, inner sep=1pt]
    at (-3.38, 1.94)
    {$\mathcal{N}[\mathcal{S}_E]$};
  \node[font=\footnotesize, anchor=north] at (\bx, \By-0.03) {$E_\star$};
  \node[font=\footnotesize, anchor=west] at (\Rx+0.07, \by) {$\mathcal{N}_\star$};

  \fill (\bx,\by) circle (1.5pt);
  \draw[thick] (\Lx,\By) rectangle (\Rx,\Ty);

  \draw[thin, gray] (-1.55, 4.2) -- (\Lx, \Ty);
  \draw[thin, gray] (-0.45, 4.2) -- (\Rx, \Ty);
\end{tikzpicture}
\vspace{-0.5em}
\subcaption{SR, $X=H^1_0(\R_+)$: stable $\Sreg^{\rm SR}$ and unstable $\Ureg^{\rm SR}$.}
\end{subfigure}\hfill
\begin{subfigure}[t]{0.48\linewidth}
\centering
\begin{tikzpicture}[>=Stealth, font=\normalsize]
  \draw[->] (-4.3,3.2) -- (0.4,3.2) node[right] {$E$};
  \draw[->] (0,3.1) -- (0,6.5) node[above] {$\mathcal{N}$};
  \draw[orange, dashed, thick, domain=0.04:4.0, samples=80]
    plot ({-\x}, {3.2+1.5*sqrt(\x)});
  \node[orange, font=\footnotesize, anchor=east] at (-3.6, 6.35) {$\mathcal{N}[\mathcal{S}_E]$};
  \fill (-1, 4.7) circle (2pt);
  \draw (-1,3.2) -- (-1,3.1);
  \node[font=\footnotesize, anchor=north] at (-1,3.1) {$E_\star$};
  \draw (0,4.7) -- (0.1,4.7);
  \node[font=\footnotesize, anchor=west] at (0.1,4.7)
    {$\mathcal{N}_\star = \mathcal{N}[\mathcal{S}_{E_\star}]$};
  \draw[thick, domain=0:2.4, samples=40]
    plot ({-1-\x}, {4.7 + 1.05*\x - 0.04*\x*\x});
  \draw[dashed, thick] (-1.55, 4.2) rectangle (-0.45, 5.2);

  \pgfmathsetmacro{\bx}{-2.3}
  \pgfmathsetmacro{\by}{1.35}
  \pgfmathsetmacro{\Lx}{-4.45}
  \pgfmathsetmacro{\Rx}{-0.15}
  \pgfmathsetmacro{\By}{0.0}
  \pgfmathsetmacro{\Ty}{2.95}
  \pgfmathsetmacro{\msol}{0.42}
  \pgfmathsetmacro{\mom}{0.85}
  \pgfmathsetmacro{\soly}{\by + (\bx - \Lx)*\msol}
  \pgfmathsetmacro{\solyr}{\by - (\Rx - \bx)*\msol}
  \pgfmathsetmacro{\omy}{\by + (\bx - \Lx)*\mom}
  \pgfmathsetmacro{\omyr}{\by - (\Rx - \bx)*\mom}
  \pgfmathsetmacro{\omxtop}{\bx - (\Ty - \by)/\mom}
  \pgfmathsetmacro{\omxbot}{\bx + (\by - \By)/\mom}

  \begin{scope}
    \clip (\Lx,\By) rectangle (\Rx,\Ty);
    \fill[red!15] (\bx,\by) rectangle (\Rx,\Ty);
    \fill[red!15] (\bx,\by) rectangle (\Lx,\By);
    \fill[green!25] (\bx,\by) -- (\bx,\Ty) -- (\omxtop,\Ty) -- cycle;
    \fill[green!25] (\bx,\by) -- (\bx,\By) -- (\omxbot,\By) -- cycle;
    \fill[red!15] (\bx,\by) -- (\Lx,\soly) -- (\Lx,\Ty) -- (\omxtop,\Ty) -- cycle;
    \fill[red!15] (\bx,\by) -- (\Rx,\solyr) -- (\Rx,\By) -- (\omxbot,\By) -- cycle;
    \fill[green!25] (\bx,\by) -- (\Lx,\soly) -- (\Lx,\by) -- cycle;
    \fill[green!25] (\bx,\by) -- (\Rx,\by) -- (\Rx,\solyr) -- cycle;
  \end{scope}

  \begin{scope}
    \clip (\Lx,\By) rectangle (\Rx,\Ty);
    \draw[thin, gray] (\Lx,\by) -- (\Rx,\by); 
    \draw[orange, dashed, thick] (\Lx,\soly) -- (\Rx,\solyr); 
    \draw[gray, dashed, thick] (\omxtop,\Ty) -- (\omxbot,\By); 
    \draw[gray, thin] (\bx,\By) -- (\bx,\Ty); 
  \end{scope}

  \begin{scope}
    \clip (\Lx,\By) rectangle (\Rx,\Ty);
    \draw[very thick, black, ->, domain=0:1.0, samples=30]
      plot ({\bx - 0.85*\x}, {\by + 1.05*\x - 0.05*\x*\x});
  \end{scope}

  \node[font=\small] at (-4.00, 1.73) {$\Sreg_2$};   
  \node[font=\small, text=red!80!black] at (-4.00, 2.52) {$\Ureg_1$};   
  \node[font=\small] at (-2.93, 2.52) {$\Sreg_1$};   
  \node[font=\small, text=red!80!black] at (-0.95, 2.52) {$\Ureg_2$};   
  \node[font=\small] at (-0.60, 0.97) {$\Sreg_2$};   
  \node[font=\small, text=red!80!black] at (-0.60, 0.18) {$\Ureg_1$};   
  \node[font=\small] at (-1.67, 0.18) {$\Sreg_1$};   
  \node[font=\small, text=red!80!black] at (-3.65, 0.18) {$\Ureg_2$};   

  \node[orange, font=\footnotesize, rotate=-23, anchor=center, fill=white, inner sep=1pt]
    at (-1.08, 0.86)
    {$\mathcal{N}[\mathcal{S}_E]$};
  \node[font=\footnotesize, anchor=north] at (\bx, \By-0.03) {$E_\star$};
  \node[font=\footnotesize, anchor=west] at (\Rx+0.07, \by) {$\mathcal{N}_\star$};

  \fill (\bx,\by) circle (1.5pt);
  \draw[thick] (\Lx,\By) rectangle (\Rx,\Ty);

  \draw[thin, gray] (-1.55, 4.2) -- (\Lx, \Ty);
  \draw[thin, gray] (-0.45, 4.2) -- (\Rx, \Ty);
\end{tikzpicture}
\vspace{-0.5em}
\subcaption{TR, $X=H^1(\R)$: $\Ureg_1^{\rm TR},\Sreg_1^{\rm TR},\Ureg_2^{\rm TR},\Sreg_2^{\rm TR}$}
\end{subfigure}
\caption{Stability regions near the bifurcation point $(E_\star,\mathcal N_\star)$;  Corollaries~\ref{cor:stability-TR} and~\ref{cor:stability-SR}. In each panel, the upper plot shows the free-soliton  curve $E\mapsto\mathcal{N}[\mathcal{S}_E]$ (orange dashed) and a representative resonance-induced branch $\eps \mapsto (E(\eps),\mathcal N[\psi_\eps])$ (solid black), emanating from the bifurcation point. The dashed box is magnified below, and the black arrow indicates the direction of increasing $\eps$ along the branch as it leaves the bifurcation point $(E_\star,\mathcal N_\star)$. In the zoom, the orange dashed line is the tangent to $\mathcal{N}[\mathcal{S}_E]$ at $(E_\star,\mathcal N_\star)$; in the TR panel, the gray dashed line is the $\Omega_\star=0$ separator. Stable sectors are shaded green and unstable sectors are shaded red; the labels $\Sreg,\Ureg$ mark the regions of Corollaries~\ref{cor:stability-TR} and~\ref{cor:stability-SR} (four in the TR panel, two in the SR panel), as in Figure~\ref{fig:f-stability}.}
\label{fig:schematic}
\end{figure}
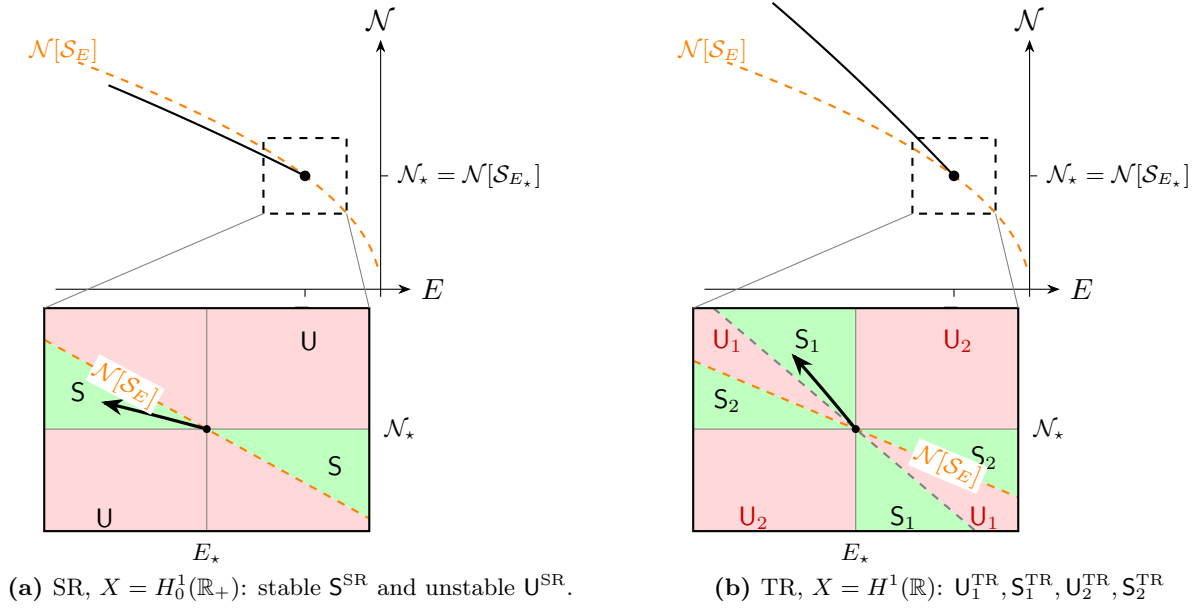

\subsection{Numerical exploration of $(E, \mathcal N)$ plane near the bifurcation point}\label{ssec:branch-examples}

We exhibit two resonance-induced branches, one in $\Sreg^{\rm SR}$ and one in $\Ureg_1^{\rm TR}$, and use Corollaries~\ref{cor:stability-TR} and~\ref{cor:stability-SR} together with the computed values of $(\gs,\Omega_\star,\mathcal B_\star,\Nslope)$ to read off their stability.
\noindent\underline{SR-induced branch.}
On the half-line $\R_+$, take the smooth-bump potential
\begin{equation}\label{eq:sr-pot}
  V(x) \;=\; V_0\,\exp\!\Bigl(-\tfrac{1}{1-(2x-1)^2}\Bigr)\,\mathbf 1_{(0,1)}(x), \qquad V_0 \;=\; -7,
\end{equation}
i.e.\ the standard $C^\infty_c$ bump on $(0,1)$ scaled to amplitude $V_0=-7$. Figure~\ref{fig:example} displays the corresponding strictly positive {\it scattering} resonance  (SR) mode $U_\star$ with $\gs\approx 1.353$ and $\Omega_\star\approx 0.022 > 0$. Panel~(b) shows the branch entering $\Sreg^{\rm SR}$ of Figure~\ref{fig:schematic}, in agreement with Corollaries~\ref{cor:stability-TR} and~\ref{cor:stability-SR}; this confirms stability near the bifurcation point.

\noindent\underline{TR-induced branch.}
Figure~\ref{fig:trans-example} shows a \emph{transmission} resonance (TR) on the full line. The sine potential $V(x) = 3\sin(\pi x)$ on $[-1,1]$ has $\Omega_\star \approx -3.01 \;<\; 0$, so the branch enters $\Ureg_1^{\rm TR}$ of Figure~\ref{fig:schematic} and is unstable. The nonlinear consequence is tracked dynamically in Section~\ref{sec:dyn-E0}. A stable transmission example is given in Subsection~\ref{ssec:cor-nonvacuity}, where it arises naturally from the construction of Corollary~\ref{cor:small-gamma-stability}.

\begin{figure}[!htbp]
    \centering
    \setlength{\abovecaptionskip}{2pt}
    \setlength{\belowcaptionskip}{0pt}
    \begin{subfigure}{0.32\linewidth}
        \includegraphics[width=\linewidth]{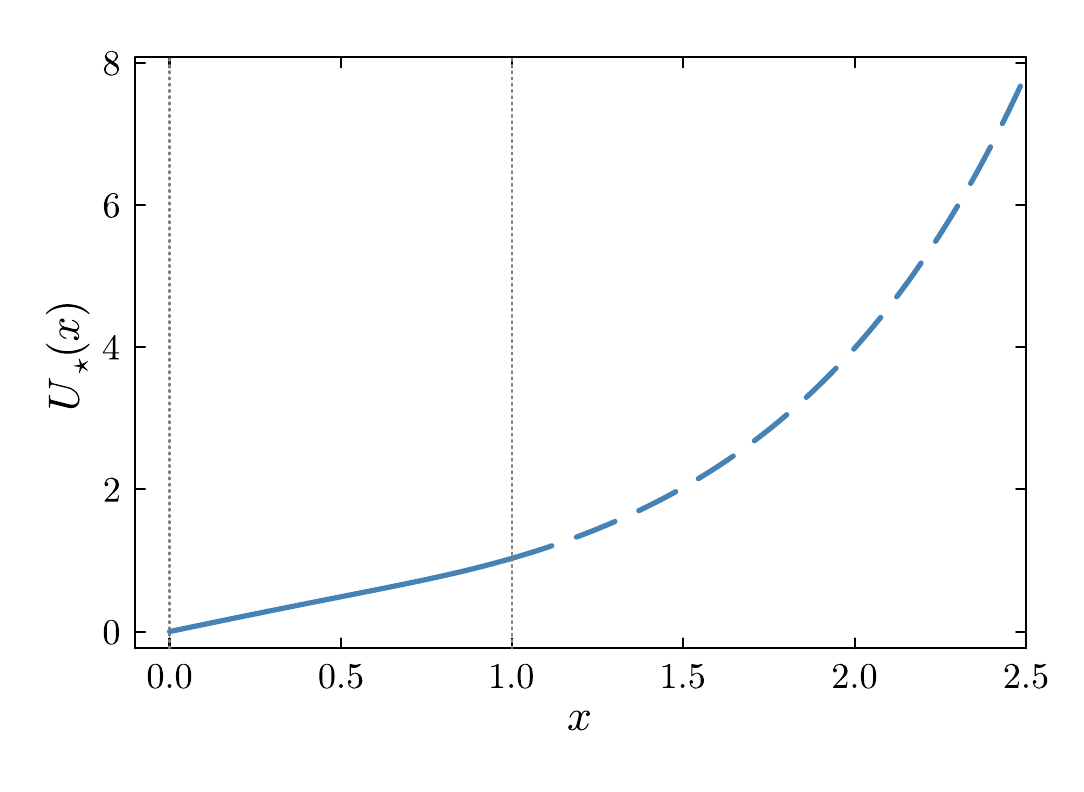}
        \vspace{-1.6em}\subcaption{}
    \end{subfigure}\hfill
    \begin{subfigure}{0.32\linewidth}
        \includegraphics[width=\linewidth]{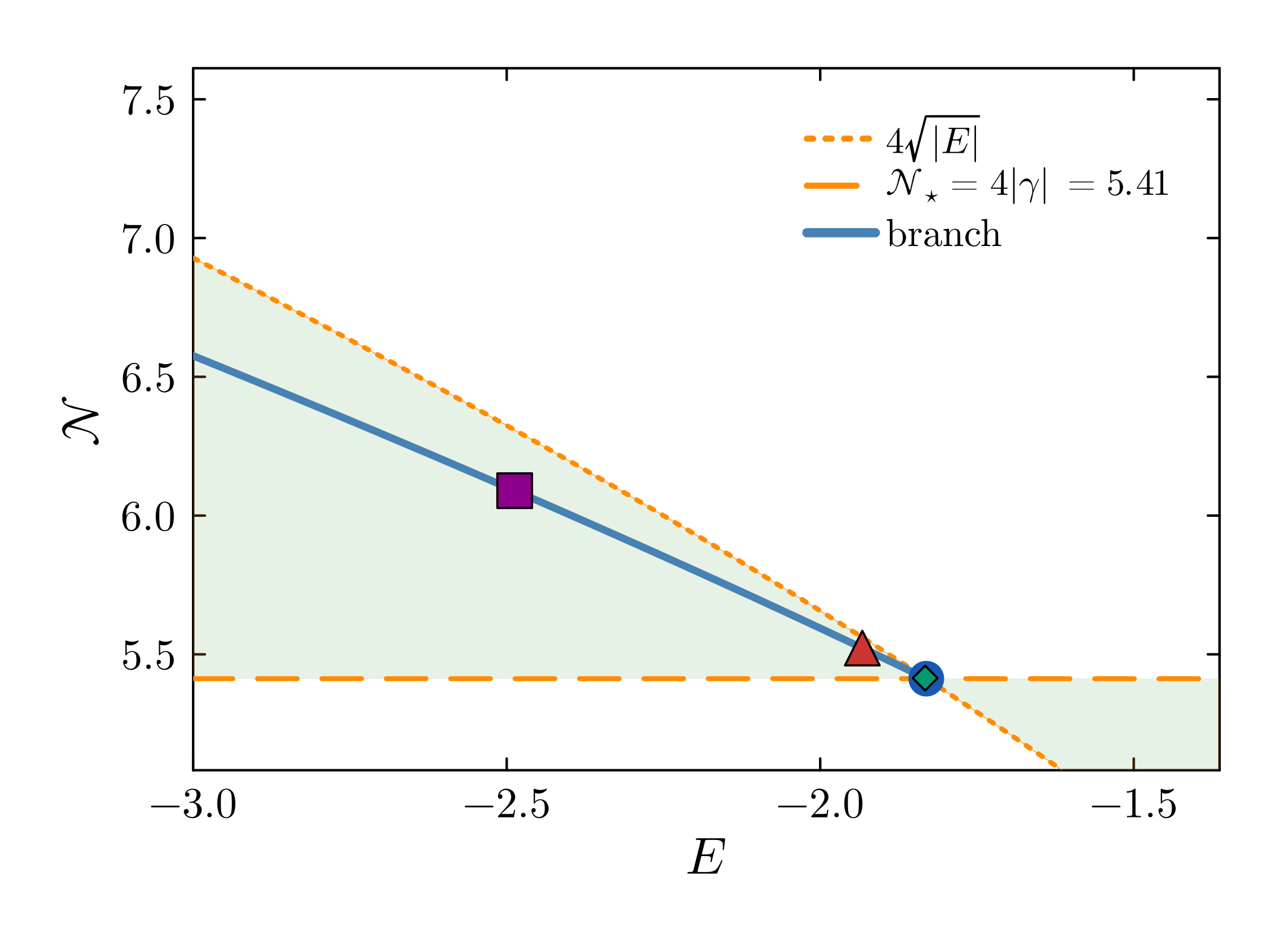}
        \vspace{-1.6em}\subcaption{}
    \end{subfigure}\hfill
    \begin{subfigure}{0.32\linewidth}
        \includegraphics[width=\linewidth]{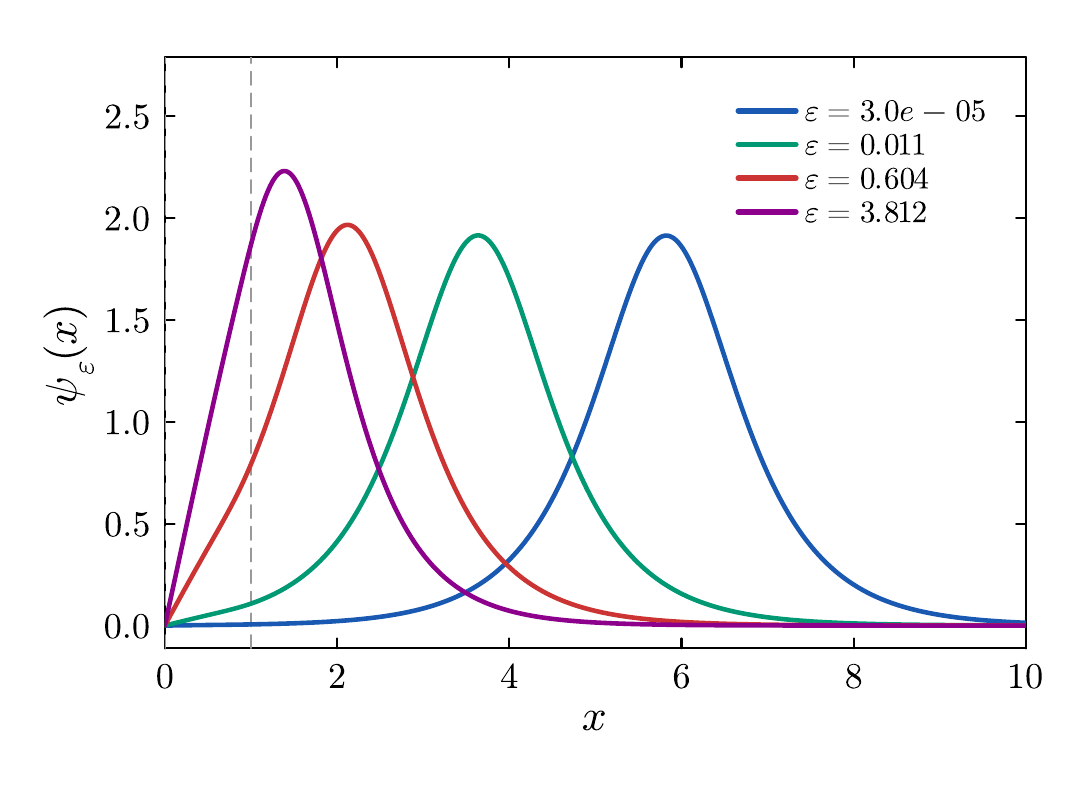}
        \vspace{-1.6em}\subcaption{}
    \end{subfigure}
    \caption{ $X=H_0^1(\R_+)$, scattering resonance for $V(x)$ in \eqref{eq:sr-pot}, $\gs \approx 1.353$. (a) Scattering resonance mode ~$U_\star(x)$. (b)~$\eps\mapsto(E(\eps),\mathcal N[\psi_\eps])$;  branch enters $\Sreg^{\rm SR}$, regions shaded as in Fig.~\ref{fig:schematic}; markers correspond to spatial $\psi_\eps(x)$ profiles in (c)}
    \label{fig:example}
\end{figure}

\begin{figure}[!htbp]
    \centering
    \setlength{\abovecaptionskip}{2pt}
    \setlength{\belowcaptionskip}{0pt}
    \begin{subfigure}{0.32\linewidth}
        \includegraphics[width=\linewidth]{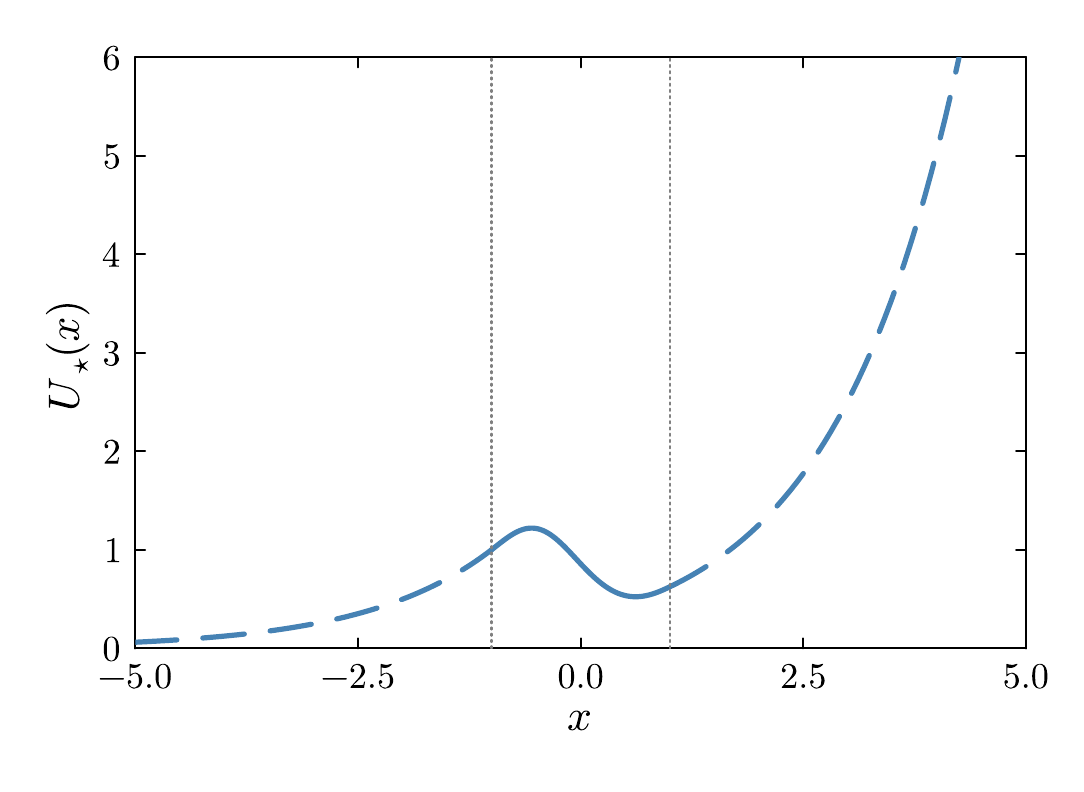}
        \vspace{-1.6em}\subcaption{}
    \end{subfigure}\hfill
    \begin{subfigure}{0.32\linewidth}
        \includegraphics[width=\linewidth]{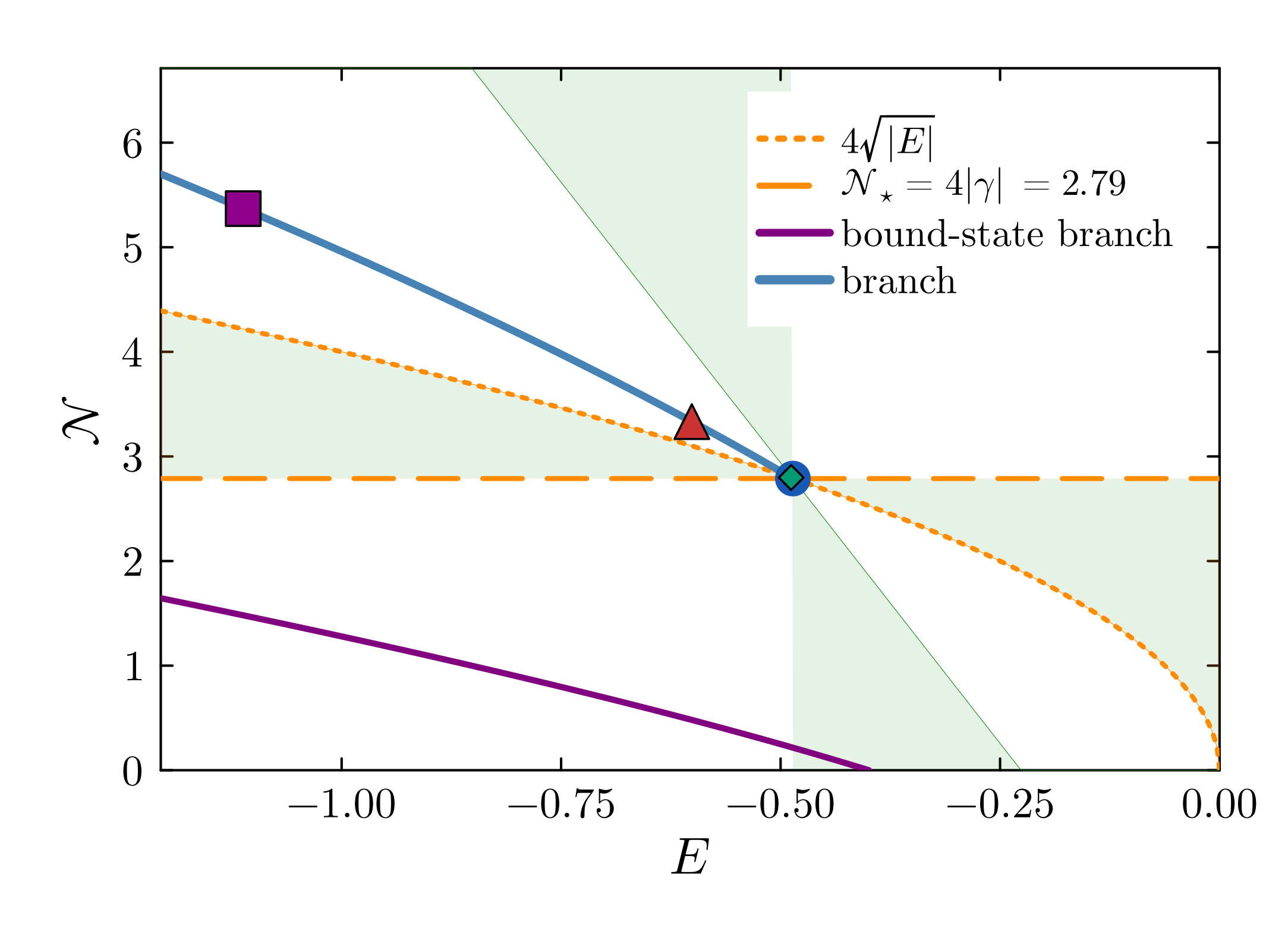}
        \vspace{-1.6em}\subcaption{}
    \end{subfigure}\hfill
    \begin{subfigure}{0.32\linewidth}
        \includegraphics[width=\linewidth]{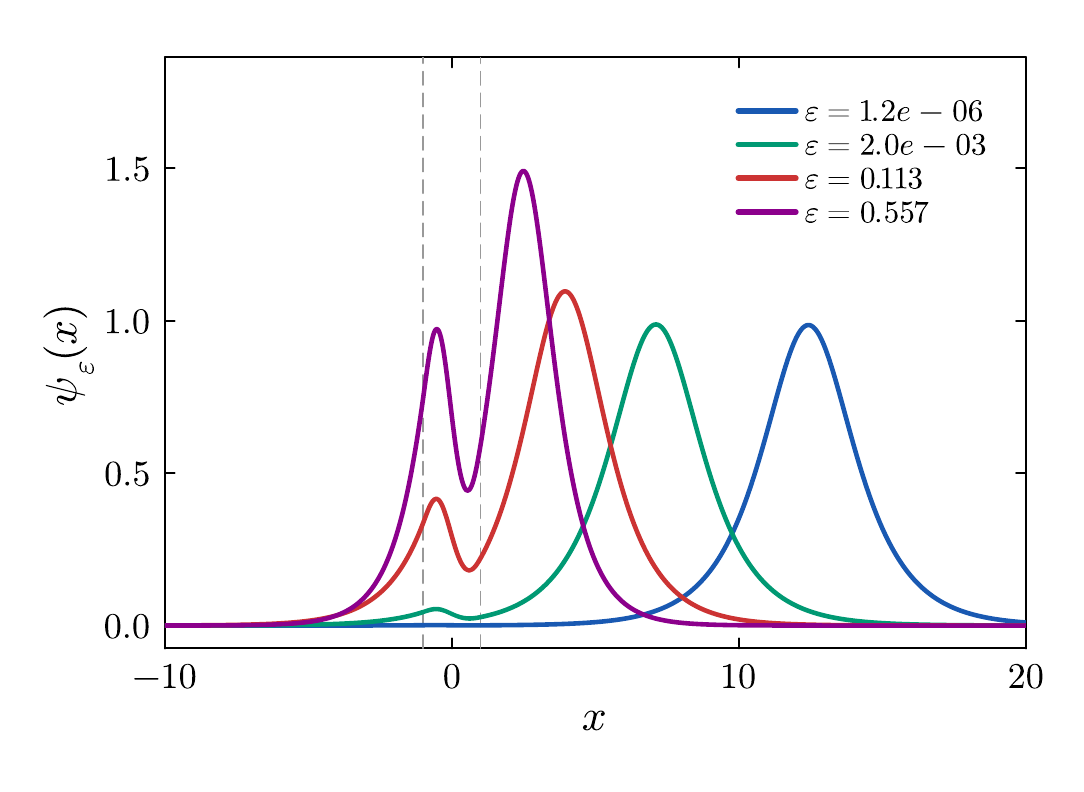}
        \vspace{-1.6em}\subcaption{}
    \end{subfigure}
    \caption{$X=H^1(\R)$, transmission resonance for $V(x) = 3\sin(\pi x)$ on $[-1,1]$, $\gs \approx 0.697$, $\Omega_\star \approx -3.01$. (a) Transmission resonance mode ~$U_\star(x)$. (b)~Bifurcation curves, regions shaded as in Figure~\ref{fig:schematic}; the resonance branch (blue) lies in $\Ureg_1^{\rm TR}$. Purple:  branch from the $L^2$ eigenvalue $E_0 \approx -0.399$. (c)~Profiles $\psi_\eps(x)$.}
    \label{fig:trans-example}
\end{figure}

\clearpage
\subsection{\texorpdfstring{Validation of the perturbed translation eigenvalue asymptotics in Proposition~\ref{prop:mu}}{Corroboration of the asymptotic formula for mu(eps)}}\label{ssec:mu-verification}

We numerically verify our approximation $\mu(\eps) = (3\eps^2/(4\gs))\,\Omega_\star + \mathcal O(\eps^3)$ of the perturbed translation eigenvalue. The reference value of $\mu(\eps)$ is computed by solving the inner IVP~\eqref{eq:phi-ivp} for $\phi_{\rm in}$ on $[-b,b]$ with $\alpha$ determined by the left match exactly (\claimref{1a}); then applying a shooting method to match it with the right side at $x=b$, then gluing to a Jost solution on either side.
Figure~\ref{fig:verify} compares the predicted $\tfrac{3\Omega_\star}{4\gs}\,\eps^2$ with computed $\mu(\eps)$ in the $\eps$-window where the leading-order asymptotic is dominant. 

For the half-line scattering resonance (Figure~\ref{fig:example}), $\Omega_\star\approx 0.022 > 0$ gives $\mu>0$ and stability. For the sine transmission resonance (Figure~\ref{fig:trans-example}), $\Omega_\star\approx -3.01<0$ gives $\mu(\eps)<0$ and spectral instability.

\begin{figure}[!htbp]
    \centering
    \setlength{\abovecaptionskip}{4pt}
    \setlength{\belowcaptionskip}{0pt}
    \begin{subfigure}{0.49\linewidth}
        \includegraphics[width=\linewidth]{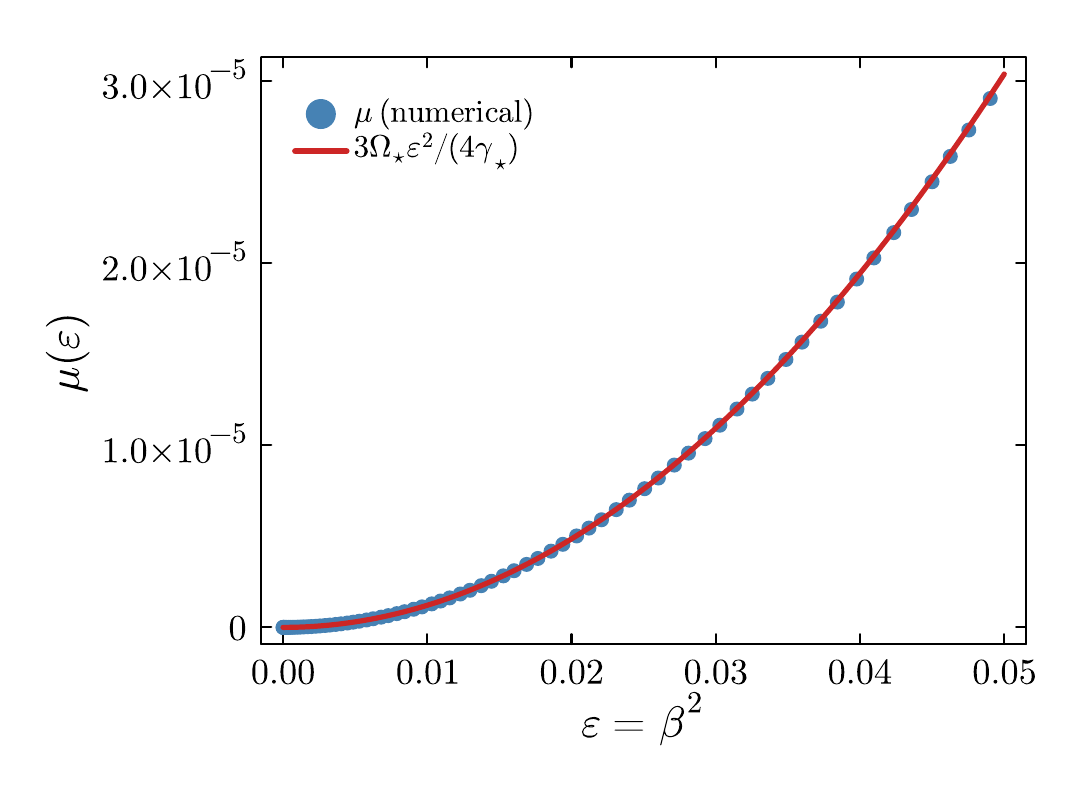}
        \vspace{-2em}\subcaption{}
    \end{subfigure}\hfill
    \begin{subfigure}{0.49\linewidth}
        \includegraphics[width=\linewidth]{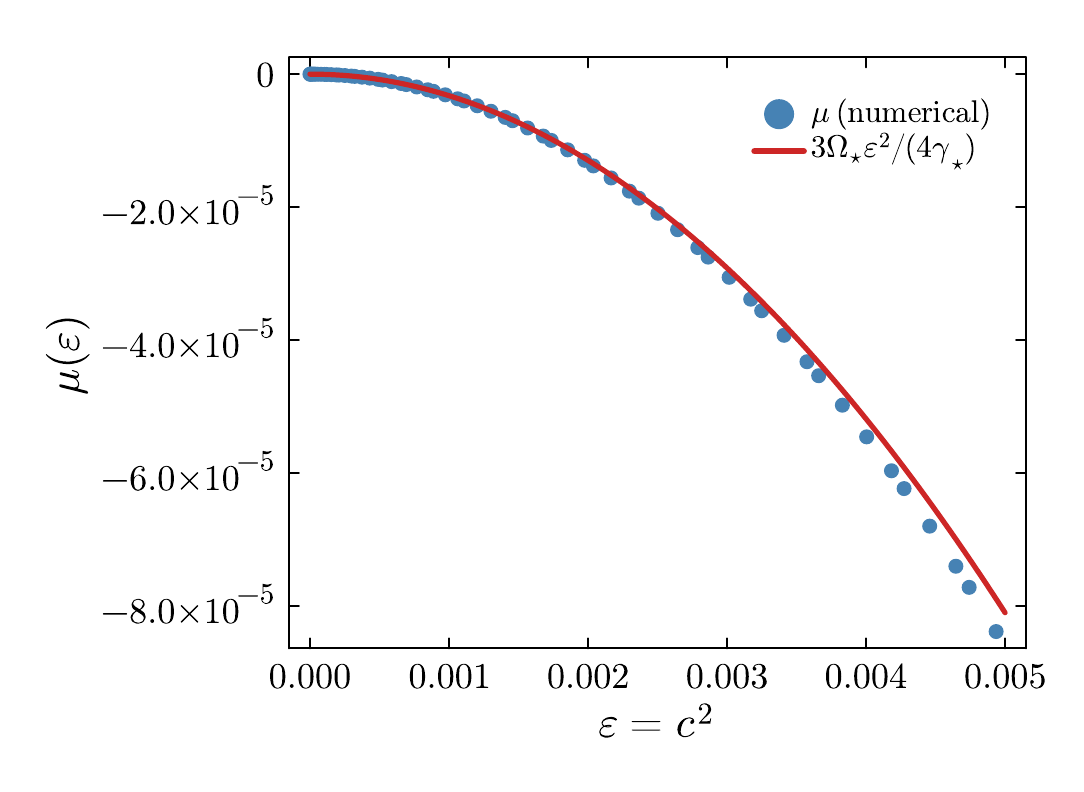}
        \vspace{-2em}\subcaption{}
    \end{subfigure}
\vspace{1em}
      \caption{Validity of approximate expression for $\mu(\eps)$, Proposition~\ref{prop:mu}: direct numerical simulation of $\mu(\eps)$ (dots) vs.\ asymptotic expression $\tfrac{3\Omega_\star}{4\gs}\,\eps^2$ (solid). (a)~Half-line scattering case (Fig.~\ref{fig:example}). (b)~Sine transmission case (Fig.~\ref{fig:trans-example}).}
        \label{fig:verify}
\end{figure}


\section{Long-time dynamics after instability}\label{sec:dynamics}

{
We present finite-difference numerical computations of the spectra of $H_V$ and long-time dynamics of the time-dependent NLS/GP equation~\eqref{eq:nlse} for two choices of linear potentials, $V$, and initial data chosen to be near corresponding unstable transmission resonance-induced bound states. 
The results are summarized in Figure~\ref{fig:dyn-E0} and  Figure~\ref{fig:dyn-E1}, which display (a) the density $x\mapsto|\Psi(x,t)|^2$, over a long time-interval and (b) a nonlinear rescaling, $x\mapsto \operatorname{asinh}(|\Psi(x,t)|^2/\rho_0)$ (with $\rho_0 := 10^{-7}\max_{x,t}|\Psi|^2$), which amplifies the low-amplitude features, such as the small radiative tails of $|\Psi(x,t)|^2$. The Julia code, the bifurcation-branch and dynamics datasets, and animated mp4s of the spacetime evolution for our examples are available in the  repository~\cite{turner2026nls-resonance-stability}.

}

{
\subsection*{Example 1:\ $V(x) = 3\sin(\pi x)\,\mathbf 1_{(-1,1)}(x)$, Figures~\ref{fig:spectral-E0}--\ref{fig:dyn-E0}.}\label{sec:dyn-E0}\quad
\medskip

Numerical solution of the bound state and resonance spectral problems for  $H_V$ yields:
\begin{itemize}
    \item a bound state pole of $r(k)$ at $i\kappa_0$, with $\kappa_0\approx\sqrt{0.399}$, and a corresponding  ground state eigenfunction $U_0>0$ satisfying $H_V U_0=E_0U_0$ with $U_0\in L^2(\R)$, where  $E_0=(i\kappa_0)^2\approx-0.399$. The eigenpair $(U_0,E_0)$ seeds a nonlinearly stable branch of nonlinear bound states \cite{rose1988bound}. 
    \item a transmission resonance zero of $r(k)$ at $-i\gs$, where $\gs\approx 0.697$, with corresponding transmission resonance mode $U_\star>0$, satisfying $H_V U_\star=E_\star U_\star$ with $U_\star\in L^2_{\rm loc}$, where  $E_\star=(i\gs)^2\approx-0.486$. By Theorem \ref{thm:intro-bifurcation} (see also \cite{turner2026resonance}) the transmission resonance pair $(\gs,U_\star)$ seeds a branch of resonance-induced nonlinear bound states. 
     We find that $\Omega_\star \approx -3.01 \;<\; 0$ and hence $n_-(L_+^\eps)=2$. Further, we have $n_-(L_-^\eps)=0$, since $\psi_\eps>0$ satisfies $L_-^\eps\psi_\eps=0$. Therefore, by Theorem~\ref{thm:intro-stability-sharp}, states along the resonance-induced branch which are near the bifurcation point are unstable.
\end{itemize}

Figure~\ref{fig:spectral-E0} displays the potential, the spectral data on the $E=k^2$ axis, and the profiles $U_0$ and $U_\star$. Figure~\ref{fig:dyn-E0} displays the NLS-GP time-evolution for initial data taken to be a numerical approximation to an unstable transmission-resonance induced mode,  $\psi_\eps$, near the bifurcation point. The perturbed (by numerical error) state $\psi_\eps$ appears to resolve, as time advances, into a free soliton part moving off to $+\infty$, a part whose energy is captured by the nonlinear ground state branch, which is asymptotically stable \cite{sw1990,sw1992,pillet-wayne1997},  
and dispersive radiation. 

}
\medskip

\begin{figure}[!htbp]
\centering
\setlength{\abovecaptionskip}{3pt}
\setlength{\belowcaptionskip}{0pt}
\begin{subfigure}[b]{0.42\linewidth}
\parbox[b][4.0cm][c]{\linewidth}{%
\centering
\begin{tikzpicture}[scale=0.88, transform shape]
  \fill[green!10] (0,-0.40) rectangle (1.9,0.40);
  \draw[-{Stealth[length=4pt,width=3pt]}, thin] (-5.0,0) -- (2.2,0)
        node[right, font=\footnotesize] {$E = k^2$};
  \draw[thin] (0,-0.24) -- (0,0.24);
  \node[below, font=\scriptsize] at (0,-0.26) {$0$};
  \node[font=\scriptsize, text=green!60!black, above] at (1.0,0.40) {cont.};
  \filldraw[figgreen, draw=black!75, line width=0.5pt] (-3.19,0) circle (3.0pt);
  \draw[gray!45, line width=0.4pt] (-3.19,0.16) -- (-3.19,0.88);
  \node[above, font=\scriptsize, text=figgreen!85!black,
        fill=white, draw=gray!30, thin, inner sep=1.2pt] at (-3.19,0.88) {$k=i\kappa_0$};
  \node[below, font=\scriptsize] at (-3.19,-0.26) {$E_0$};
  \draw[zeroblue, line width=1.6pt, line cap=round]
    (-3.89-0.16,-0.16)--(-3.89+0.16,0.16)
    (-3.89-0.16, 0.16)--(-3.89+0.16,-0.16);
  \draw[gray!45, line width=0.4pt] (-3.89,-0.16) -- (-3.89,-0.88);
  \node[below, font=\scriptsize, text=zeroblue,
        fill=white, draw=gray!30, thin, inner sep=1.2pt] at (-3.89,-0.88) {$k={-}i\gs$};
  \node[above, font=\scriptsize] at (-3.89,0.26) {$E_\star$};
\end{tikzpicture}%
}%
\subcaption{}
\end{subfigure}\hfill
\begin{subfigure}[b]{0.55\linewidth}
\centering
\includegraphics[width=\linewidth]{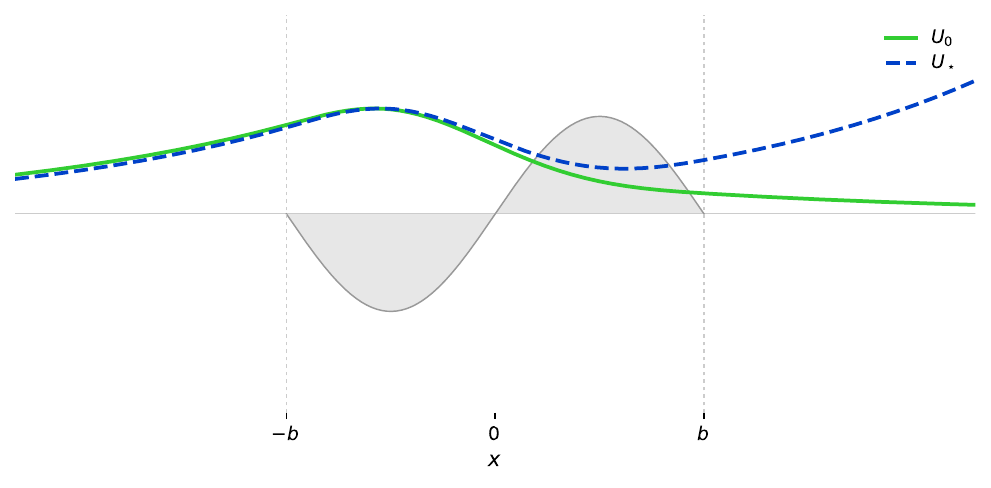}
\subcaption{}
\end{subfigure}
\caption{%
  $V(x)=3\sin(\pi x)\,\mathbf{1}_{(-1,1)}(x)$, $b=1$.
  \textbf{(a)}~Spectral line $E=k^2$: \specpole\ $=$ $L^2$ bound state
  at $k=i\kappa_0$ ($E_0$); \speczero\ $=$ TR zero at $k=-i\gs$ ($E_\star$);
  shaded band $=$ continuum $E\ge0$.
  \textbf{(b)}~Potential $V$ (gray, scaled), ground state $U_0$ (solid),
  and resonance mode $U_\star$ (dashed).%
}
\label{fig:spectral-E0}
\end{figure}
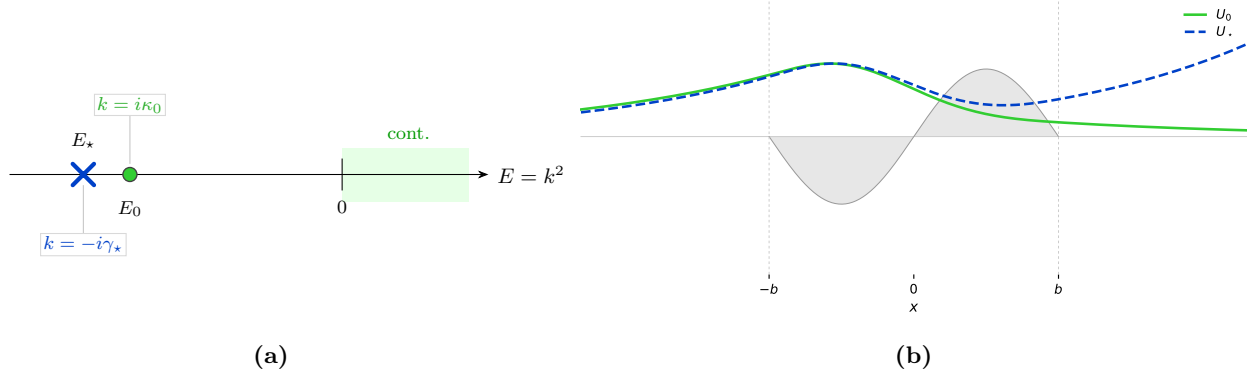

\begin{figure}[!htbp]
    \centering
    \setlength{\abovecaptionskip}{4pt}
    \setlength{\belowcaptionskip}{0pt}
    \includegraphics[width=0.62\linewidth]{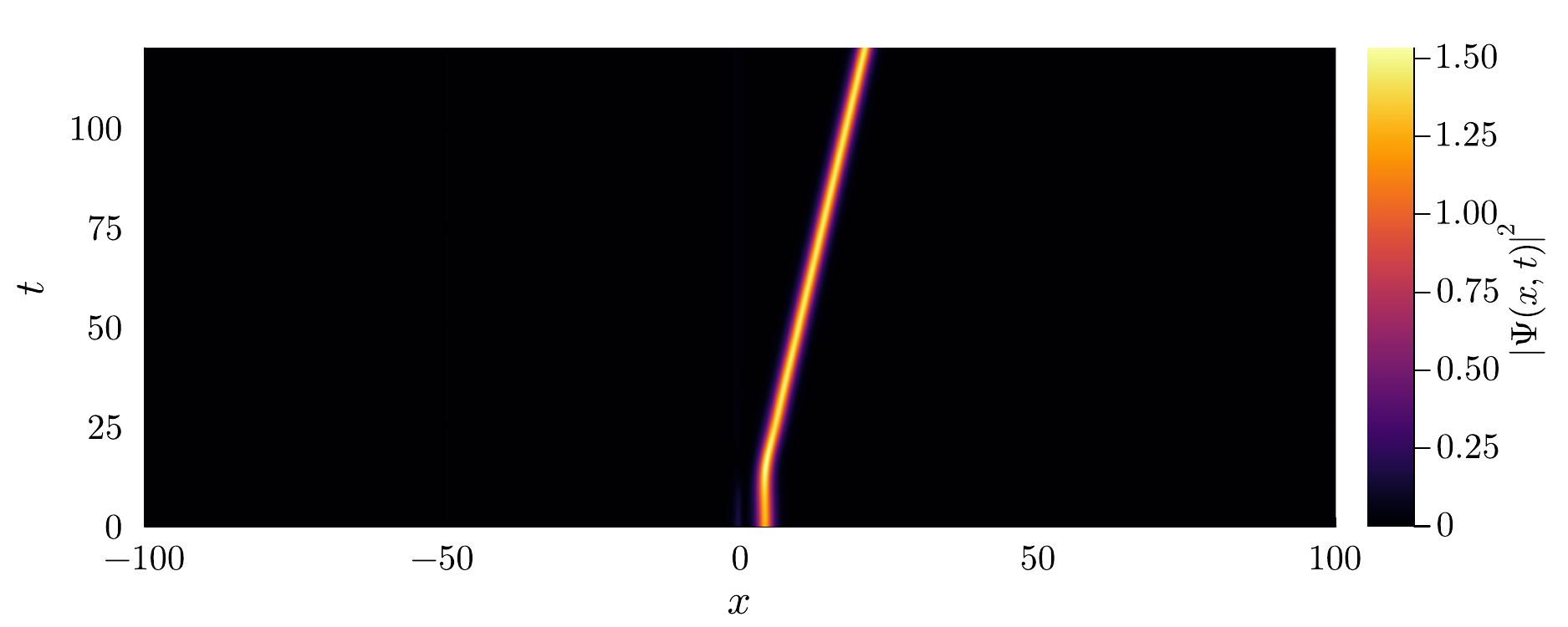}\\[1pt]
    \includegraphics[width=0.62\linewidth]{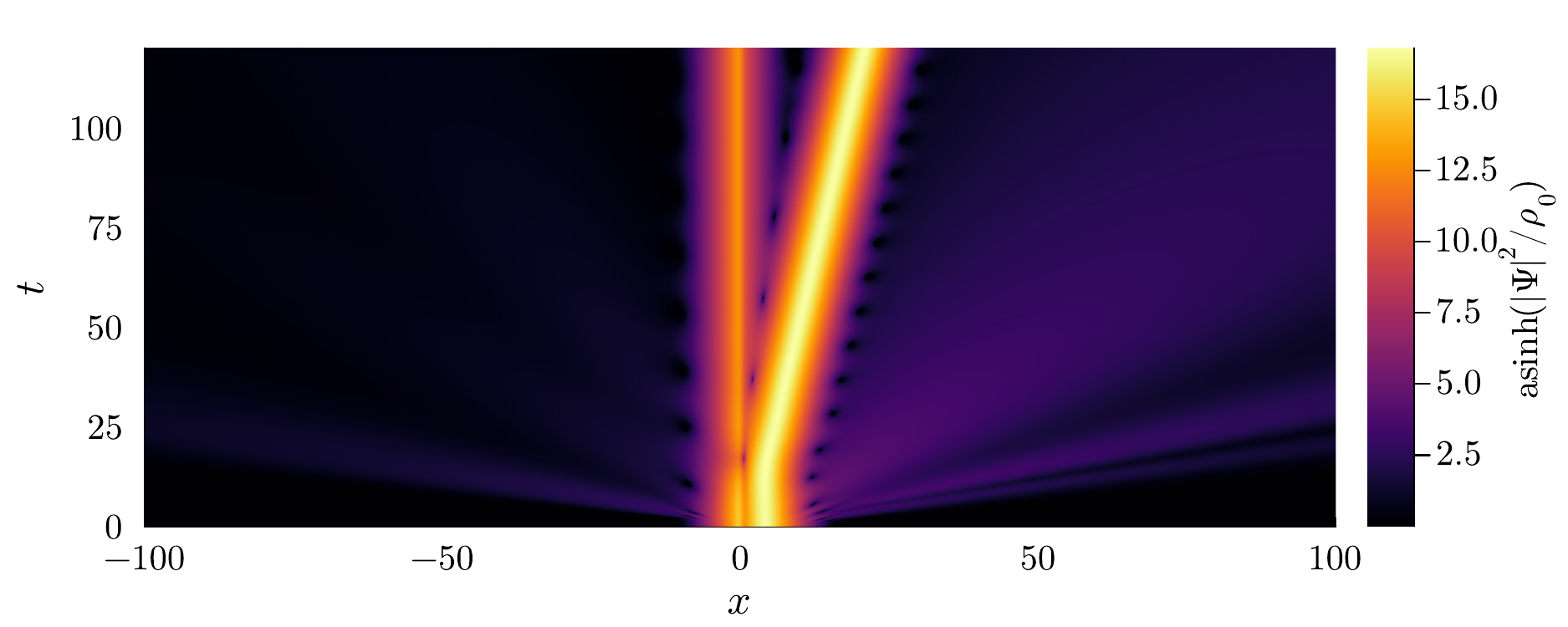}
    \caption{$V=3\sin(\pi x)\mathbf 1_{(-1,1)}$ at $T=120$; $\mathcal N_{\mathrm{rad}} \approx 4\!\times\!10^{-4}$. Top: raw density $|\Psi(x,t)|^2$. Bottom: rescaled $\operatorname{asinh}(|\Psi|^2/\rho_0)$ exposing low-amplitude radiation.}
    \label{fig:dyn-E0}
\end{figure}

{
\subsection*{Example 2: $V(x) = 5\cos(\tfrac{3\pi}{2}x)\,\mathbf 1_{(-1,1)}(x)$, Figures~\ref{fig:spectral-E1}--\ref{fig:dyn-E1}.}
 \label{sec:dyn-E1}\quad
\medskip

Numerical solution of the bound state and resonance spectral problems for  $H_V$ yields:
\begin{itemize}
    \item two bound state poles of $r(k)$ at $i\kappa_0$ and $i\kappa_1$ with $\kappa_0\approx\sqrt{0.727}$ and  
    $\kappa_1\approx\sqrt{0.071}$ and corresponding  ground state and excited state eigenfunctions $U_0>0$ and $U_1$, respectively,  satisfying $H_V U_j=E_jU_j$ with $U_j\in L^2(\R)$, where  $E_0=(i\kappa_0)^2\approx-0.727$ and $E_1=(i\kappa_1)^2\approx-0.071$. Each bound state seeds a branch of  nonlinear bound states \cite{rose1988bound}. 
    
    For NLS-GP, generic small data are expected to evolve, as $t\to\infty$, toward a nonlinear ground state and outgoing radiation; the nonlinear excited state branch is unstable; see  \cite{sw2004,weinstein-survey:15} for a study of a three-dimensional setting.
    \item a transmission resonance zero of $r(k)$ at $-i\gs$, where $\gs\approx 0.49347$, with corresponding transmission resonance mode $U_\star$,  satisfying $H_V U_\star=E_\star U_\star$ with $U_\star\in L^2_{\rm loc}$, where  $E_\star=(i\gs)^2\approx-0.244$. In this case, $U_\star$ changes its sign once. By Theorem \ref{thm:intro-bifurcation} (see also \cite{turner2026resonance}) the transmission resonance pair $(\gs,U_\star)$ seeds a branch of resonance-induced nonlinear bound states $\eps\mapsto\big(E(\eps),\psi_\eps\big)$. To investigate the stability of this branch near the bifurcation point, note that since $U_\star(x)$ changes sign once, $\psi_\eps(x)$  changes sign once. And since $L_-^\eps\psi_\eps=0$, zero is the second eigenvalue of $L_-^\eps$ and we have that $n_-(L_-^\eps)=1$.

    Turning now to $\Lop$, for $\eps>0$ and small we have, by perturbation arguments, that $\Lop$ has one negative eigenvalue very near to the ground state eigenvalue of  $L_+^{\rm sol}$ (Section \ref{ssec:free-soliton}). Further, since $\Omega_\star \approx -0.749 <0$, the perturbed  zero energy (translation mode) eigenvalue of $L_+^{\rm sol}$ is strictly negative. And finally, for $\eps>0$ and small,  there is a strictly negative eigenvalue of $\Lop$ near one of the two eigenvalues of $H_V$ (since $E_0+\gs^2 \approx -0.484 \;<\; 0$ and $E_1+\gs^2 \approx 0.172 \;>\; 0$). Therefore,  $n_-(L_+^\eps) \;\ge\; 3$. Therefore,
\[
  n_-(L_+^\eps)-n_-(L_-^\eps) \ge 3-1 = 2 \;\notin\; \{0,1\},
\]
Theorem~\ref{thm:W-GSS}(ii)(\ref{wgss-b}) implies instability of the transmission-resonance branch.
\end{itemize}

Figure~\ref{fig:spectral-E1} displays the potential, the spectral data on the $E=k^2$ axis, and the profiles $U_0$, $U_1$, and $U_\star$. Figure~\ref{fig:dyn-E1} presents a simulation analogous to that in Example 1. The solution evolves away from the transmission-resonance branch and resolves into: a part which is localized around the support of $V$, appearing to be captured by the nonlinear excited state branch (seeded by the eigenstate $(E_1,U_1)$ of $H_V$), and transporting free-soliton and dispersive  radiation parts. We believe that, by a mechanism of  nonlinear resonance of the excited state with radiation modes, the nonlinear excited state component observed in  Figure~\ref{fig:dyn-E1} will eventually decay into the nonlinear ground state channel and dispersive radiation; see \cite{sw2004,weinstein-survey:15}.

}

\begin{figure}[!htbp]
\centering
\setlength{\abovecaptionskip}{3pt}
\setlength{\belowcaptionskip}{0pt}
\begin{subfigure}[b]{0.42\linewidth}
\parbox[b][4.0cm][c]{\linewidth}{%
\centering
\begin{tikzpicture}[scale=0.73, transform shape]
  \fill[green!10] (0,-0.40) rectangle (1.9,0.40);
  \draw[-{Stealth[length=4pt,width=3pt]}, thin] (-6.5,0) -- (2.2,0)
        node[right, font=\footnotesize] {$E = k^2$};
  \draw[thin] (0,-0.24) -- (0,0.24);
  \node[below, font=\scriptsize] at (0,-0.26) {$0$};
  \node[font=\scriptsize, text=green!60!black, above] at (1.0,0.40) {cont.};
  \filldraw[figgreen, draw=black!75, line width=0.5pt] (-5.82,0) circle (3.0pt);
  \draw[gray!45, line width=0.4pt] (-5.82,0.16) -- (-5.82,0.88);
  \node[above, font=\scriptsize, text=figgreen!85!black,
        fill=white, draw=gray!30, thin, inner sep=1.2pt] at (-5.82,0.88) {$k=i\kappa_0$};
  \node[below, font=\scriptsize] at (-5.82,-0.26) {$E_0$};
  \draw[zeroblue, line width=1.6pt, line cap=round]
    (-1.95-0.16,-0.16)--(-1.95+0.16,0.16)
    (-1.95-0.16, 0.16)--(-1.95+0.16,-0.16);
  \draw[gray!45, line width=0.4pt] (-1.95,-0.16) -- (-1.95,-0.88);
  \node[below, font=\scriptsize, text=zeroblue,
        fill=white, draw=gray!30, thin, inner sep=1.2pt] at (-1.95,-0.88) {$k={-}i\gs$};
  \node[above, font=\scriptsize] at (-1.95,0.26) {$E_\star$};
  \filldraw[figgreen, draw=black!75, line width=0.5pt] (-0.57,0) circle (3.0pt);
  \draw[gray!45, line width=0.4pt] (-0.57,0.16) -- (-0.57,0.88);
  \node[above, font=\scriptsize, text=figgreen!85!black,
        fill=white, draw=gray!30, thin, inner sep=1.2pt] at (-0.57,0.88) {$k=i\kappa_1$};
  \node[below, font=\scriptsize] at (-0.57,-0.26) {$E_1$};
\end{tikzpicture}%
}%
\subcaption{}
\end{subfigure}\hfill
\begin{subfigure}[b]{0.55\linewidth}
\centering
\includegraphics[width=\linewidth]{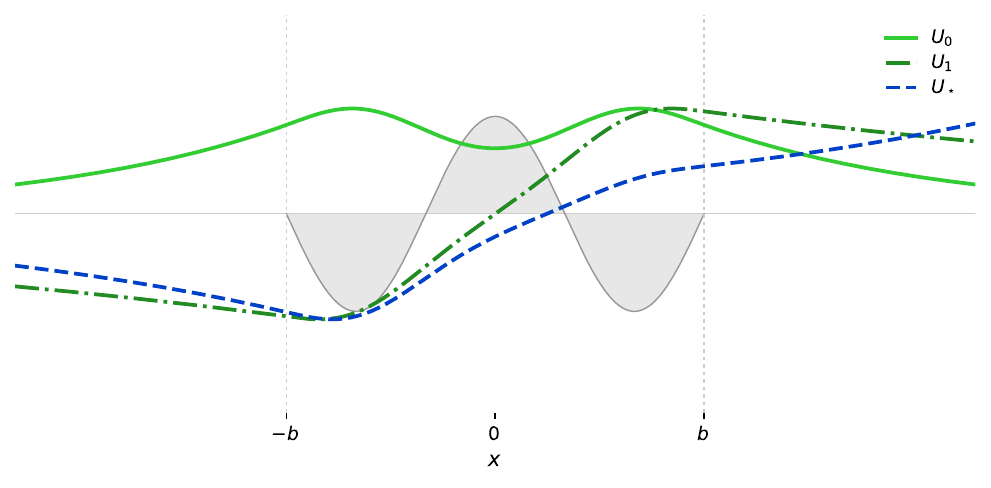}
\subcaption{}
\end{subfigure}
\caption{%
  $V(x)=5\cos\!\bigl(\tfrac{3\pi}{2}x\bigr)\,\mathbf{1}_{(-1,1)}(x)$, $b=1$.
  \textbf{(a)}~Spectral line $E=k^2$: \specpole\ $=$ $L^2$ bound states at
  $k=i\kappa_0$ ($E_0$) and $k=i\kappa_1$ ($E_1$); \speczero\ $=$
  TR zero at $k=-i\gs$ ($E_\star$); shaded band $=$ continuum $E\ge0$.
  \textbf{(b)}~Potential $V$ (gray, scaled), eigenfunctions $U_0$ (solid),
  $U_1$ (dash-dot), and resonance mode $U_\star$ (dashed);
  $U_\star$ changes sign once inside $\operatorname{supp}V$.%
}
\label{fig:spectral-E1}
\end{figure}
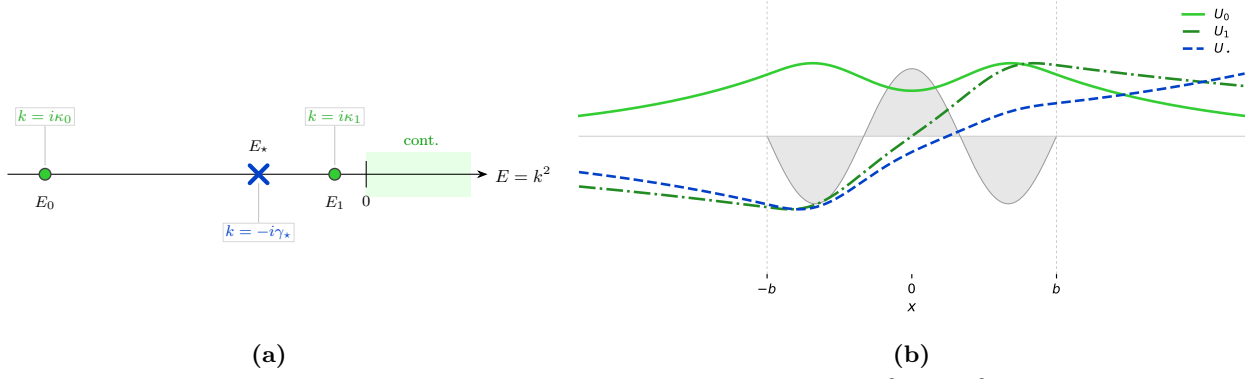

\begin{figure}[!htbp]
    \centering
    \setlength{\abovecaptionskip}{4pt}
    \setlength{\belowcaptionskip}{0pt}
    \includegraphics[width=0.62\linewidth]{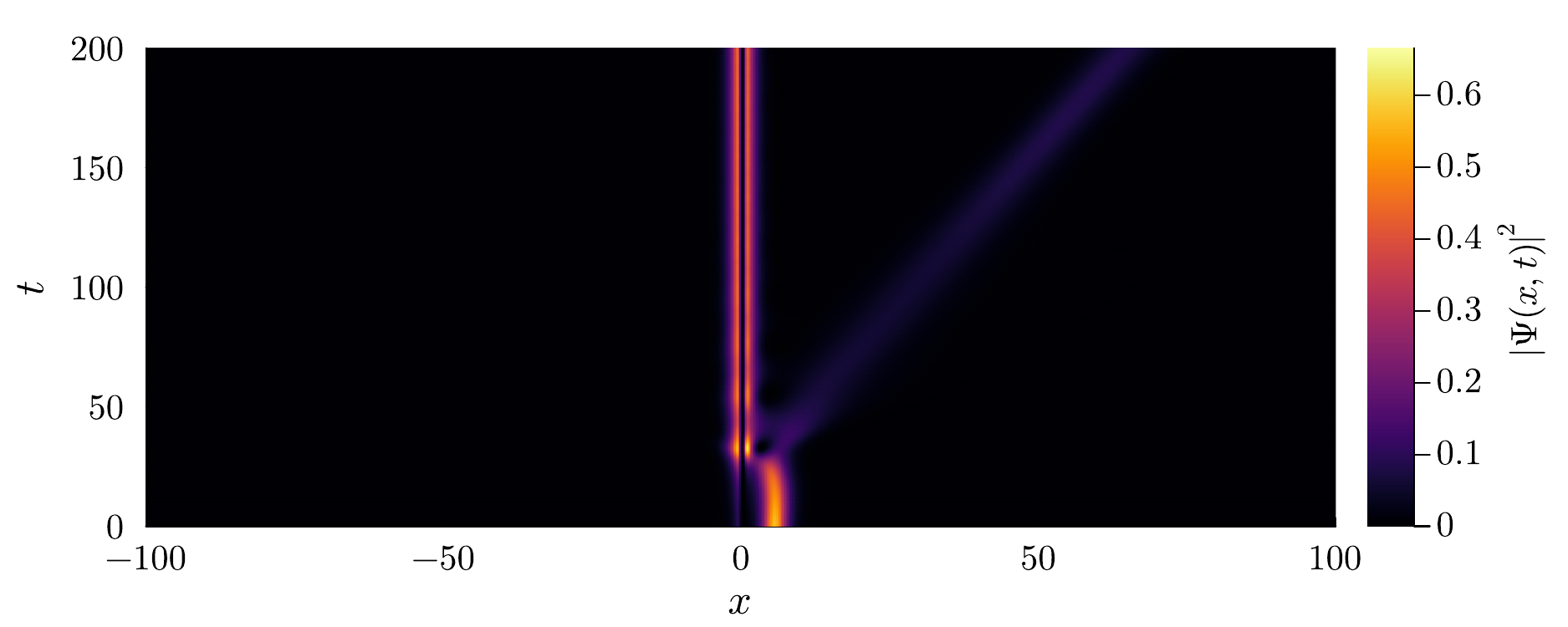}\\[1pt]
    \includegraphics[width=0.62\linewidth]{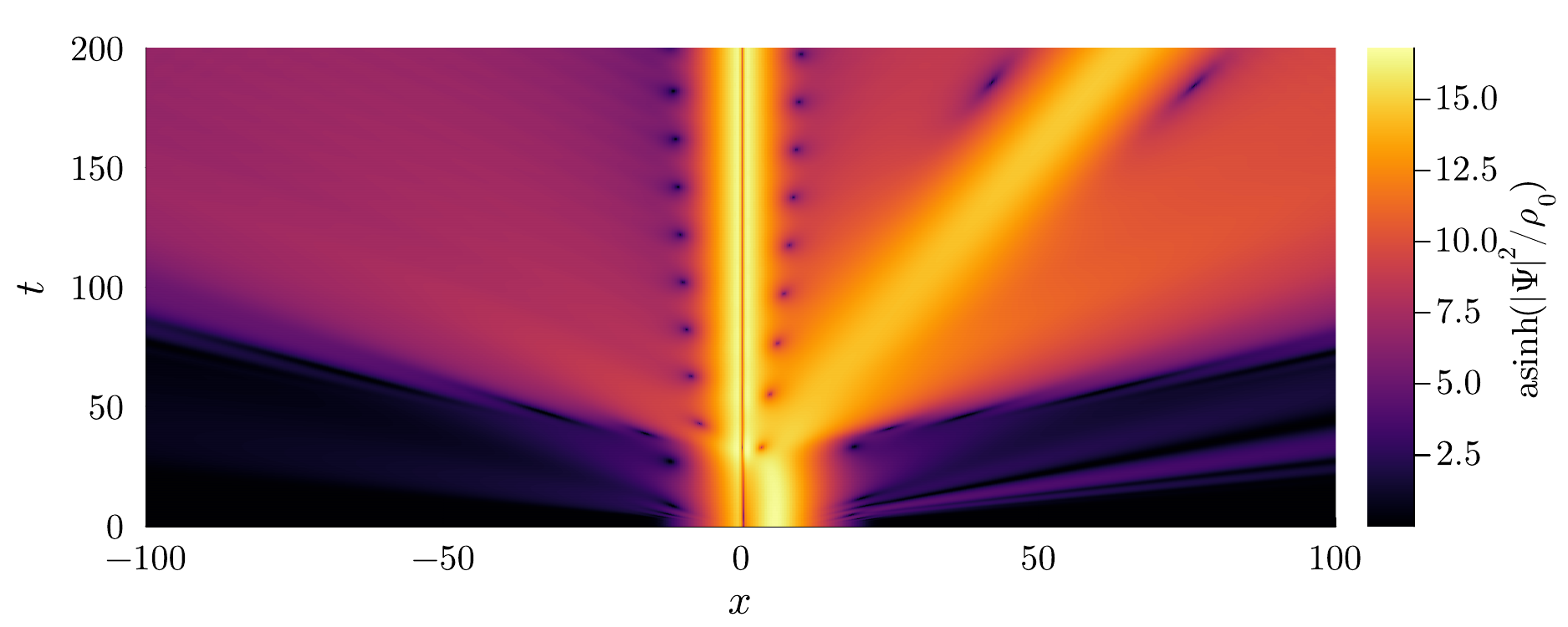}
    \caption{$V=5\cos(\tfrac{3\pi}{2}x)\mathbf 1_{(-1,1)}$. Top: raw density. Bottom: rescaled $\operatorname{asinh}(|\Psi|^2/\rho_0)$; a diffuse radiation fan around the soliton track is visible in the rescaled view but invisible in the raw view.}
    \label{fig:dyn-E1}
\end{figure}

  \section{Discussion and Future Directions}\label{sec:discussion}

Scattering and transmission resonances of $H_V$ on the imaginary axis induce nonlinear bound states of the focusing one-dimensional NLS/GP above a positive $L^2$ excitation threshold. Theorem~\ref{thm:intro-stability-sharp} gives explicit conditions for orbital stability and instability near the bifurcation point, in terms of positivity of the resonance mode and the scalar functionals $\Omega_\star$ and $\nu_\star$. Corollary~\ref{cor:small-gamma-stability} shows that sufficiently near-threshold resonances produce orbitally stable nonlinear bound states.

Our spectral analysis extends to the general class of focusing power nonlinearities $|\Psi|^{2\sigma}\Psi$ with $\sigma > 0$; the bifurcation construction of resonance-induced nonlinear bound states in~\cite{turner2026resonance} extends, as does the stability/instability analysis based on the spectral theory of Schr\"odinger operators with P\"oschl--Teller potentials; see~\cite{LandauLifshitzQM}, \cite{Flugge1971}. For perturbations of the nonlinearity, such as $|\psi|^2\psi\to |\psi|^2\psi - \eta |\psi|^4\psi$ with $\eta$ small, we expect the same conclusions to persist by perturbative arguments. More challenging would be to extend the present work to nonlinearities for which the soliton part of the nonlinear bound state, outside the support of $V$, has no explicit analytical form and no accompanying explicit spectral theory.

\subsection*{Future Directions}
Our article \cite{turner2026resonance}, on the bifurcation of resonance-induced nonlinear bound states,  outlines an agenda for studying: non-compactly supported $V$, periodic and discrete analogues, multiple bifurcations from a threshold resonance, variational characterizations, time-dependent scattering. We single out a number of directions most directly related to the present stability work.

\begin{enumerate}
\item \emph{Variational perspective.} In \cite{turner2026columbia-thesis} and in a forthcoming work \cite{tw-variational}, we make use of our bifurcation construction of resonance induced states to explore the question of when they are actually global minimizers of the Hamiltonian, $\mathcal{H}$, subject to fixed $\mathcal{N}$.
    \item \emph{``Islands of stability'' in higher dimensions.} Our original motivation  was a multidimensional setting; numerical experiments in the $d = 3$ NLS/GP reveal a resonance-induced phenomenon: when a scattering resonance is near zero energy (similar to Corollary~\ref{cor:small-gamma-stability}), ranges of optical power, $\mathcal{N}$, emerge within which  nonlinear bound states are stabilized. This scenario was discussed in \cite{rose1988bound,weinstein89}, and remains an important open problem.

    \item \emph{Linear stability analysis of the Hamiltonian flow.} In our matching approach, we obtain an explicit asymptotic expression for the small perturbed translation eigenvalue $\mu(\eps)$ of $\Lop$. Combined with a computation of $\left.d\mathcal N/dE\right|_{\eps=0}$, this yields Theorem~\ref{thm:W-GSS}. The same approach could, in principle, be applied to the linearization of the Hamiltonian time-dynamics:
    \[
      J\mathcal L^\eps \;:=\; \begin{pmatrix} 0 & L_-^\eps \\ -\Lop & 0 \end{pmatrix},
    \]
    whose spectrum governs linear stability. In principle, our approach, together with the scattering theory of  P\"oschl--Teller potentials, could be applied to study the spectrum of $J\mathcal L^\eps$. This could lead to an understanding of Hamiltonian eigenvalue collisions and possible Hamiltonian Hopf bifurcations. 

\item \emph{Center-stable manifold description of long-time dynamics.} Center-stable manifolds and global asymptotic dynamics around ground states of the focusing NLS without potential have been studied in~\cite{beceanu2008centre,nakanishi2012global}, and asymptotic-stability results for nonlinear bound states bifurcating from point spectra appear in~\cite{sw1990,sw1992,buslaev1995scattering,pillet-wayne1997,sw2004,goodman2004strong,kirr2007asymptotic,cuccagna2008asymptotic,kirr2009asymptotic,Chen2026,Cuccagna2001,CuccagnaMaeda2021}.  How these constructions extend to the unstable resonance-induced branches studied here is an open question. The numerics in Section~\ref{sec:dynamics} (Figures~\ref{fig:dyn-E0}--\ref{fig:dyn-E1}) demonstrates behavior such a theory should seek to classify.

    \item \emph{Off-axis resonances.} Our construction is restricted to resonance frequencies $k_\star \in i\R$, those for which the energy $E_\star=k_\star^2$ is negative. Typically Schr\"odinger operators also have resonances which are off the imaginary axis, complex poles and zeros of $k\mapsto r_-(k)$ in the complex plane. The work in~\cite{sacchetti2023tunnel} indicates that such off-axis resonances can  produce long-lived nonlinear states, which are stabilized by the nonlinearity, although the mechanism appears to be  different from that explored in this article. In  \cite{turner2026resonance}, distinct purely imaginary resonances, which give rise to distinct bifurcation curves terminating abruptly at $L^2$ thresholds, may merge (under deformation of the potential) to produce a single smooth bifurcation curve. An understanding of these bifurcation phenomena and their manifestations in the time-dynamics would be of great interest.
\end{enumerate}

\appendix

\section{Expansion of $D(\eps,\mu)$ for $\eps$ and $\mu$ small}
\label{sec:D-expand}


{
In this appendix we derive  Taylor expansion through second-order of $D(\eps,\mu)$, defined in ~\eqref{eq:Ddef}. This then establishes ~\eqref{eq:D-parabolic}, used in the proof of Claim~(1c).

\begin{prop}\label{prop:D-expand}
$D(\eps,\mu)$, given by \eqref{eq:Ddef}, is analytic on a neighborhood $\mathcal Q=\{|\eps|<C_1,\ |\mu|<C_1\}$ of $(0,0)$, and
 has the expansion:
\begin{equation}\label{eq:D-parabolic-derived}
  D(\eps,\mu) \;=\; -\,4\gs\,\mu \;+\; 3\,\Omega_\star\,\eps^2 \;+\; \mathcal O(\eps^3 + \eps\mu + \mu^2),\qquad (\eps,\mu)\in\mathcal Q \ ;
\end{equation}
see~\eqref{eq:D-parabolic} of Claim~(1c). Here $\Omega_\star=\Omega_\star^{\rm TR}$ is displayed in ~\eqref{eq:Omega}.
\end{prop}
\begin{proof} We write the expression for $D(\eps,\mu)$ in ~\eqref{eq:Ddef} as
\begin{equation}\label{eq:D-PWQ}
  D(\eps,\mu) \;=\; -4\gs^2\,\bigl[\,p(m,t_R)\,(\win^R+\la) - q\,\bigr]\ ,
\end{equation}
where 
$p(m,t_R)=p\big(m(\mu,\eps),t_R(\eps)\big)$ is a function of $(\eps,\mu)$ through $m=m(\eps,\mu)$ and $t_R=t_R(\eps)$, and 
\begin{equation}\label{eq:PWQ-def}
  q=q(\eps,\mu) \;:=\; \la_0(\eps)\bigl(1-t_R^2(\eps)\bigr)\bigl(3m(\eps,\mu)+6t_R(\eps)\bigr),
\end{equation}
\medskip

\noindent\underline{\emph{Analyticity of $D(\mu,\eps)$}}\ Since $-E(0)=\gs^2>0$, the maps $\la_0(\eps)=\sqrt{-E(\eps)}$, $\la(\eps,\mu)=\sqrt{\la_0^2-\mu}$ and $m=\la/\la_0$ are real-analytic near $(0,0)$. By~\eqref{eq:1ptR}, $t_R=-\sqrt{1-\psi_\eps(b)^2/2\la_0^2}$ has a real-analytic radicand equal to $1+\mathcal O(\eps)$, hence is real-analytic (in particular the logarithmic singularity of $y_R$ does not enter $D$). Finally $\win^R=\phi_{\rm in}'(b)/\phi_{\rm in}(b)$ is real-analytic in $(\eps,\mu)$ by real-analytic dependence of the IVP~\eqref{eq:phi-ivp} on its parameters, together with $\phi_{\rm in}(b)=\Ub+\mathcal O(\eps+|\mu|)\neq0$ from~\eqref{eq:phi-in-pm-b}. As $p(m,t_R)$, $q$ and $\win^R+\la$ are polynomial (resp.\ real-analytic) in these quantities, $D$ is real-analytic on $\mathcal Q$. Consequently for $(\mu,\eps)$ near $(0,0)$:
\begin{equation}\label{eq:D-taylor-abstract}
  D(\eps,\mu) = D(0,0) + \partial_\eps D(0,0)\,\eps + \partial_\mu D(0,0)\,\mu + \tfrac12\partial_\eps^2 D(0,0)\,\eps^2 + \mathcal O(\eps^3+\eps\mu+\mu^2).
\end{equation}
We now compute the Taylor coefficients through second order by directly computing derivatives of $D(\mu,\eps)$ at $(0,0)$, and using asymptotic formulae for $\mu$ and $\eps$ small.

\medskip
\noindent\underline{\emph{Preliminary expansions}}\ Recall from~\eqref{eq:m-expand-b},~\eqref{eq:la-expand},~\eqref{eq:win-expand} and the tail expansion
\begin{equation}\label{eq:1ptR}
  t_R(\eps) := t(\xi)\Big|_{\xi=b-y_R(\eps)} = -\sqrt{1-\tfrac{\psi_\eps(b)^2}{2\la_0^2(\eps)}}
  \;=\; -1 + \frac{\Ub^2}{4\gs^2}\,\eps + \mathcal O(\eps^2)
\end{equation}
(obtained by substituting~\eqref{eq:psi_eps-expand} and~\eqref{eq:la0}) that 
\[ m=1-\tfrac{\mu}{2\gs^2}+\mathcal O(\eps\mu+\mu^2)\quad {\rm and}\quad  t_R=-1+\tfrac{\Ub^2}{4\gs^2}\eps+\mathcal O(\eps^2).\] Substituting these into $p(m,t)=m^2-1+3mt+3t^2$ (whose gradient at $(m,t)=(1,-1)$ is $(\partial_m p,\partial_t p)=(2m+3t,\,3m+6t)|_{(1,-1)}=(-1,-3)$) gives
\begin{equation}\label{eq:wout-den-R}
  p(m,t_R) \;=\; \frac{\mu}{2\gs^2} - \frac{3\Ub^2}{4\gs^2}\,\eps + \mathcal O(\eps^2 + \eps\mu + \mu^2)
  \;=\; -\,\frac{3\Ub^2\eps - 2\mu}{4\gs^2} + \mathcal O(\eps^2 + \eps\mu + \mu^2)\ .
\end{equation}
Further, using $1-t_R^2=(1-t_R)(1+t_R)=\tfrac{\Ub^2}{2\gs^2}\eps+\mathcal O(\eps^2)$, we have 
\begin{equation}\label{eq:wout-num}
  q(\mu,\eps) \;=\; \la_0(1-t_R^2)(3m + 6t_R) \;=\; -\frac{3\Ub^2}{2\gs}\,\eps + \mathcal O(\eps^2 + \eps\mu + \mu^2).
\end{equation}
In particular, we have  $p(m,t_R)\big|_{(0,0)}=0$, $q(0,0)=0$, and $(\win^R+\la)\big|_{(0,0)}=\gs+\gs=2\gs$, and  first-order derivative information:
\begin{equation}\label{eq:PWQ-derivs}
  \partial_\eps p\big|_{(0,0)} = -\frac{3\Ub^2}{4\gs^2},\quad
  \partial_\mu p\big|_{(0,0)} = \frac{1}{2\gs^2},\quad
  \partial_\eps q\big|_{(0,0)} = -\frac{3\Ub^2}{2\gs},\quad
  \partial_\mu q\big|_{(0,0)} = 0,
\end{equation}
with the abbreviation $\partial p:=\partial\,[\,p(m,t_R)\,]$.
\medskip

\noindent\underline{\emph{Expressions for the derivatives of $D$.}}\ Differentiating~\eqref{eq:D-PWQ} by the product rule (with $p:=p(m,t_R)$),
\begin{subequations}\label{eq:D-partials}
\begin{align}
  \partial_\eps D &= -4\gs^2\bigl(\partial_\eps p\,(\win^R+\la) + p\,\partial_\eps(\win^R+\la) - \partial_\eps q\bigr), \label{eq:D-eps}\\
  \partial_\mu D &= -4\gs^2\bigl(\partial_\mu p\,(\win^R+\la) + p\,\partial_\mu(\win^R+\la) - \partial_\mu q\bigr), \label{eq:D-mu}\\
  \partial_\eps^2 D &= -4\gs^2\bigl(\partial_\eps^2 p\,(\win^R+\la) + 2\,\partial_\eps p\,\partial_\eps(\win^R+\la) + p\,\partial_\eps^2(\win^R+\la) - \partial_\eps^2 q\bigr). \label{eq:D-epseps}
\end{align}
\end{subequations}
Since $p(m,t_R)\big|_{(0,0)}=0$, terms containing  an undifferentiated factor $p$ drop out when~\eqref{eq:D-partials} is evaluated at $(0,0)$. Thus, we can then read off the first order Taylor coefficients of $D(\mu,\eps)$ from the expansions~\eqref{eq:wout-den-R}--\eqref{eq:PWQ-derivs}.
\medskip

\noindent\underline{\emph{$D(0,0)=0$.}}\ By~\eqref{eq:D-PWQ}, $D(0,0)=-4\gs^2\bigl(p(m,t_R)|_{(0,0)}\,(\win^R+\la)|_{(0,0)}-q(0,0)\bigr)=-4\gs^2(0\cdot2\gs-0)=0$.

\medskip
\noindent\underline{\emph{$\partial_\eps D(0,0)=0$.}}\ Evaluating~\eqref{eq:D-eps} at $(0,0)$ and dropping the $p\,\partial_\eps(\win^R+\la)$ term, yields
\begin{equation}
  \partial_\eps D(0,0) = -4\gs^2\bigl(\partial_\eps p\big|_{(0,0)}\,(\win^R+\la)\big|_{(0,0)} - \partial_\eps q\big|_{(0,0)}\bigr)
  = -4\gs^2\Bigl(-\tfrac{3\Ub^2}{4\gs^2}\cdot 2\gs + \tfrac{3\Ub^2}{2\gs}\Bigr) = 0.
\end{equation}

\medskip
\noindent\underline{\emph{$\partial_\mu D(0,0)=-4\gs$.}}\ Evaluating~\eqref{eq:D-mu} at $(0,0)$ and dropping the $p\,\partial_\mu(\win^R+\la)$ term, yields
\begin{equation}
  \partial_\mu D(0,0) = -4\gs^2\bigl(\partial_\mu p\big|_{(0,0)}\,(\win^R+\la)\big|_{(0,0)} - \partial_\mu q\big|_{(0,0)}\bigr)
  = -4\gs^2\Bigl(\tfrac{1}{2\gs^2}\cdot 2\gs - 0\Bigr) = -4\gs.
\end{equation}

Summarizing, 
\begin{equation}\label{eq:D-expanded}
  D(\eps,\mu) \;=\; -\,4\gs\,\mu \;+\; \mathcal O(\eps^2+\eps\mu+\mu^2),
\end{equation}
where
\begin{equation}\label{eq:D-derivatives}
  D(0,0)=0,\qquad \partial_\eps D(0,0)=0,\qquad \partial_\mu D(0,0)=-4\gs\neq 0.
\end{equation}

\medskip
\noindent\underline{\emph{$\tfrac12\partial_\eps^2 D(0,0)=3\Omega_\star$.}}\ We evaluate  $\tfrac12\partial_\eps^2 D(0,0)$ by computing  the order $\eps^2$ Taylor coefficient of $D(\eps,0)$.

For $\mu=0$ one has $\la(0,\eps)=\la_0(\eps)$, hence $m(0,\eps) \equiv1$. So, $p(1,t_R)=3t_R(1+t_R)$ and $q=3\la_0(1-t_R^2)(1+2t_R)$; then~\eqref{eq:Ddef} gives
\begin{equation}\label{eq:D-mu0}
  D(\eps,0) \;=\; -12\gs^2\,(1+t_R)\,\Bigl[\,t_R\bigl(\win^R(\eps,0)+\la_0\bigr) - \la_0\,(1-t_R)(1+2t_R)\,\Bigr].
\end{equation}
Both factors vanish at $\eps=0$ (there $t_R=-1$ and $\win^R=\la_0=\gs$), consistent with $D(\eps,0)=\mathcal O(\eps^2)$. By~\eqref{eq:1ptR}, $1+t_R=\tfrac{\Ub^2}{4\gs^2}\eps+\mathcal O(\eps^2)$; inserting this together with $\la_0(\eps)=\gs-\tfrac{\dot E}{2\gs}\eps+\mathcal O(\eps^2)$ (from~\eqref{eq:la0}) and $\win^R(\eps,0)=\gs-\tfrac{\eps}{\Ub^2}\bigl(3\!\int_{-b}^b U_\star^4\,dx+\dot E(\int_{-b}^b U_\star^2\,dx+\tfrac{1}{2\gs})\bigr)+\mathcal O(\eps^2)$ (from~\eqref{eq:win-expand}), the bracket in~\eqref{eq:D-mu0} satisfies
\begin{equation}\label{eq:bracket-slope}
  \Bigl[\,\cdots\,\Bigr] \;=\; \Biggl(\frac{3\!\int_{-b}^b U_\star^4\,dx + \dot E\bigl(\int_{-b}^b U_\star^2\,dx + \tfrac{1}{2\gs}\bigr)}{\Ub^2} - \frac{\dot E}{2\gs} - \frac{3\Ub^2}{4\gs}\Biggr)\eps + \mathcal O(\eps^2).
\end{equation}
Substituting $\dot E$ from~\eqref{eq:Edot-solved} (TR case),
\begin{equation}
  \dot E = \frac{\tfrac{\Ub^4-1}{2\gs} - 2\!\int_{-b}^b U_\star^4\,dx}{2\!\int_{-b}^b U_\star^2\,dx + \tfrac{1-\Ub^2}{\gs}},
\end{equation}
the coefficient of $\eps$ in~\eqref{eq:bracket-slope} simplifies to $-\Omega_\star/\Ub^2$, with $\Omega_\star=\tfrac{2\Ub^4+1}{4\gs}-2\!\int_{-b}^b U_\star^4\,dx$ as in~\eqref{eq:Omega}. Hence
\begin{equation}\label{eq:D-mu0-expanded}
  D(\eps,0) \;=\; -12\gs^2\cdot\frac{\Ub^2}{4\gs^2}\cdot\Bigl(-\frac{\Omega_\star}{\Ub^2}\Bigr)\eps^2 + \mathcal O(\eps^3) \;=\; 3\,\Omega_\star\,\eps^2 + \mathcal O(\eps^3),
\end{equation}
so $\tfrac12\partial_\eps^2 D(0,0)=3\Omega_\star$.


This proves Proposition~\ref{prop:D-expand} and the proof of   Claim~(1c) is now complete.
\end{proof}
}%

\section{Extensions of Theorem~4.1 of \cite{turner2026resonance}}

\subsection{First order correction to $E(\eps)$; Proof of~\eqref{eq:Edot-solved} in Theorem \ref{thm:bifurcation}}\label{ssec:proof-Edot}

Set $u(x;\eps) := \psi_\eps(x)/\sqrt\eps$ on $[-b,b]$. By Theorem~\ref{thm:bifurcation},
\begin{equation}
  u(x;\eps) \;=\; U_\star(x) + \eps\,\dot u(x) + \mathcal O(\eps^2),
  \qquad
  E(\eps) \;=\; -\gs^2 + \eps\,\dot E + \mathcal O(\eps^2).
\end{equation}
Substituting into $-u'' + Vu - \eps u^3 = E(\eps)\,u$ and matching $\mathcal O(\eps)$ terms,
\begin{equation}\label{eq:Edot-lin}
  (H_V + \gs^2)\,\dot u \;=\; U_\star^3 + \dot E\,U_\star \qquad \text{on } [-b,b].
\end{equation}

Outside $[-b,b]$, $V=0$ and $\psi_\eps$ solves $-\psi_\eps'' - \psi_\eps^3 = E(\eps)\psi_\eps$. Multiplying by $\psi_\eps'$ and integrating to infinity gives
\begin{equation}\label{eq:first-integral}
  (\psi_\eps')^2 + E(\eps)\,\psi_\eps^2 + \tfrac12\,\psi_\eps^4 \;=\; 0 \qquad (|x|>b),
\end{equation}
which extends to $x = \pm b$ by $C^1$ continuity of $\psi_\eps$. Substituting $\psi_\eps = \sqrt\eps\,u$ in~\eqref{eq:first-integral} gives
\begin{equation}
  \eps(u')^2 + E(\eps)\,\eps u^2 + \tfrac12\eps^2 u^4 \;=\; 0.
\end{equation}
Dividing by $\eps$ and differentiating in $\eps$ at $\eps = 0$, with $u(\,\cdot\,;0) = U_\star$, $E(0)=-\gs^2$, and $U_\star'(b)=\gs\Ub$ from~\eqref{eq:TR-BC},
\begin{equation}\label{eq:Edot-BC-b}
  \dot u'(b) - \gs\,\dot u(b) \;=\; -\frac{\dot E\,\Ub}{2\gs} - \frac{\Ub^3}{4\gs}.
\end{equation}
The same computation at $x = -b$ gives
\begin{equation}\label{eq:Edot-BC-a}
  \dot u'(-b) - \gs\,\dot u(-b) \;=\; -\frac{\dot E\,U_\star(-b)}{2\gs} - \frac{U_\star(-b)^3}{4\gs}.
\end{equation}
The TR parametrization $\psi_\eps(-b)^2 = \eps$ gives $u(-b;\eps)^2 \equiv 1$. Differentiating at $\eps=0$,
\begin{equation}
  2\,U_\star(-b)\,\dot u(-b) \;=\; 0,
\end{equation}
so the normalization $U_\star(-b)=1$ from~\eqref{eq:resonance-normalization} forces $\dot u(-b)=0$. Hence~\eqref{eq:Edot-BC-a} reduces to
\begin{equation}\label{eq:Edot-alpha}
  \dot u'(-b) \;=\; -\frac{\dot E\,U_\star(-b)}{2\gs} - \frac{U_\star(-b)^3}{4\gs}.
\end{equation}

Apply Lemma~\ref{prop:boundary-integral} to $f = \dot u$, $g = U_\star^3 + \dot E\,U_\star$ (the right side of~\eqref{eq:Edot-lin}), and $\alpha = \dot u'(-b)$:
\begin{equation}\label{eq:Edot-wronskian}
  \dot u'(b) - \gs\,\dot u(b) \;=\; \frac{1}{\Ub}\Bigl[\dot u'(-b)\,U_\star(-b) - \int_{-b}^b U_\star^4\,dx - \dot E\!\int_{-b}^b U_\star^2\,dx\Bigr].
\end{equation}
By~\eqref{eq:Edot-alpha},
\begin{equation}
  \dot u'(-b)\,U_\star(-b) \;=\; -\frac{\dot E\,U_\star(-b)^2}{2\gs} - \frac{U_\star(-b)^4}{4\gs}.
\end{equation}
Multiplying~\eqref{eq:Edot-BC-b} by $\Ub$ and equating with $\Ub$ times~\eqref{eq:Edot-wronskian}, then collecting $\dot E$,
\begin{equation}
  \dot E\!\left[\int_{-b}^b U_\star^2\,dx - \frac{\Ub^2 - U_\star(-b)^2}{2\gs}\right]
  \;=\; \frac{\Ub^4 - U_\star(-b)^4}{4\gs} - \!\int_{-b}^b U_\star^4\,dx.
\end{equation}
The left bracket is $\mathcal B_\star^{\rm TR}$ from~\eqref{eq:Bstar}; the right side equals $\tfrac12\bigl(\Omega_\star^{\rm TR} - 3 U_\star(-b)^4/(4\gs)\bigr)$ from~\eqref{eq:Omega}. With $U_\star(-b)=1$, this is the TR branch of~\eqref{eq:Edot-solved}.

For SR (Remark~\ref{rem:SR-stability}), the same algebra with $U_\star(0) = 0$ and $\dot u(0) = \dot u'(0) = 0$ yields the SR branch $\dot E = \Omega_\star^{\rm SR}/(2\mathcal B_\star^{\rm SR})$ of~\eqref{eq:Edot-solved}. \qed

\subsection{Soliton centering; Proof of~\eqref{eq:yRL-tanh},~\eqref{eq:yL-tanh} in Theorem \ref{thm:bifurcation}}\label{ssec:proof-centers}

By Theorem~\ref{thm:bifurcation}, on the right exterior $\{x>b\}$,
\begin{equation}
  \psi_\eps(x) \;=\; \Scal\bigl(x - y_R(\eps);\,E(\eps)\bigr), \qquad \Scal(\xi;E) = \sqrt{2|E|}\,\sech\bigl(\la_0(\eps)\,\xi\bigr),
\end{equation}
with $\la_0(\eps) = \sqrt{|E(\eps)|}$ from~\eqref{eq:la0}. Continuity at $x = b$ gives
\begin{equation}\label{eq:center1}
  \sech^2\!\bigl(\la_0(\eps)(b - y_R(\eps))\bigr) \;=\; \frac{\psi_\eps(b)^2}{2\,\la_0(\eps)^2}.
\end{equation}
By Theorem~\ref{thm:bifurcation}(iv), $y_R(\eps) \to +\infty$ as $\eps \to 0^+$, hence $b - y_R(\eps) < 0$ for $\eps$ small, so
\begin{equation}
  \tanh\!\bigl(\la_0(\eps)(b - y_R(\eps))\bigr) \;<\; 0.
\end{equation}
The identity $\sech^2 + \tanh^2 = 1$ and \eqref{eq:center1} then give
\begin{equation}
  \tanh\!\bigl(\la_0(\eps)(b - y_R(\eps))\bigr) \;=\; -\sqrt{1 - \frac{\psi_\eps(b)^2}{2\,\la_0(\eps)^2}},
\end{equation}
which inverts to~\eqref{eq:yRL-tanh}. The same argument at $x = -b$, using the TR parametrization $\psi_\eps(-b)^2 = \eps$, yields~\eqref{eq:yL-tanh}.

For SR, the left tail is absent (Remark~\ref{rem:SR-stability}); only \eqref{eq:yRL-tanh} applies, with the same derivation. \qed

\section{A potential satisfying the conditions of Corollary~\ref{cor:small-gamma-stability}}\label{ssec:cor-nonvacuity}

\bigskip

To construct potentials which satisfy the transmission-resonance smallness condition $\gs < \Gamma_{\rm TR}(M,b)$ of Corollary~\ref{cor:small-gamma-stability}, we observe that:
\begin{itemize}
  \item Any smooth and strictly positive function $U(x)$ on $[-b,b]$, such that $U'(\pm b)=0$, generates a potential $V(x)=U''(x)/U(x)$ for which $E_\star=0$, $U_\star(x)=U(x)$ is a ``threshold'' zero-energy resonance.
  \item Small perturbations of $V$ yield resonance branches with $\gs>0$ arbitrarily small while $\| V \|_{L^1}$ remains bounded below.
\end{itemize}

\begin{prop}[TR threshold-perturbation construction]\label{prop:threshold-construction}
Let $b>0$ and let $U\in C^2([-b,b])$ satisfy $U\ge c>0$ on $[-b,b]$, $U'(\pm b)=0$, $U(-b)=1$, and $U(b)\ne 1$. Set
\begin{equation}\label{eq:V-from-U}
  V(x) := \frac{U''(x)}{U(x)} \text{ on } [-b,b], \qquad V\equiv 0 \text{ outside.}
\end{equation}
Then:
\begin{enumerate}[label=\textup{(\roman*)}, leftmargin=2em]
  \item $(\gs,U_\star)=(0,U)$ is a threshold/zero-energy resonance of $H_V$, with $U$ extended by $U(\pm b)$ outside $[-b,b]$.
  \item Let $W\in L^\infty(\R)$ with $\supp W\subset[-b,b]$ and
  \begin{equation}\label{eq:IW-def}
    \mathcal I(W) \;:=\; \int_{-b}^{b} W\,U^2\,dx \;\ne\; 0;
  \end{equation}
  by replacing $W$ with $-W$ if necessary we may assume the sign convention $\mathcal I(W)>0$. Then there exists $\tau_0>0$ such that the family $V_\tau:=V+\tau W$ supports, for every $\tau\in(0,\tau_0)$, a TR resonance pair $(\gs(\tau),U_\star^{(\tau)})$ with $\gs(\tau)>0$, $U_\star^{(\tau)}>0$ on $[-b,b]$, $\gs(\cdot)$ continuous on $(0,\tau_0)$, and $\gs(\tau)\to 0^+$ as $\tau\to 0^+$.
  \item For $\tau$ in this neighborhood, sufficiently small, $\gs(\tau)<\Gamma_{\rm TR}(2\|V\|_{L^1},b)$, and Corollary~\ref{cor:small-gamma-stability} applies: $\psi_\eps^{(\tau)}$ is orbitally stable for $\eps>0$ small.
\end{enumerate}
\end{prop}

\begin{proof}
\emph{Part (i).} On $[-b,b]$,
\begin{equation}
  H_V U \;=\; -U'' + \frac{U''}{U}\,U \;=\; 0,
\end{equation}
so $U$ solves the resonance equation. The constant extension outside $[-b,b]$ is $C^1$ since $U'(\pm b)=0$, is bounded, and is not in $L^2$. At $\gs=0$ the TR condition~\eqref{eq:TR-BC} reduces to $U'(\pm b)=0$, which holds by hypothesis.

\medskip
\emph{Part (ii).} Let $\phi_{\gs,\tau}$ be the IVP solution of
\begin{equation}
  \phi'' \;=\; (V_\tau+\gs^2)\,\phi,
  \qquad
  \phi(-b)=1,
  \qquad
  \phi'(-b)=\gs,
\end{equation}
which automatically matches the left-side TR condition $\phi'(-b)=\gs\,\phi(-b)$ in~\eqref{eq:TR-BC}. Define the right-side residual
\begin{equation}
  F(\gs,\tau) \;:=\; \phi_{\gs,\tau}'(b) - \gs\,\phi_{\gs,\tau}(b).
\end{equation}
Then $F(\gs,\tau)=0$ iff $V_\tau$ admits a TR resonance with parameter $\gs$ and eigenfunction $\phi_{\gs,\tau}$, and at threshold
\begin{equation}
  F(0,0) \;=\; U'(b) \;=\; 0.
\end{equation}

Let $W_1$ be the second IVP solution of the unperturbed equation $\phi''=V\phi$ with
\begin{equation}
  W_1(-b)=0, \qquad W_1'(-b)=1.
\end{equation}
The Wronskian $W[U,W_1]$ is constant in $x$ and equals $1$ at $x=-b$, so at $x=b$, using $U'(b)=0$,
\begin{equation}
  W_1'(b) \;=\; \frac{1}{U(b)}.
\end{equation}
A variation-of-parameters expansion of $F$ around $(0,0)$ then yields
\begin{equation}\label{eq:F-expansion}
  F(\gs,\tau)
  \;=\; \frac{1-U(b)^2}{U(b)}\,\gs
   \;+\; \frac{\mathcal I(W)}{U(b)}\,\tau
   \;+\; O(\gs^2+\gs\tau+\tau^2).
\end{equation}
Under the hypothesis $U(b)\neq 1$, the $\gs$-coefficient is nonzero, so the implicit function theorem applies to $F=0$ directly and produces a $C^1$ branch
\begin{equation}
  \gs(\tau) \;=\; -\frac{\mathcal I(W)}{1-U(b)^2}\,\tau \;+\; O(\tau^2)
\end{equation}
through $(0,0)$. After replacing $W\mapsto -W$ if necessary so that $\gs(\tau)>0$ on one side of $\tau=0$, we have $\gs(\tau)\to 0^+$ continuously as $\tau\to 0^+$. Setting $U_\star^{(\tau)}:=\phi_{\gs(\tau),\tau}$, the bound $U\ge c>0$ and continuous dependence give $U_\star^{(\tau)}>0$ on $[-b,b]$ for $\tau$ small.

\medskip
\emph{Part (iii).} The bound
\begin{equation}
  \|V_\tau\|_{L^1} \;\le\; \|V\|_{L^1} + |\tau|\,\|W\|_{L^1}
\end{equation}
is uniform on $(0,\tau_0)$. The hypothesis~\eqref{eq:TR-Vside-dichotomy} of
Corollary~\ref{cor:small-gamma-stability} is automatically verified for
$\tau$ small: by continuous dependence, $U_\star^{(\tau)}(b) \to U(b)$ as
$\tau \to 0^+$, so for $\tau$ in a neighborhood of $0$
\begin{itemize}
\item \emph{Case I} ($U(b) < 1$): $U_\star^{(\tau)}(b) < 1$ for $\tau$
  small, so the first alternative of~\eqref{eq:TR-Vside-dichotomy} holds.
\item \emph{Case II} ($U(b) > 1$): $U_\star^{(\tau)}(b)^2 - 1 \ge A_0$
  for $\tau$ small, with $A_0 := (U(b)^2 - 1)/2 > 0$, so the second
  alternative holds.
\end{itemize}
The remaining hypotheses \ref{H:Bstar} and \ref{H:Omega} of
Theorem~\ref{thm:intro-stability-sharp} hold along the family in either case.
Writing $A^{(\tau)} := U_\star^{(\tau)}(b)^2 - 1$ and using the identity
\[
  \Omega_\star^{\rm TR}-\tfrac{3}{4\gs} = \frac{A(A+2)}{2\gs} - 2 I_4,
\]
we obtain:
\begin{itemize}
\item In Case~I, $A^{(\tau)} \le 0$ gives
  $\mathcal B_\star^{\rm TR}(\tau) = I_2 - A^{(\tau)}/(2\gs) \ge I_2 > 0$
  and $\Omega_\star^{\rm TR}(\tau) - \tfrac{3}{4\gs(\tau)}
  \le -2 I_4 < 0$, both bounded from $0$ uniformly in $\gs(\tau)$.
\item In Case~II, $A^{(\tau)} \ge A_0$ gives the divergent leading
  behavior $\mathcal B_\star^{\rm TR}(\tau) \le -A_0/(2\gs(\tau)) + O(1)$
  and $\Omega_\star^{\rm TR}(\tau) - \tfrac{3}{4\gs(\tau)}
  \ge A_0(A_0+2)/(2\gs(\tau)) - O(1)$,
  both nonzero for $\tau$ small.
\end{itemize}
Either way, $\mathcal B_\star^{\rm TR}\neq 0$ and
$\Omega_\star^{\rm TR}-\tfrac{3}{4\gs}\neq 0$ for $\tau$ small.
In the non-degenerate cases, $\Gamma_{\rm TR}$ is bounded below by a
positive constant on $\tau\in(0,\tau_0)$, and since $\gs(\tau)\to 0^+$,
shrinking $\tau_0$ ensures
\begin{equation}
  \gs(\tau) \;<\;
  \begin{cases}
    \Gamma_{\rm TR}\bigl(b, \|V_\tau\|_{L^1}\bigr) & \text{if } U(b) < 1,\\[2pt]
    \Gamma_{\rm TR}\bigl(b, \|V_\tau\|_{L^1}, A_0\bigr) & \text{if } U(b) > 1.
  \end{cases}
\end{equation}
\end{proof}

\begin{rem}[Symmetric wells: $U(b)=1$]\label{rem:Ub-equal-1}
The hypothesis $U(b)\ne 1$ excludes the symmetric case $U(b)=U(-b)=1$ (typical of even potentials $V$). There the $\gs$-coefficient in~\eqref{eq:F-expansion} vanishes, the bifurcation becomes parabolic ($\gs(\tau)\sim\sqrt{|\tau|}$), and the leading $1/\gs$ blow-up of $\mathcal B_\star^{\rm TR}$ and of $\Omega_\star^{\rm TR}-\tfrac{3}{4\gs}$ disappears, so the verifications of~\ref{H:Bstar}, \ref{H:Omega}, and the dichotomy~\eqref{eq:TR-Vside-dichotomy} all require a separate finer analysis. We do not pursue this case here.

This restriction is specific to the TR side: the SR side of Corollary~\ref{cor:small-gamma-stability} carries no analog of the dichotomy~\eqref{eq:TR-Vside-dichotomy}, since the half-line BCs are asymmetric ($U_\star(0)=0$, $U_\star'(b)=\gs\,U_\star(b)$) and the analog of $(1-U(b)^2)$ does not arise. The SR construction in Section~\ref{sssec:half-line-example} therefore proceeds with no such hypothesis on $U(b)$.
\end{rem}

\subsubsection{Full-line (TR) example}\label{sssec:full-line-example}
To exhibit a potential to which Corollary~\ref{cor:small-gamma-stability}
applies in the TR case, we use a deliberately asymmetric threshold
profile $U$ with $U(b) > 1$, so that the perturbed branch enters the
second alternative of~\eqref{eq:TR-Vside-dichotomy} with gap constant
$A_0 := (U(b)^2 - 1)/2 > 0$ inherited from the threshold profile. Fix $b = 1$ and set
\begin{equation}\label{eq:Uflat}
  U_\flat(x) \;:=\; \tfrac{3}{2} - \tfrac{3}{10}\cos\!\bigl(\tfrac{\pi(x+1)}{2}\bigr) - \tfrac{1}{5}\cos(\pi(x+1))
  \quad \text{on } [-1,1],
\end{equation}
extended outside $[-1,1]$ as the constants $U_\flat(\pm 1)$.
A direct check gives $U_\flat \in [1,\,\tfrac{61}{32}] \subset [1, 2]$, $U_\flat(-1) = 1$,
$U_\flat(1) = \tfrac{8}{5} = 1.6$, and $U_\flat'(\pm 1) = 0$
(both sine terms vanish at the endpoints). Hence
Proposition~\ref{prop:threshold-construction} applies with $U = U_\flat$,
and the corresponding potential
\begin{equation}\label{eq:Vflat}
  V_\flat(x) \;:=\; \frac{U_\flat''(x)}{U_\flat(x)}
  \;=\; \frac{(3\pi^2/40)\cos\!\bigl(\tfrac{\pi(x+1)}{2}\bigr) + (\pi^2/5)\cos(\pi(x+1))}{U_\flat(x)}
  \quad\text{on } [-1,1],
\end{equation}
with $V_\flat \equiv 0$ outside, has $\|V_\flat\|_{L^1}\approx 1.89$
and supports the threshold/zero-energy resonance $(\gs,U_\star) = (0, U_\flat)$.
Crucially, the threshold value $U_\flat(b)^2 = 2.56$ satisfies
$|U_\flat(b)^2 - 1| = 1.56 \ne 0$, the \emph{non-degenerate} branch of
Proposition~\ref{prop:threshold-construction}(ii). Figure~\ref{fig:cor-nonvacuity}(a)
shows $V_\flat$ and $U_\flat$ at threshold.

\smallskip\noindent\emph{Perturbation family.} Take $W = -V_\flat$ in
Proposition~\ref{prop:threshold-construction}
(the sign $-$ is chosen so that $\gs(\tau) > 0$ for $\tau > 0$;
the integration-by-parts computation $\mathcal I(V_\flat) = \int V_\flat U_\flat^2 = -\int (U_\flat')^2 < 0$
combined with $1 - U_\flat(b)^2 < 0$ pins the sign of the leading-order slope $-\mathcal I(W)/(1 - U(b)^2)$),
giving the one-parameter family
\begin{equation}
  V_\tau(x) \;:=\; (1-\tau)\,V_\flat(x),
  \qquad \tau \in [0, 0.02],
  \qquad \|V_\tau\|_{L^1} \le \|V_\flat\|_{L^1} \;=:\; M_0 \approx 1.89.
\end{equation}
By Proposition~\ref{prop:threshold-construction}(ii) (non-degenerate
branch), for small $\tau > 0$ this produces a TR pair
$(\gs(\tau), U_\star^{(\tau)})$ with $\gs(\tau) \approx 0.40\,\tau$ and
$U_\star^{(\tau)}(b)^2 \to U_\flat(b)^2 = 2.56$ as $\tau \to 0^+$. Hence
the second alternative of~\eqref{eq:TR-Vside-dichotomy} is satisfied for
$\tau$ small with $A_0 := (2.56 - 1)/2 = 0.78$. Theorem~\ref{thm:bifurcation} then seeds,
for each such pair, a bifurcation branch
$(\psi_\eps^{(\tau)}, E^{(\tau)}(\eps))$ for $\eps > 0$ small.

\smallskip\noindent\emph{Stability via Theorem~\ref{thm:intro-stability-sharp}.}
Computing $(\gs, \Omega_\star^{\rm TR}, \mathcal B_\star^{\rm TR}, \Nslope)$
numerically from~\eqref{eq:Omega},~\eqref{eq:nu} as $\tau$
varies in $(0, \tfrac{1}{2}]$, we find $\Omega_\star^{\rm TR} > 0$ and
$\Nslope < 0$ throughout, both diverging as $\tau \to 0^+$ (see
Figure~\ref{fig:cor-nonvacuity}(b)). Theorem~\ref{thm:intro-stability-sharp}
then gives orbital stability of $\psi_\eps^{(\tau)}$ for small $\eps$
on the whole interval, with the branch in $\Sreg_2^{\rm TR}$ of
Corollary~\ref{cor:stability-TR}.

\smallskip\noindent\emph{An example satisfying the hypotheses of Corollary~\ref{cor:small-gamma-stability}.}
Along this family $\|V_\tau\|_{L^1} \le M_0$ is uniformly bounded and
$\gs(\tau) \to 0^+$, and the second alternative
of~\eqref{eq:TR-Vside-dichotomy} holds with $A_0 = 0.78$, so by
Proposition~\ref{prop:threshold-construction}(iii),
for $\tau > 0$ sufficiently small
\begin{equation}
  \gs(\tau) \;<\; \Gamma_{\rm TR}(1, M_0, 0.78).
\end{equation}
Corollary~\ref{cor:small-gamma-stability} therefore applies and
independently confirms orbital stability of $\psi_\eps^{(\tau)}$ at
small $\eps$, demonstrating its non-vacuity.

\begin{figure}[!htbp]
    \centering
    \setlength{\abovecaptionskip}{4pt}
    \setlength{\belowcaptionskip}{0pt}
    \begin{subfigure}{0.49\linewidth}
        \includegraphics[width=\linewidth]{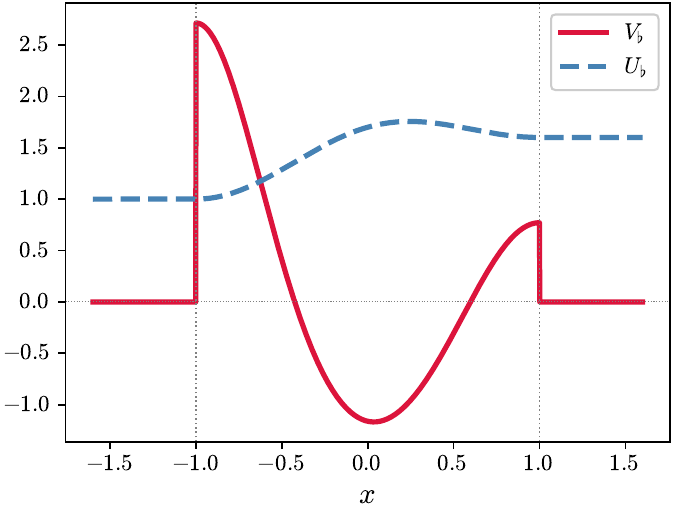}
        \vspace{-1.6em}\subcaption{}
    \end{subfigure}\hfill
    \begin{subfigure}{0.49\linewidth}
        \includegraphics[width=\linewidth]{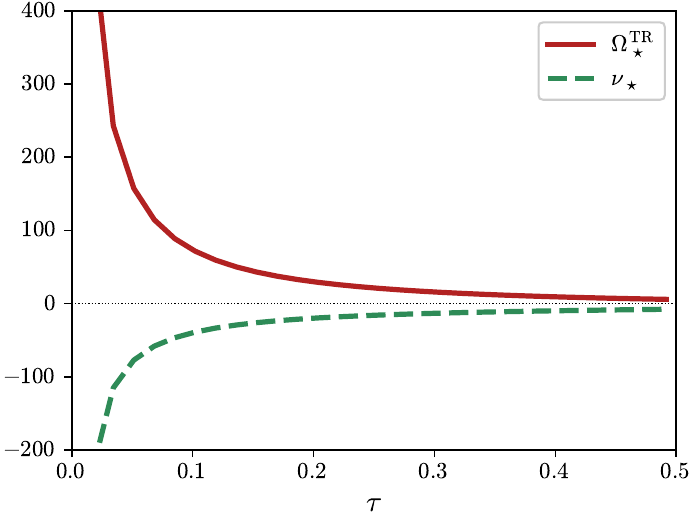}
        \vspace{-1.6em}\subcaption{}
    \end{subfigure}
\vspace{1em}
    \caption{Non-vacuity of Corollary~\ref{cor:small-gamma-stability} via
    the asymmetric perturbation $V_\tau = (1-\tau)V_\flat$ on $[-1,1]$.
    (a)~Threshold potential $V_\flat(x)$ (red) and threshold eigenfunction
    $U_\flat(x)$ (blue dashed), with $U_\flat(-1) = 1$ and $U_\flat(1) = 1.6 \ne 1$
    (non-degenerate). (b)~Stability terms vs $\tau$: both
    $\Omega_\star^{\rm TR}(\tau)$ (red) and $|\Nslope(\tau)|$ (green dashed)
    grow like $1/\gs(\tau) \sim 1/\tau$, placing the branch in $\Sreg_2^{\rm TR}$
    across the entire shown range.}
    \label{fig:cor-nonvacuity}
\end{figure}

\subsubsection{Half-line (SR) example}\label{sssec:half-line-example}
The threshold-perturbation idea of Proposition~\ref{prop:threshold-construction} carries over to the half-line $\R_+$, with the Dirichlet BC $U(0)=0$ replacing $U(-b)=1$ and the Neumann BC $U'(b)=0$ on the right unchanged. (We do not restate the SR analog of the proposition; the construction below verifies all needed properties directly for the concrete example.) Set
\begin{equation}\label{eq:Usharp}
  U_\sharp(x) := \tfrac{2}{\pi}\sin(\pi x/2) \text{ on } [0,1], \qquad U_\sharp(x) := \tfrac{2}{\pi} \text{ on } (1,\infty).
\end{equation}
Then $U_\sharp\ge 0$, $U_\sharp(0)=0$, $U_\sharp(1)=2/\pi$, and $U_\sharp'(1)=0$. The corresponding potential is the square well
\begin{equation}\label{eq:Vsharp}
  V_\sharp(x) := \frac{U_\sharp''(x)}{U_\sharp(x)} = -\frac{\pi^2}{4}\,\chi_{(0,1]}(x), \qquad \|V_\sharp\|_{L^1}=\frac{\pi^2}{4}\approx 2.47.
\end{equation}
The pair $(\gs,U_\star)=(0,U_\sharp)$ is a threshold/zero-energy SR resonance: the BC~\eqref{eq:SR-BC} at $\gs=0$ reduces to $U_\sharp(0)=0$ and $U_\sharp'(1)=0$.

\smallskip\noindent\emph{Perturbation family.} Take
\begin{equation}
  V_\tau^\sharp \;:=\; (1-\tau)\,V_\sharp,
  \qquad \tau\in[0,1),
  \qquad \|V_\tau^\sharp\|_{L^1} \;\le\; \|V_\sharp\|_{L^1} \;=:\; M_0^\sharp.
\end{equation}
Shallowing rather than deepening preserves $\gs(\tau)>0$; deepening would convert the resonance into a true $L^2$ bound state. A half-line analog of the implicit-function-theorem argument from Part~(ii) of Proposition~\ref{prop:threshold-construction} (IVP from the Dirichlet endpoint $x=0$) produces, for each small $\tau>0$, an SR pair $(\gs(\tau),U_\star^{(\tau)})$ with $\gs(\tau)>0$, $U_\star^{(\tau)}\ge 0$, and $\gs(\tau)\to 0^+$ as $\tau\to 0^+$. Theorem~\ref{thm:bifurcation} seeds the bifurcation branch $(\psi_\eps^{(\tau)},E^{(\tau)}(\eps))$ for each such pair.

\smallskip\noindent\emph{Stability via Theorem~\ref{thm:intro-stability-sharp}.} Computing $(\gs,\Omega_\star^{\rm SR},\mathcal B_\star^{\rm SR},\Nslope)$ from~\eqref{eq:Omega},~\eqref{eq:nu} as $\tau$ varies, we find $\Omega_\star^{\rm SR}>\gs(\mathcal B_\star^{\rm SR})^2$ throughout $\tau\in(0,1)$, with $\Omega_\star^{\rm SR}\to\infty$ and $\Nslope\to-\infty$ as $\tau\to 0^+$ (see Figure~\ref{fig:cor-nonvacuity-SR}(b)); in contrast to the full-line case, no $\Omega_\star^{\rm SR}=0$ crossover occurs. Theorem~\ref{thm:intro-stability-sharp} therefore gives orbital stability of $\psi_\eps^{(\tau)}$ at small $\eps$ for every $\tau\in(0,1)$, even up to $\tau$ near $1$ (where the well is shallowed nearly to vacuum): the SR construction is very robust.

\smallskip\noindent\emph{Non-vacuity of Corollary~\ref{cor:small-gamma-stability}.} Along this family $\|V_\tau^\sharp\|_{L^1}\le M_0^\sharp$ is uniformly bounded and $\gs(\tau)\to 0^+$, so the half-line analog of Proposition~\ref{prop:threshold-construction}(iii) gives, for $\tau>0$ sufficiently small,
\begin{equation}
  \gs(\tau) \;<\; \Gamma_{\rm SR}(M_0^\sharp,1),
\end{equation}
and Corollary~\ref{cor:small-gamma-stability} (SR case) applies, independently confirming stability and demonstrating non-vacuity.

\begin{figure}[!htbp]
    \centering
    \setlength{\abovecaptionskip}{4pt}
    \setlength{\belowcaptionskip}{0pt}
    \begin{subfigure}{0.49\linewidth}
        \includegraphics[width=\linewidth]{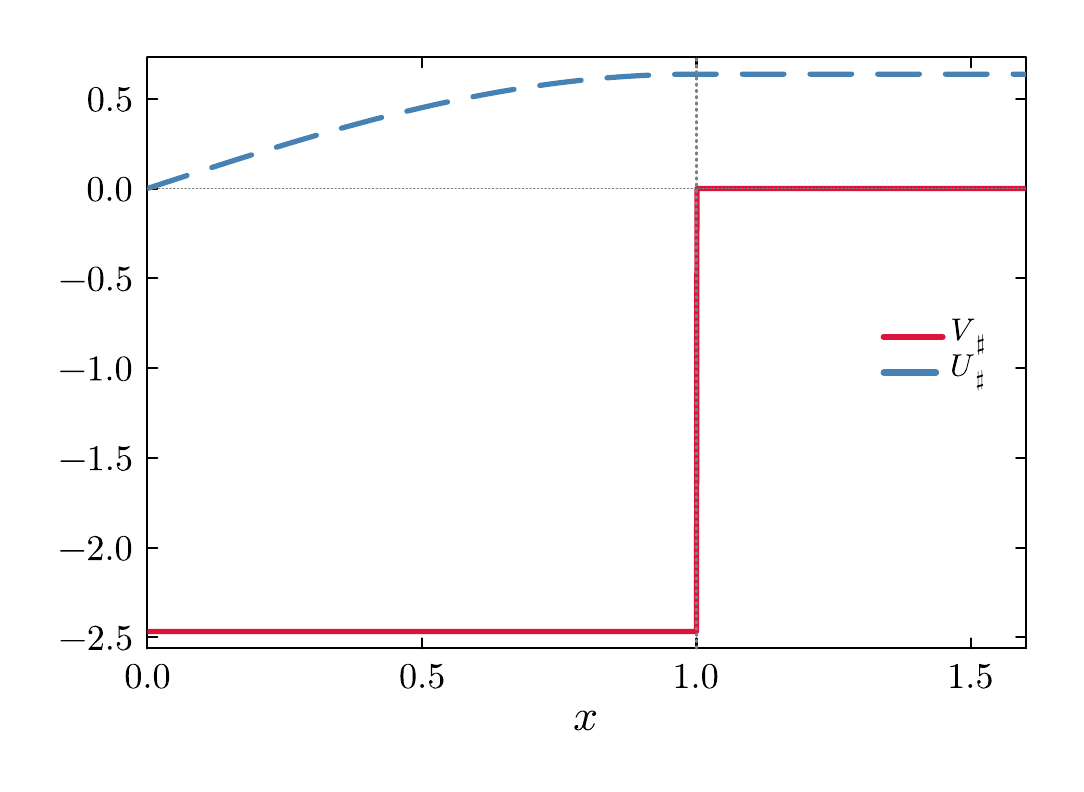}
        \vspace{-1.6em}\subcaption{}
    \end{subfigure}\hfill
    \begin{subfigure}{0.49\linewidth}
        \includegraphics[width=\linewidth]{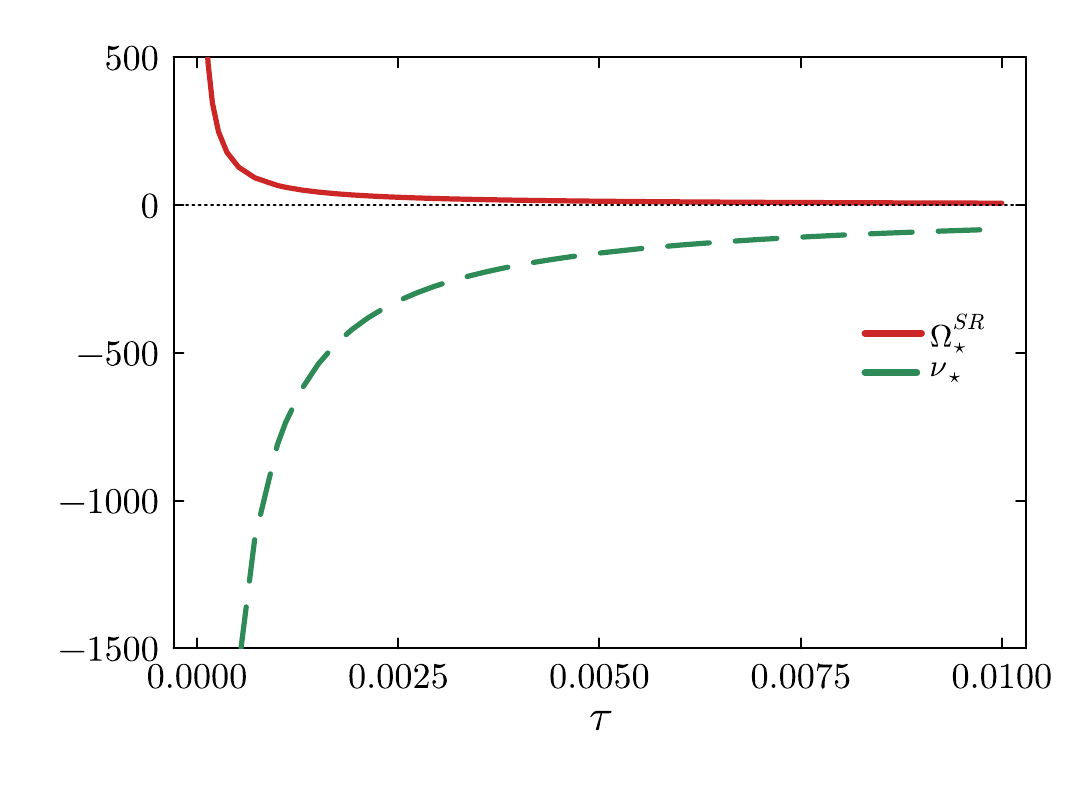}
        \vspace{-1.6em}\subcaption{}
    \end{subfigure}
\vspace{1em}
    \caption{SR analog of Figure~\ref{fig:cor-nonvacuity}: non-vacuity of Corollary~\ref{cor:small-gamma-stability} via the perturbation $V_\tau^\sharp = (1-\tau)V_\sharp$ on $\R_+$. (a)~Threshold potential $V_\sharp(x)=-\tfrac{\pi^2}{4}\chi_{(0,1]}(x)$ (red) and threshold eigenfunction $U_\sharp(x)=(2/\pi)\sin(\pi x/2)$ on $[0,1]$, $\equiv 2/\pi$ on $(1,\infty)$ (blue dashed). (b)~$\Omega_\star^{\rm SR}(\tau)$ (red) and $\Nslope(\tau)$ (green dashed); both remain in $\Sreg^{\rm SR}$ throughout, with $\gs(\tau)\to 0^+$ as $\tau\to 0^+$ (small-$\tau$ extremes clipped by the $y$-axis range).}
    \label{fig:cor-nonvacuity-SR}
\end{figure}

\vspace{4em}
\printbibliography

\end{document}